\documentclass[prx,aps,superscriptaddress,footinbib,twocolumn]{revtex4-2}

\newcommand{\abs}[1]{\vert #1 \vert}

\usepackage{amsthm,amsmath,amssymb}
\usepackage{mathrsfs}
\usepackage{ulem}

\newcommand{\Tr}{\operatorname{Tr}}
\newcommand{\norm}[1]{\Vert #1 \Vert}

\newcommand{\ket}[1]{\vert{ #1 }\rangle}

\newcommand{\ketbra}[2]{\vert #1 \rangle \langle #2 \vert}

\newcommand{\mean}[1]{\langle #1 \rangle}

\usepackage{amsthm}
\newtheorem{theorem}{Theorem}
\newtheorem{proposition}{Proposition}
\newtheorem{lemma}{Lemma}
\newtheorem{corollary}{Corollary}

\theoremstyle{definition}
\newtheorem{definition}{Definition}
\theoremstyle{remark}

\newtheorem{example}{Example}

\usepackage{amssymb}
\usepackage{graphicx}
\usepackage{epstopdf}
\newcommand{\figpath}{./figures}

\usepackage{algorithm}
\usepackage[noend]{algpseudocode}

\usepackage{natbib}

\usepackage[breaklinks=true,colorlinks=true,citecolor=blue,linkcolor=magenta]{hyperref}
\usepackage{xcolor}

\usepackage{comment}
\begin{document}

% \pagenumbering{arabic}
\title{Noise-Agnostic Unbiased Quantum Error Mitigation for Logical Qubits}

\author{Haipeng Xie}
\affiliation{Graduate School of China Academy of Engineering Physics, Beijing 100193, China}

\author{Nobuyuki Yoshioka}
\affiliation{\mbox{International Center for Elementary Particle Physics, University of Tokyo, 7-3-1 Hongo, Bunkyo-ku, Tokyo 113-0033, Japan}}

\author{Kento Tsubouchi}
\affiliation{\mbox{Department of Applied Physics, University of Tokyo, 7-3-1 Hongo, Bunkyo-ku, Tokyo 113-8656, Japan}}

\author{Ying Li}
\email{yli@gscaep.ac.cn}
\affiliation{Graduate School of China Academy of Engineering Physics, Beijing 100193, China}

\begin{abstract}
Probabilistic error cancellation is a quantum error mitigation technique capable of producing unbiased computation results but requires an accurate error model. Constructing this model involves estimating a set of parameters, which, in the worst case, may scale exponentially with the number of qubits. In this paper, we introduce a method called spacetime noise inversion, revealing that unbiased quantum error mitigation can be achieved with just a single accurately measured error parameter and a sampler of Pauli errors. The error sampler can be efficiently implemented in conjunction with quantum error correction. We provide rigorous analyses of bias and cost, showing that the cost of measuring the parameter and sampling errors is low---comparable to the cost of the computation itself. Moreover, our method is robust to the fluctuation of error parameters, a limitation of unbiased quantum error mitigation in practice. These findings highlight the potential of integrating quantum error mitigation with error correction as a promising approach to suppress computational errors in the early fault-tolerant era. 
\end{abstract}

\maketitle

\section{Introduction}
Quantum error correction~\cite{calderbank_good_1996,fowler_surface_2012,breuckmann_quantum_2021} and error mitigation~\cite{temme_error_2017,li_efficient_2017,cai_quantum_2023} are two key strategies for reducing errors in quantum computing. Unlike error correction, which requires additional qubits to encode logical information, error mitigation suppresses the impact of errors usually without increasing qubit overhead. As a result, error mitigation is particularly useful in scenarios where error correction is not feasible. Moreover, in fault-tolerant quantum computing, error mitigation can extend the achievable circuit depth while maintaining the same qubit cost and reliability~\cite{suzuki_quantum_2022,zhang_demonstrating_2025,wahl2023zne,yoshioka_error_2025,tsubouchi_symmetric_2025,liu_quantum_2025,piveteau_error_2021,lostaglio_error_2021,zimborás2025mythsquantumcomputationfault}. 

Among quantum error mitigation techniques, probabilistic error cancellation (PEC)~\cite{temme_error_2017,endo_practical_2018,van_den_berg_probabilistic_2023} stands out for its ability to achieve bias-free computation. This property is particularly valuable for applications requiring precise confidence level estimation. However, PEC critically depends on an accurate characterization of the underlying noise, typically obtained through benchmarking methods such as gate set tomography~\cite{merkel_self-consistent_2013,nielsen_gate_2021,endo_practical_2018} or sparse Pauli-Lindblad learning~\cite{van_den_berg_probabilistic_2023,kim_evidence_2023}. Due to the limitations of these techniques, PEC is most effective for sparse error models, where errors are uncorrelated or involve only few-qubit correlations~\cite{guo_quantum_2022}. In contrast, general error models with many-qubit correlations require an exponentially large number of parameters to characterize, making accurate benchmarking impractical~\cite{torlai_quantum_2023,gebhart_learning_2023}. Notably, even for sparse models, benchmarking is usually considered costly~\cite{kim_evidence_2023}. In this work, we introduce an error mitigation approach that also achieves bias-free computation, but does not require an accurate error model. 

We develop a variant of PEC combined with quantum error correction. Specifically, we introduce spacetime noise inversion (SNI), a method that mitigates errors collectively across an entire quantum circuit. Unlike conventional PEC, which requires an accurate error model, SNI relies on two ingredients: a sampler of errors and a single accurate parameter—the total error rate of the circuit. We propose a protocol that leverages quantum error correction techniques to realize the error sampler and estimate the total error rate. Our protocol is supported by rigorous bounds on both the bias and cost. Notably, an error analysis of PEC that accounts for the benchmarking stage is essential for estimating the confidence level of the results, and such analysis has previously been performed only for sparse error models~\cite{van_den_berg_probabilistic_2023}. This theoretical result shows that when the total error rate is moderate, the costs of estimating the total error rate, generating error samples, and performing the computation are comparable. Furthermore, we demonstrate that our method is robust against temporal fluctuations in error rates~\cite{burnett_decoherence_2019,klimov_fluctuations_2018,schlor_correlating_2019}, approaching unbiased results up to assumptions about the fluctuation timescale. Our approach introduces a new paradigm of unbiased quantum error mitigation in conjunction with error correction, minimizing the number of parameters measured in noise benchmarking, extending applicability to scenarios with prevalent many-qubit errors [e.g.,~logical qubits in constant-rate quantum low-density parity check (qLDPC) codes~\cite{gottesman_fault-tolerant_2014,breuckmann_quantum_2021}], and accommodating cases where error parameters are unstable. 

\section{Noise map and conventional PEC}

For each operation in quantum computing, the associated noise can be described by a trace-preserving completely positive map $\mathcal{N}$. For example, a noisy gate is represented by a product of two maps $\mathcal{N}[U]$, where $[U]\bullet = U\bullet U^\dagger$ is the map denoting the error-free gate. In conventional PEC, errors are mitigated by applying the inverse of the noise map using a Monte Carlo method, resulting in the effective operation $\hat{\mathcal{N}}^{-1} \mathcal{N}[U]$~\cite{temme_error_2017,endo_practical_2018}. Here, $\hat{\mathcal{N}}$ is the error model for the operation. If the error model perfectly matches the actual noise, i.e., $\hat{\mathcal{N}} = \mathcal{N}$, we can effectively implement the ideal gate, achieving unbiased quantum computing. However, even when considering a Pauli error model, $\hat{\mathcal{N}}$ may have $4^n - 1$ parameters to determine for an $n$-qubit system. Consequently, PEC is typically considered a method reliant on a sparse error model. 

\section{Spacetime noise inversion}
We deal with the entire noise in a quantum circuit as a whole. For now, we assume the circuit consists of unitary gates and that the noise associated with each gate is Pauli. Later, we will show that this approach extends to non-Pauli noise and randomized dynamic circuits, i.e., circuits involving mid-circuit measurements, feedback operations, and randomized compiling. Let $\rho_i$ be the initial state. Suppose that the final state of the circuit is given by $\rho_f = \mathcal{N}_N[U_N] \cdots \mathcal{N}_2[U_2] \mathcal{N}_1[U_1] \rho_i$, where $U_j$ is the $j$th gate, and $\mathcal{N}_j$ represents the corresponding noise. The entire noise in the circuit can be expressed as a spacetime noise map: $\mathcal{N}_{st} = \mathcal{N}_1 \otimes \mathcal{N}_2 \otimes \cdots \otimes \mathcal{N}_N$. 

\begin{figure}[htbp]
\centering
\includegraphics[width=\linewidth]{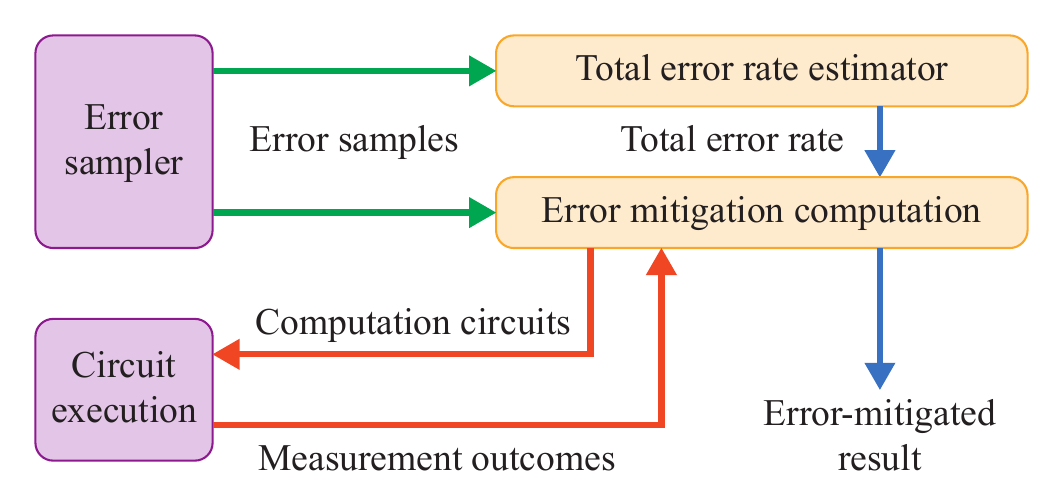}
\caption{
Workflow of spacetime noise inversion. The error sampler generates spacetime Pauli errors, which are used to estimate the total error rate and modify the quantum circuit: Pauli gates are inserted into the circuit according to the generated errors. The final error-mitigated result is obtained as the expected value over measurement outcomes from the modified circuits, incorporating necessary normalization and phase factors from the Monte Carlo summation. For detailed protocols and pseudocodes, see Appendices \ref{app:ideal_protocol}, \ref{app:practical_protocol}, and \ref{app:codes}. 
}
\label{fig:scheme}
\end{figure}

To mitigate errors, we construct the inverse of $\mathcal{N}_{st}$. Since each noise map is Pauli, the spacetime noise takes the form $\mathcal{N}_{st} = \sum_{\sigma\in \mathbb{P}_n^{\otimes N}} \epsilon(\sigma)[\sigma]$, where $\mathbb{P}_n$ is the set of $n$-qubit Pauli operators, and $\epsilon(\sigma)$ is the rate of the error $[\sigma]$ (the error is called trivial when $\sigma$ is identity). We can rewrite it as $\mathcal{N}_{st} = (1 - P) [\openone^{\otimes N}] + P \mathcal{E}$, where $\mathcal{E}= \frac{1}{P}\sum_{\sigma\in \mathbb{P}_n^{\otimes N}-\{\openone^{\otimes N}\}} \epsilon(\sigma)[\sigma]$ is a map representing the erroneous component, and $P = \sum_{\sigma\in \mathbb{P}_n^{\otimes N}-\{\openone^{\otimes N}\}} \epsilon(\sigma)$ is the total error rate. Using a Taylor expansion, the inverse map can be written as 
\begin{equation}
\mathcal{N}_{st}^{-1} = \sum_{k=0}^\infty \frac{(-1)^k P^k}{(1-P)^{k+1}} \mathcal{E}^k
\end{equation}
when $P<1/2$. We can realize the inverse map by using the Monte Carlo method to simulate the above summation formula: We sample $k$ according to $\mathrm{Pro}(k) \propto P^k/(1-P)^{k+1}$ and generate Pauli errors from the distribution $\mathcal{E}^k$, then we modify the circuit according to the generated errors and take an average over random samples; see Fig.~\ref{fig:scheme}. We assume access to a sampler that generates Pauli errors according to the distribution $\mathcal{N}_{st}$; later, we will give a practical protocol of the error sampler without knowing the error model. By post-selecting only non-trivial errors, we can effectively generate errors from the distribution $\mathcal{E}^k$: post-select $k$ non-trivial errors and take the product. In this way, we only need one accurate parameter $P$ to achieve unbiased error mitigation. Note that the parameter $P$ can be estimated from the errors observed in the sampler. 

\section{Practical error sampler}
We propose a method to sample Pauli errors associated with an operation by preparing a suitable initial state. When the operation is applied, Pauli errors transform the state into distinct, mutually orthogonal final states. By measuring these final states, we can identify the errors. For gate operations, the error sampler circuit is adapted from quantum process tomography using a Bell state~\cite{dariano_imprinting_2003,dariano_tomography_2001,leung_chois_2003,dur_nonlocal_2001,PhysRevA.105.032435,seif2024entanglement} and is integrated with quantum error correction; see Fig.~\ref{fig:protocol}(a). A similar approach can be applied to sample errors in state preparation and measurement operations; see Appendix~\ref{app:practical_protocol}. To sample a spacetime error, we collect errors from individual operations, which together constitute the spacetime error. 

\begin{figure}[htbp]
\centering
\includegraphics[width=\linewidth]{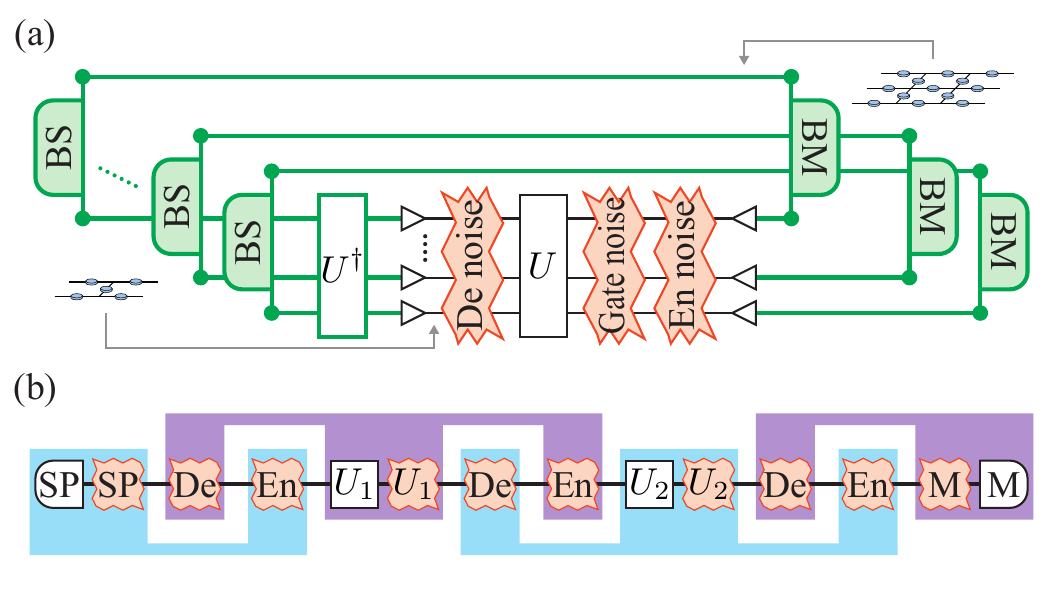}
\caption{
(a) Error sampler circuit for a gate. An ancilla qubit is introduced for each qubit the gate $U$ acts on. First, each qubit pair is initialized in the Bell state (BS) $(\ket{00}+\ket{11})/\sqrt{2}$. Then, the inverse gate $U^\dag$ and the gate $U$ are applied sequentially. Finally, the Bell measurement (BM) is applied to each qubit pair, i.e. measuring $XX$ and $ZZ$ operators. Triangles represent encoding (En) and decoding (De) operations, which transfer states between low-distance logical qubits and high-distance super qubits, which are represented by thin black and thick green lines, respectively. 
(b) Noise-boosted computation circuit. Smooth-boundary squares represent ideal operations, while zigzagged squares denote noise maps (which describe the noises associated with corresponding operations). Encoding and decoding errors are stochastically inserted after state preparation (SP) operations and unitary gates ($U_1$ and $U_2$) but not after measurement (M) operations. Each operation, together with its inherent noise and the inserted encoding/decoding noise, constitutes the effective noisy operation, as indicated by the blue and purple boxes. Note that Pauli twirling is applied to the encoding and decoding operations, such that the associated errors are effectively Pauli, i.e., commutative. 
}
\label{fig:protocol}
\end{figure}

To achieve unbiased error mitigation using the practical error sampler, we address two critical problems and introduce necessary refinements to our protocol. First, errors may not be strictly Pauli. Second, even if the errors are Pauli, the practical error sampler deviates from the target error distribution due to additional errors in an error sampler circuit, which are caused by operations other than the one being benchmarked. Both of them can cause bias. 

We use twirling operations~\cite{knill2004faulttolerantpostselectedquantumcomputation,geller_efficient_2013, PhysRevA.94.052325} to address the first problem and convert general errors into Pauli errors. As an example, we focus on a universal set of operations consisting of two types: (i) Stabilizer operations, which include Clifford gates and state preparations/measurements in the Pauli basis; and (ii) Non-stabilizer operations, which include gates from the third level of the Clifford hierarchy~\cite{gottesman_demonstrating_1999} and state preparations/measurements in the basis of Hermitian Clifford operators. We assume that, except for two-qubit Clifford gates, all other operations are single-qubit. Pauli gates on logical qubits typically exhibit negligible noise~\cite{piveteau_error_2021, suzuki_quantum_2022}, allowing them to serve as ideal tools for twirling stabilizer operations. However, twirling non-stabilizer operations requires the use of non-Pauli stabilizer operations~\cite{yoshioka_error_2025}, which may themselves be noisy. As a result, the twirled noise in non-stabilizer operations remains not strictly Pauli. To achieve perfect twirling, it is therefore necessary to mitigate errors in the non-Pauli stabilizer operations used for twirling, making them effectively error-free---that is, to mitigate errors in a randomized circuit subject to indefinite spacetime noise. 

To ensure unbiasedness, we must resolve the issue of indefinite spacetime noise. To this end, we develop a general formalism for characterizing noise in randomized dynamic circuits, termed maximum spacetime noise. This framework is presented in Appendix~\ref{app:NRDC}. When combined with twirling, it enables the reduction of arbitrary noise to effective Pauli noise, which can then be completely removed using SNI. 

We use error correction to address the second problem and eliminate the additional errors in error sampler circuits. While this approach is general, we illustrate it with the following example: Qubits are encoded into a surface code with a moderate distance $d$, leaving some residual logical errors that need to be mitigated [see Fig.~\ref{fig:protocol}(a)]. In an error sampler circuit, qubits are initially encoded with a larger code distance $d_S$, ensuring that logical errors are negligible; we refer to these as super qubits. Just before applying the operation to be benchmarked [gate $U$ in Fig.~\ref{fig:protocol}(a)], we reduce the code distance to $d$ (decoding). Then, we apply the operation and increase the code distance again (encoding). This procedure ensures that the dominant source of error originates from the operation being benchmarked, while errors in the remaining computational operations are negligible. Note that encoding and decoding errors still persist and will be addressed separately. 

For the surface code, the code distance can be increased and decreased using lattice surgery operations~\cite{horsman_surface_2012,Vuillot_2019,litinski_game_2019}. For general error correction codes, various techniques can control the distance. One example is code concatenation~\cite{knill_concatenated_1996,yamasaki_time-efficient_2024}: We can encode super qubits on top of low-distance logical qubits using another code, effectively increasing the code distance. Additionally, super qubits can also be realized by modifying the decoding algorithm without physically increasing the code distance~\cite{rodriguez2024experimentaldemonstrationlogicalmagic}. This can be done by post-selecting events based on error syndromes to suppress the logical error rate, at the cost of increased time overhead. In the extreme case, post-selecting only events with no detected syndromes can effectively double the code distance. 

Error sampler can be performed without additional physical qubits although enlarging the code distance usually increases the encoding overhead. Note that error sampler circuits can operate independently. Consider surface codes as an example. Suppose the quantum processor has $n$ low-distance logical qubits. Error sampler circuits can be executed in parallel on approximately $d^2n/(2d_S^2)$ qubits, leading to a time overhead of about $2d_S^2/d^2$. An additional overhead factor of $O(d_S/d)$ may arise due to the increased time required for logical operations on super qubits. To sample a single instance of the spacetime error, the error sampler circuit must be executed for each operation in the computation circuit. Consequently, the time cost of one full run of spacetime error sampler is comparable to a single run of the computation circuit, with an overhead factor of $O(\mathrm{Poly}(d_S/d))$ for surface codes. 

The above analysis extends to qLDPC codes, where each code block can encode multiple logical qubits. Approaches to operating these logical qubits include the concatenation with other codes and lattice surgery, enabling fault-tolerant quantum computation with a constant qubit overhead~\cite{gottesman_fault-tolerant_2014,tamiya_polylog-time-_2024,nguyen_quantum_2024,Cohen_2022,zhang_time-efficient_2025,cowtan2025parallellogicalmeasurementsquantum}. Both approaches are compatible with the realization of super qubits. Similar to surface codes, error sampler circuits can be implemented on a subset of blocks in parallel, introducing only a polynomial time overhead in terms of code distance ratio. Notably, logical errors within each qLDPC code block may exhibit strong correlations, simultaneously affecting multiple logical qubits~\cite{zhang_demonstrating_2025}. In such cases, a code block encoding $k$ logical qubits could introduce $O(4^k)$ parameters to the error model, making conventional PEC, which relies on a sparse error model, ineffective. In contrast, SNI remains robust in such scenarios. 

See Appendices~\ref{app:applications_SC} and \ref{app:applications_qLDPC} for detailed discussions on applications to surface codes and general qLDPC codes. 

Lastly, encoding and decoding operations may also introduce additional errors. To handle this, we can align the error distributions in the error sampler and computation circuit by boosting the noise in the computation circuit. Specifically, we sample the encoding and decoding errors by removing the gates $U$ and $U^\dag$ from the error sampler circuit in Fig.~\ref{fig:protocol}(a), and then insert the observed errors into the computation circuit, as illustrated in Fig.~\ref{fig:protocol}(b). In this way, the error sampler is effectively perfect. 

Alternatively, the impact of encoding and decoding errors can be reduced using zero-noise extrapolation techniques such as gate folding~\cite{digital_giurgica-tiron_2020}, taking advantage of the fact that ideal encoding and decoding operations are identities and thus well-suited for folding. Though practical, gate folding may retain some residual bias. 

\section{Performance and theoretical analysis}
Consider computing the expected value of an observable $A$. The following results characterize the residual error after mitigation and the corresponding sampling overhead. 

\begin{theorem}
Apply SNI with an exact error sampler to an arbitrary randomized dynamic circuit. Suppose errors are temporally uncorrelated; suppose Pauli gates are error-free and errors in all other operations are Pauli. Let $P$ denote the total error rate of the maximum spacetime noise, which is smaller than $1/2$, and let $\hat{P}$ be its estimate. 
Define $M$ as the number of circuit runs used to evaluate $\hat{A}_{QEM}$, and $M_P$ as the number of maximum spacetime error instances used to estimate the total error rate. For the error-mitigated estimator, the bias has the upper bound $\norm{a}_{L^\infty} \left\vert\frac{1}{1-2\hat{P}} - \frac{1}{1-2P}\right\vert$, where $\norm{a}_{L^\infty}$ is the maximum absolute value of the observable across all possible measurement outcomes. For any positive numbers $\delta$ and $f$, the error $\vert\hat{A}_{QEM}-\mean{A}_I\vert$ is smaller than $\delta\norm{a}_{L^\infty}$ with a probability at least $1-f$ under conditions $M_P\geq \frac{1}{2t_P^2}\ln\frac{4}{f}$ and $M\geq \frac{8}{\delta^2(1-2P-2t_P)^2}\ln\frac{4}{f}$, where $t_P = \min\left\{\frac{\delta(1-2P)^2}{4+2\delta(1-2P)},\frac{1}{2}-P\right\}$. Furthermore, let $M_{es}$ be the total number of spacetime error instances generated from the error sampler. The expected value and variance of the sampling cost $M_{es}$ are given by $M_P + \frac{M\hat{P}}{P(1-2\hat{P})}$ and $\frac{M\hat{P}(2-P-3\hat{P}+2P\hat{P})}{P^2(1-2\hat{P})^2}$, respectively. 
\label{the:error_cost}
\end{theorem}

In Appendices \ref{app:error_cost_ideal}, \ref{app:error_cost_practical}, and \ref{app:correlations}, we provide the proof and extend the results to address the practical error sampler and non-Pauli errors (Corollary~\ref{cor:error_cost_practical}). The error bounds and cost estimators presented above hold rigorously when using the practical error sampler, provided that errors on super qubits are negligible. These results remain valid for non-Pauli errors, assuming an operation set suitable for twirling. Moreover, the conclusions extend to certain forms of temporally correlated noise (Corollary~\ref{cor:correlation}), an example of which will be discussed later. 

The bias upper bound in the theorem confirms that SNI is unbiased when the total error rate is estimated accurately, i.e., $\hat{P} = P$. Regarding the sampling cost, consider a small permitted error $\delta$, which results in a large number of circuit runs $M$. In this regime, the standard deviation $\sqrt{\mathrm{Var}(M_{es}\vert\hat{P})} = \Theta(\sqrt{M})$ is significantly smaller than the expected value $\mathrm{E}[M_{es}\vert\hat{P}] = \Theta(M)$. Consequently, the expected value $\mathrm{E}[M_{es}\vert\hat{P}] \simeq M_P + \frac{M}{1-2P}$ serves as a reliable estimate. Since $\delta$ is small, setting $M_P \simeq \frac{8}{\delta^2(1-2P)^4}\ln\frac{4}{f}$ and $M \simeq \frac{8}{\delta^2(1-2P)^2}\ln\frac{4}{f}$ is sufficient. Then, the relative overhead due to error sampler circuits is $M_{es}/M \simeq \frac{1}{(1-2P)^2} + \frac{1}{1-2P}$, which is six when $P = 1/4$ and approaches two when $P$ is small. Notably, this cost analysis is valid for non-sparse error models that involve cross-qubit correlations. 

According to the theorem, SNI works as long as $P < 1/2$. 
This constraint arises from the normalization factor $\gamma = 1/(1 - 2P)$ introduced in the Monte Carlo summation, which governs the sampling cost in error-mitigated computation. Specifically, to achieve the same variance in the results as unmitigated computation, SNI requires $\gamma^2$ times more samples. As $P$ approaches $1/2$, this cost becomes divergent. In contrast, conventional PEC does not suffer from this limitation; its corresponding factor is roughly $\gamma \approx e^{2Np}$, where $N$ is the number of gates and $p$ is the per-gate error rate. Even when $Np$ exceeds one, the cost remains finite. This fundamental difference stems from the Taylor expansion used in noise inversion. It is possible that alternative expansion formulas could yield a smaller normalization factor $\gamma$, thereby reducing the associated cost. 
If $P \geq 1/2$, SNI can be adapted by decomposing $\mathcal{N}_{max}$ into a product of multiple noise maps, each with an error rate below $1/2$. Each noise map can then be mitigated independently using SNI, though this approach requires measuring multiple error rates. Additionally, when $P$ is small, $\mathrm{Var}(M_{es}\vert\hat{P})$ becomes large. This issue can be resolved with a minor modification:
Intentionally boost $P$ by regarding the identity map as trivially acting noise and incorporate it into $\mathcal{E}$. 
For further details, see Appendix~\ref{app:modifications}. 

\begin{figure}[htbp]
\centering
\includegraphics[width=\linewidth]{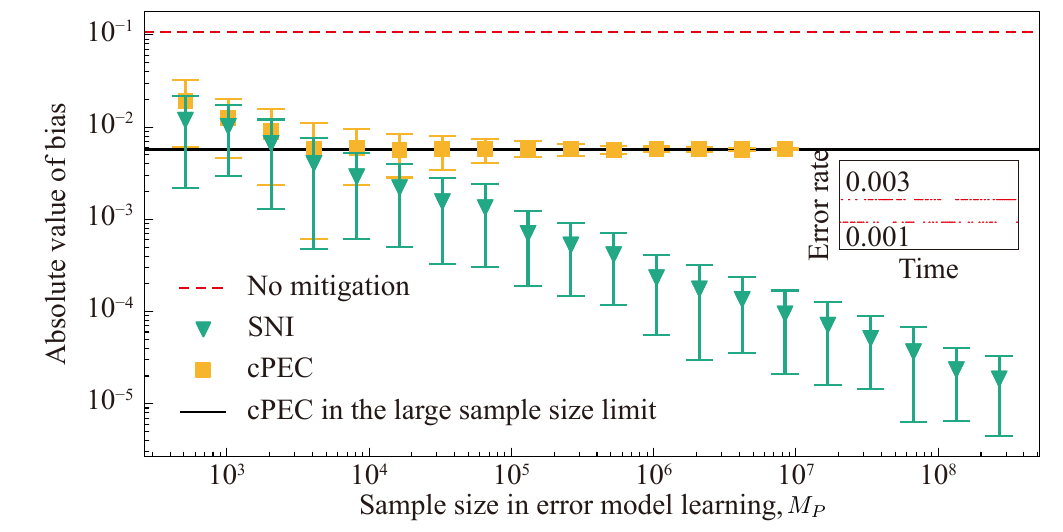}
\caption{
Bias in error mitigation under error parameter fluctuations. In the numerical simulation, we evaluate the bias of a quantum circuit implementing the transformation $[e^{-i\frac{\pi}{8}(X_1+X_2)}e^{-i\frac{\pi}{8}Z_1Z_2}]^8$ on two qubits initialized in the $\ket{+}$ state, with the observable being $X$ on the first qubit. The transformation is decomposed into controlled-NOT, Hadamard, and $T$ gates. 
The simulation is performed at the logical-qubit level, assuming a model of logical errors: 
Each non-Pauli operation is subject to a noise map, where all noise maps are parameterized by an error rate $p$ that is randomly drawn from $\{0.001,0.003\}$ with equal probabilities for each individual run of the computation circuit and spacetime error generation. 
The noise includes both Pauli and coherent errors. SNI with a practical error sampler and the conventional PEC (cPEC) are applied for error mitigation. In both protocols, $M_P$ instances of the spacetime error are used to estimate error parameters. In cPEC, a sparse error model is assumed without accounting for temporal correlations. Error bars represent standard deviations, each estimated from 100 instances. See Appendix~\ref{app:time} for details. 
}
\label{fig:time_main}
\end{figure}

\section{Robustness to unstable error parameters}
When error parameters fluctuate over time, the effective spacetime noise exhibits temporal correlations. In the case of Pauli noise, the error parameters are denoted as $\mathbf{p}$, representing the rates of Pauli errors, and follow a distribution characterized by the probability density function $g(\mathbf{p})$. We assume that the time scale over which $\mathbf{p}$ varies is much longer than the time required for a single circuit run. Under this assumption, the effective spacetime noise is given by $\mathcal{N}_{ave} = \int d\mathbf{p} g(\mathbf{p}) \mathcal{N}_{max}(\mathbf{p})$, which is no longer a product of independent noise maps for individual operations, implying that errors are correlated in time. Such correlations cause bias in conventional PEC, which can be eased by real-time adaptation and Bayesian inference~\cite{dasgupta2023adaptivemitigationtimevaryingquantum,daguerre2025realtimeadaptationquantumnoise}. 
We also remark that a purification-based QEM method for eigenstate calculation can also counteract error fluctuations~\cite{yoshioka2022generalized}, while it is not completely unbiased. 
Nevertheless, SNI remains effective in the same manner as when error rates are constant. 

In a numerical simulation, we illustrate the performance of SNI in mitigating errors with a fluctuating parameter; see Fig.~\ref{fig:time_main}. Unlike non-sparse error models with cross-qubit correlations, where the rigorous analysis covers non-Pauli errors, the unbiasedness for temporally varying non-Pauli noise is not guaranteed theoretically. Therefore, we include non-Pauli errors in our numerical simulations to assess their impact. The results demonstrate that SNI effectively mitigates such errors, with the bias decreasing as $M_P$ increases, approaching an unbiased result. Moreover, the observed $1/\sqrt{M_P}$ scaling is consistent with the prediction of Theorem~\ref{the:error_cost}, supporting the conclusion that the result becomes unbiased in the large-$M_P$ limit. In contrast, conventional PEC converges to a finite bias, beyond which increasing $M_P$ provides no further improvement. 

\section{Conclusions}
We have presented an error mitigation protocol that relies on a single accurately estimated noise parameter, replacing detailed error-model characterization with direct error sampling. By integrating quantum error correction and other fault-tolerant techniques, our approach removes imperfections in the sampling circuits and enables unbiased mitigation. As a consequence, it naturally supports mitigating correlated noise, including those present in high-encoding-rate qLDPC codes and in temporally varying noise processes. The protocol is primarily intended for logical-qubit computations and is specifically designed to avoid the need to characterize extremely small logical error rates or complex, highly correlated logical noise models. Although it may, in principle, also be applied at the physical-qubit level by treating physical qubits as low-distance logical qubits, doing so requires the capability to encode and manipulate at least four logical qubits (used as super-qubits) with fault-tolerant logical operations, which in turn demands operation below threshold and a sufficiently large physical-qubit budget. Overall, our results demonstrate the practical potential of unifying error correction and error mitigation, offering a promising route toward resource-efficient quantum computation in the early fault-tolerant era.
%We have presented an error mitigation protocol that relies on a single accurately measured parameter, streamlining noise characterization by replacing a detailed error model with an error sampler. By leveraging quantum error correction and some other techniques, we eliminate imperfections in the error sampler circuits, enabling unbiased error mitigation. In contrast to the conventional approach, our protocol is inherently suited to handling correlated errors, making it applicable to scenarios such as qLDPC codes and temporally varying noise. These advantages arise from the integration of error correction into the mitigation framework. Our results highlight the practical potential of unifying error correction and mitigation, offering a promising path forward in the early fault-tolerant era of quantum computing. 

\begin{acknowledgments}
This work is supported by the National Natural Science Foundation of China (Grant Nos. 12225507, 12088101) and NSAF (Grant No. U1930403). N.Y. is supported by JST Grant Number JPMJPF2221, JST CREST Grant Number JPMJCR23I4, IBM Quantum, JST ASPIRE Grant Number JPMJAP2316, JST ERATO Grant Number JPMJER2302, and Institute of AI and Beyond of the University of Tokyo.
K.T. is supported by the Program for Leading Graduate Schools (MERIT-WINGS) and JST BOOST Grant Number JPMJBS2418. The source codes for the numerical simulation are available at \cite{Code_Xie}.
\end{acknowledgments}

\appendix

\begin{widetext}

\addtocounter{theorem}{1}

Notations are introduced in Sec.~\ref{app:notations}, including definitions of operation sets, circuit types, and noise models. A detailed discussion of maximum spacetime noise is presented in Sec.~\ref{app:NRDC}. Concepts relevant to observable evaluation in quantum computing are outlined in Sec.~\ref{app:observable}. Complete protocols based on ideal and practical error samplers are described in Secs.~\ref{app:ideal_protocol}~and~\ref{app:practical_protocol}, respectively, with accompanying pseudocodes provided in Sec.~\ref{app:codes}. Error and cost analyses for both protocols appear in Secs.~\ref{app:error_cost_ideal}~and~\ref{app:error_cost_practical}: Theorem~1 is proved in Sec.~\ref{app:error_cost_ideal}, while rigorous results for the practical-sampler protocol are given in Sec.~\ref{app:error_cost_practical}. Theoretical analysis of mitigating temporally correlated errors is provided in Sec.~\ref{app:correlations}. Impact of Pauli-gate and super-qubit errors is provided in Sec.~\ref{app:pauli_SQ_errors}. In Sec.~\ref{app:modifications}, we also present modifications to the protocols for handling cases where the total error rate is either too large or too small. The protocol and cost analysis of benchmarking of surface code operations is shown in Sec.~\ref{app:benchmarking_SC}. Details of the numerical simulations can be found in Sec.~\ref{app:numerics}. We also compare SNI to other quantum error mitigation methods in Sec.~\ref{app:comparisons}. Finally, applications of SNI to surface codes and qLDPC codes are presented in Sec.~\ref{app:applications_SC} and \ref{app:applications_qLDPC}, respectively.

\section{Notations}
\label{app:notations}

\subsection{Qubits and operations}

We consider a quantum computer of $n$ logical qubits. The label set of logical qubits is $Q = \{1,2,\ldots,n\}$. The operation set on logical qubits is $\mathbb{O}$. The Pauli operator set of $n$ logical qubits is $\mathbb{P}_n$. 

\begin{example}
An example of a universal operation set is 
\begin{eqnarray}
\mathbb{O} &=& \{H_j,S_j,D_j,P_j\vert j\in Q,P=I,X,Y,Z\} \cup \{CNOT_{i,j}\vert i,j\in Q,i\neq j\} \notag \\
&& \{SP_{X,j},SP_{Z,j}\vert j\in Q\} \cup \{M_{X,j},M_{Z,j}\vert j\in Q\} \cup \{SP_{A,j}\vert j\in Q\}.
\end{eqnarray}
Here, $H_j$, $S_j$ $D_j$ and $P_j$ denote the Hadamard, phase, delay (idle operation for time synchronization) and Pauli gates on qubit-$j$, respectively; $CNOT_{i,j}$ denotes the controlled-NOT gate, where qubit-$i$ (qubit-$j$) is the control (target) qubit; $SP_{\sigma,j}$ and $M_{\sigma,j}$ denote the state preparation and measurement on qubit-$j$ in the $\sigma = X,Z$ basis, respectively; and $SP_{A,j}$ denotes the state preparation on qubit-$j$ in the magic state $\vert A\rangle = (\vert 0\rangle + e^{i\pi/4}\vert 1\rangle)/\sqrt{2}$, which is the eigenstate of the Clifford operator $(X+Y)/\sqrt{2}$ with the eigenvalue $+1$. 
\label{exp:operationset}
\end{example}

\begin{definition}
{\bf Operation sets.} In what follows, we will use the following notations for operation sets: 
\begin{eqnarray}
\mathbb{O}_{1} &=& \{\text{Single-qubit gates}\}\cup\{\text{Delay gates}\}, \\
\mathbb{O}_{2} &=& \{\text{Two-qubit gates}\}, \\
\mathbb{O}_{S.P.} &=& \{\text{State preparations}\}, \\
\mathbb{O}_{Mea.} &=& \{\text{Measurements}\}, \\
\mathbb{O}_{P} &=& \{\text{Pauli gates, i.e. operators in $\mathbb{P}_n$}\}, \\
\mathbb{O}_{S} &=& (\{\text{Single-qubit and two-qubit Clifford gates,} \notag \\
&& \text{Single-qubit state preparations in the eigenstate of a Pauli operator,} \notag \\
&& \text{Single-qubit measurements of a Pauli operator}\} - \mathbb{O}_{P})\cup\{\text{Delay gates}\}, \\
\mathbb{O}_{non-S} &=& \{\text{Single-qubit gates in the Clifford hierarchy at
the third level,} \notag \\
&& \text{Single-qubit state preparations in the eigenstate of a Hermitian Clifford operator,} \notag \\
&& \text{Single-qubit measurements of a Hermitian Clifford operator}\} - (\mathbb{O}_{S}\cup\mathbb{O}_{P}).
\end{eqnarray}
Notice that Pauli gates $\mathbb{O}_{P}$ include the identity gate $\openone$, and they are excluded from $\mathbb{O}_{S}$. The identity gate and delay gates are both idle operations, however, we assume that the identity gate has a negligible error while delay gates potentially have significant errors because of the execution time. 
\end{definition}

We remark that although all the operation sets listed above consist solely of primitive single- and two-qubit operations, in general, an operation set may also include multi-qubit operations, such as a layer of parallel primitive operations. 
 
\subsection{Completely positive maps}

We let each operation output a classical message, i.e. an outcome, to simplify expressions. Each operation $\alpha\in \mathbb{O}$ has an outcome set $m_\alpha$. If $\alpha$ is a measurement, $m_\alpha$ is the set of measurement outcomes; otherwise, $m_\alpha$ has only one element, then the message is trivial. 

\begin{example}
For the operation set in Example~\ref{exp:operationset}, the outcome set is $m_\alpha = \{\pm 1\}$ when the operation $\alpha\in \{M_{X,j},M_{Z,j}\vert j\in Q\}$ is a measurement, and the outcome set is $m_\alpha = \{0\}$ (the element could be an arbitrary letter, which does not have any physical meaning) when $\alpha\in \mathbb{O} - \{M_{X,j},M_{Z,j}\vert j\in Q\}$ is not a measurement. 
\end{example}

Each operation $\alpha\in \mathbb{O}$ is described by a set of completely positive maps 
\begin{eqnarray}
\{\mathcal{M}(\alpha,\mu)\vert\mu\in m_\alpha\} \notag
\end{eqnarray}
satisfying the condition that $\sum_{\mu\in m_\alpha}\mathcal{M}(\alpha,\mu)$ is a trace-preserving completely positive map. Given an input state $\rho$, the outcome $\mu$ occurs with a probability of $\Tr[\mathcal{M}(\alpha,\mu)\rho]$, and the corresponding the output state is 
\begin{eqnarray}
\frac{\mathcal{M}(\alpha,\mu)\rho}{\Tr[\mathcal{M}(\alpha,\mu)\rho]}. \notag
\end{eqnarray}

\subsection{Circuits}
\label{app:circuits}

\begin{definition}
{\bf Static circuit.} A static circuit of $N$ operations is an $N$-tuple $C = (c_1,c_2,\ldots,c_N)$, where $c_j\in \mathbb{O}$ is the $j$th operation in the circuit. 
\end{definition}

Let $\rho_i$ be the initial state of $n$ logical qubits. For a static circuit, the final state reads 
\begin{eqnarray}
\rho_f(\mu) = \mathcal{M}(c_N,\mu_N)\cdots\mathcal{M}(c_2,\mu_2)\mathcal{M}(c_1,\mu_1)\rho_i,
\label{eq:rhof}
\end{eqnarray}
where $\mu = (\mu_1,\mu_2,\ldots,\mu_N)$ denotes outcomes, and $\mu_j\in m_{c_j}$ is the outcome of the $j$th operation. Notice that the final state depends on outcomes. 

We remark that expressing the final state in the form of Eq.~(\ref{eq:rhof}) assumes that errors are temporally uncorrelated. 

A general quantum circuit may include not only gates but also mid-circuit state preparations and measurements. During execution, the circuit may adapt based on measurement outcomes, enabling feedback operations. Additionally, the circuit could depend on a random variable (or a set of variables) in the context of randomized compiling. We refer to such circuits as randomized dynamic quantum circuits. In these circuits, each operation depends on a random variable $\lambda$ and measurement outcomes $\mu$. 

\begin{definition}
{\bf Dynamic circuit.} A dynamic circuit of $N$ operations is an $N$-tuple $C = (c_1,c_2,\ldots,c_N)$. In the circuit, the $j$th operation is $c_j(\mu_{<j})\in \mathbb{O}$, which depends on previous measurement outcomes denoted by $\mu_{<j} = (\mu_1,\mu_2,\ldots,\mu_{j-1})$. 
\end{definition}

For a dynamic circuit, the final state reads 
\begin{eqnarray}
\rho_f(\mu) = \mathcal{M}(c_N(\mu_{<N}),\mu_N)\cdots\mathcal{M}(c_2(\mu_{<2}),\mu_2)\mathcal{M}(c_1(\mu_{<1}),\mu_1)\rho_i,
\end{eqnarray}
where $\mu_{<1} = ()$ is a $0$-tuple. 

\begin{definition}
{\bf Randomized dynamic circuit.} A randomized dynamic circuit of $N$ operations is a 2-tuple $(w,C)$, where $w(\lambda)$ is a normalized weight function of a variable $\lambda$ (called internal variable), and $C = (c_1,c_2,\ldots,c_N)$. In each circuit shot, a value of the variable $\lambda$ is generated according to the distribution $w(\lambda)$ in advance, then the $j$th operation in the circuit is $c_j(\lambda,\mu_{<j})\in \mathbb{O}$, which depends on the variable $\lambda$ and previous measurement outcomes $\mu_{<j}$. 
\label{def:RDC}
\end{definition}

For a randomized dynamic circuit, the final state reads 
\begin{eqnarray}
\rho_f(\lambda,\mu) = \mathcal{M}(c_N(\lambda,\mu_{<N}),\mu_N)\cdots\mathcal{M}(c_2(\lambda,\mu_{<2}),\mu_2)\mathcal{M}(c_1(\lambda,\mu_{<1}),\mu_1)\rho_i.
\label{eq:noisy_circuit}
\end{eqnarray}

\begin{definition}
{\bf Parametrized randomized dynamic circuit.} A parametrized randomized dynamic circuit of $N$ operations is a 2-tuple $(w,C)$, where $(w(\theta),C(\theta))$ is a randomized dynamic circuit of $N$ operations that depends on a variable $\theta$ (called external variable). Before implementing the circuit, the value of $\theta$ must be specified. Then, the distribution of $\lambda$ is given by $w(\theta,\lambda)$, and the $j$th operation in the circuit is $c_j(\theta,\lambda,\mu_{<j})\in \mathbb{O}$. 
\label{def:PRDC}
\end{definition}

\begin{definition}
{\bf Composite circuit.} A composite circuit is a randomized dynamic circuit. If composed of $\tilde{N}$ sub-circuits, the composite circuit is a 3-tuple $(W,\tilde{w},\tilde{C})$, where $W(\theta)$ is a normalized weight function of a variable $\theta$, $\tilde{w} = (w^{(1)},w^{(2)},\ldots,w^{(\tilde{N})})$ and $\tilde{C} = (C^{(1)},C^{(2)},\ldots,C^{(\tilde{N})})$. Each pair $(w^{(l)},C^{(l)})$, a sub-circuit, is a parametrized randomized dynamic circuit; measurement outcomes in the sub-circuit is denoted by $\mu^{(l)}$; and external variable of the sub-circuit is $(\theta,\mu^{(<l)})$, where $\mu^{(<l)}$ denotes measurement outcomes of previous sub-circuits. In each circuit shot, a value of $\theta$ is generated according to the distribution $W(\theta)$ in advance, then the sub-circuits $(w^{(l)}(\theta,\mu^{(<l)}),C^{(l)}(\theta,\mu^{(<l)}))$ are applied one by one in the order of $l = 1,2,\ldots,\tilde{N}$. 

Suppose the $l$th sub-circuit is a circuit of $N^{(l)}$ operations. The composite circuit is a randomized dynamic circuit $(w,C)$ of $N = \sum_{l=1}^{\tilde{N}} N^{(l)}$ operations. 
\end{definition}

\begin{definition}
{\bf Twirled operation.} A twirled operation is a parametrized dynamic circuit $(u(\alpha),T(\alpha))$, where the external variable is an operation $\alpha\in \mathbb{O}$. For operations in $\mathbb{O}_{P}\cup\mathbb{O}_{S}\cup\mathbb{O}_{non-S}$, we list the corresponding circuits in Table~\ref{tab:twirled_operations}. 

Notice that operation lists in Table~\ref{tab:twirled_operations} have different lengths, and we can make them the same length by adding trivial operations $\openone$ to the list. For example, we replace $(\alpha^\dag P\alpha,\alpha,P)$ with $(\openone,\openone,\alpha^\dag P\alpha,\alpha,P)$. Here, $\openone$ denotes that there is not any physical operation applied, and the only purpose of adding trivial operations is to satisfy Definition~\ref{def:PRDC}, in which operation lists have the same length $N$ for all values of the external variable. The operation number in the twirled operation is $N = 5$, which is the length of the longest operation list in Table~\ref{tab:twirled_operations}. 
\label{def:twirled_operation}
\end{definition}

\begin{table*}[!htbp]
    \setlength{\tabcolsep}{12pt}
    \renewcommand{\arraystretch}{1.5}
    \centering
    \caption{\label{tab:twirled_operations} {\bf Twirled operations.} The distribution is always uniform. When $\alpha\in \mathbb{O}_P$, the circuit is deterministic. When $\alpha\in \mathbb{O}_{S.P.}$ ($\alpha\in \mathbb{O}_{Mea.}$), $\kappa$ is the basis of state preparation (measurement), i.e. the prepared state is the eigenstate of $\kappa$ with the eigenvalue of $+1$ (the operation $\kappa$ is measured). When $\alpha\in \mathbb{O}_{S}\cap(\mathbb{O}_{S.P.}\cup\mathbb{O}_{Mea.})$, $\kappa$ is a Pauli operator; and when $\alpha\in \mathbb{O}_{non-S}\cap(\mathbb{O}_{S.P.}\cup\mathbb{O}_{Mea.})$, $\kappa$ is a Hermitian Clifford operator but not Pauli. 
    }
    \begin{tabular}{c|cccc}
        External variable & Internal variable & Distribution $u(\alpha,\bullet)$ & Operation list $T(\alpha)$ \\
        \hline \hline
        $\alpha\in \mathbb{O}_{P}$ & $P\in\{\openone\}$ & $u(\alpha,P) = 1$ & $(\alpha)$ \\
        \hline
        $\alpha\in \mathbb{O}_{S}\cap(\mathbb{O}_{1}\cup\mathbb{O}_{2})$ & $P\in\mathbb{P}_{Q(\alpha)}$ & $u(\alpha,P) = \frac{1}{\vert\mathbb{P}_{Q(\alpha)}\vert}$ & $(\alpha^\dag P\alpha,\alpha,P)$ \\
        \hline
        $\alpha\in \mathbb{O}_{S}\cap\mathbb{O}_{S.P.}$ & $P\in\{\openone,\kappa\}$ & $u(\alpha,P) = \frac{1}{2}$ & $(\alpha,P)$ \\
        \hline
        $\alpha\in \mathbb{O}_{S}\cap\mathbb{O}_{Mea.}$ & $P\in\{\openone,\kappa\}$ & $u(\alpha,P) = \frac{1}{2}$ & $(P,\alpha)$ \\
        \hline
        $\alpha\in \mathbb{O}_{non-S}\cap\mathbb{O}_{1}$ & $(P,P')\in\mathbb{P}_{Q(\alpha)}\times\mathbb{P}_{Q(\alpha)}$ & $u(\alpha,P,P') = \frac{1}{\vert\mathbb{P}_{Q(\alpha)}\vert^2}$ & $(\alpha^\dag P\alpha P'\alpha^\dag P\alpha,\alpha^\dag P\alpha,P',\alpha,P)$ \\
        \hline
        $\alpha\in \mathbb{O}_{non-S}\cap\mathbb{O}_{S.P.}$ & $(P,P')\in\{\openone,\kappa\}\times\mathbb{P}_{Q(\alpha)}$ & $u(\alpha,P,P') = \frac{1}{2\vert\mathbb{P}_{Q(\alpha)}\vert}$ & $(\alpha,PP'P,P,P')$ \\
        \hline
        $\alpha\in \mathbb{O}_{non-S}\cap\mathbb{O}_{Mea.}$ & $(P,P')\in\{\openone,\kappa\}\times\mathbb{P}_{Q(\alpha)}$ & $u(\alpha,P,P') = \frac{1}{2\vert\mathbb{P}_{Q(\alpha)}\vert}$ & $(PP'P,P,P',\alpha)$ \\
        \hline
    \end{tabular}
\end{table*}

\begin{definition}
{\bf Twirled circuit.} A twirled circuit is a composite circuit $(W,\tilde{w},\tilde{C})$. Each sub-circuit is in the form $w^{(l)}(\theta,\mu^{(<l)}) = u(\alpha_l(\theta,\mu^{(<l)}))$ and $C^{(l)}(\theta,\mu^{(<l)}) = T(\alpha_l(\theta,\mu^{(<l)}))$, i.e. $(w^{(l)}(\theta,\mu^{(<l)}),C^{(l)}(\theta,\mu^{(<l)}))$ is a twirled operation that the external-variable operation $\alpha_l(\theta,\mu^{(<l)})\in \mathbb{O}$ depends on $\theta$ and measurement outcomes of previous sub-circuits. 

Suppose the twirled circuit is composed of $\tilde{N}$ twirled operations, the circuit is a randomized dynamic circuit $(w,C)$ of $N = 5\tilde{N}$ operations. 
\label{def:twirled_circuit}
\end{definition}

\subsection{Noise}

An operation with noise can be expressed as 
\begin{eqnarray}
\mathcal{M}(\alpha,\mu) = \mathcal{N}^L(\alpha)\mathcal{M}^I(\alpha,\mu)\mathcal{N}^R(\alpha),
\label{eq:noisy_operation}
\end{eqnarray}
where $\mathcal{M}^I(\alpha,\mu)$ is the ideal operation, $\mathcal{N}^L(\alpha)$ and $\mathcal{N}^R(\alpha)$ are trace-preserving completely positive maps that representing the noise. When $\alpha$ is a measurement, the noise occurs before the operation, i.e. $\mathcal{N}^L(\alpha) = [\openone]$; otherwise, the noise occurs after the operation, i.e. $\mathcal{N}^R(\alpha) = [\openone]$. Here, $\openone$ is the identity operation acting on the Hilbert space of $n$ logical qubits, and $[U]\bullet = U\bullet U^\dag$. 

For each operation, there is only one non-trivial noise map in $\mathcal{N}^L(\alpha)$ and $\mathcal{N}^R(\alpha)$. We use $\mathcal{N}(\alpha)$ to denote the non-trivial noise map of the operation $\alpha$. When $\alpha$ is a measurement, $\mathcal{N}(\alpha) = \mathcal{N}^R(\alpha)$; otherwise, $\mathcal{N}(\alpha) = \mathcal{N}^L(\alpha)$. 

\begin{definition}
{\bf Operation support.} We use $Q(\alpha)\subseteq Q$ to denote the support of the operation $\alpha \in \mathbb{O}$, which represents the subset of qubits on which the ideal operation acts non-trivially. 
\end{definition}

\begin{definition}
{\bf Noisy operation support.} We use $\tilde{Q}(\alpha)\subseteq Q$ to denote the support of the operation $\alpha \in \mathbb{O}$ with noise. Let $Q_\mathcal{N}(\alpha)\subseteq Q$ be the subset of qubits on which the noise map $\mathcal{N}(\alpha)$ acts non-trivially. Then, $\tilde{Q}(\alpha) = Q(\alpha)\cup Q_\mathcal{N}(\alpha)$. 
\end{definition}

\begin{definition}
{\bf Pauli noise.} The noise in an operation $\alpha$ is said to be Pauli errors if and only if the corresponding noise map is in the form 
\begin{eqnarray}
\mathcal{N}(\alpha) = \sum_{\tau\in \mathbb{P}_n} \epsilon(\alpha,\tau)[\tau],
\end{eqnarray}
where $\epsilon(\alpha,\tau)$ is the rate of the Pauli error $[\tau]$. We can rewrite the map as 
\begin{eqnarray}
\mathcal{N}(\alpha) = \left(1-p(\alpha)\right)[\openone] + p(\alpha)\mathcal{E}(\alpha),
\end{eqnarray}
where 
\begin{eqnarray}
\mathcal{E}(\alpha) = \frac{1}{p(\alpha)}\sum_{\tau\in \mathbb{P}_n-\{\openone\}} \epsilon(\alpha,\tau)[\tau]
\end{eqnarray}
is a map denoting errors, and $p(\alpha) = \sum_{\tau\in \mathbb{P}_n-\{\openone\}} \epsilon(\alpha,\tau)$ is the total error rate of the operation. 
\label{def:Pauli_noise}
\end{definition}

\section{Universal representation of noisy randomized dynamic circuits}
\label{app:NRDC}

A general quantum circuit may include not only gates but also mid-circuit state preparations and measurements. During execution, the circuit may adapt based on measurement outcomes, enabling feedback operations. Additionally, the circuit could depend on a random variable (or a set of variables) in the context of randomized compiling. We refer to such circuits as randomized dynamic quantum circuits. In these circuits, each operation depends on a random variable $\lambda$ and measurement outcomes $\mu$. 

In a randomized dynamic quantum circuit, the spacetime noise $\mathcal{N}_{st}$ depends on both $\lambda$ and $\mu$. If we mitigate errors by inverting $\mathcal{N}_{st}$, we have to measure the total error rates for all possible $\lambda$ and $\mu$. To solve this problem, we introduce the maximum spacetime noise map $\mathcal{N}_{max}$; see Definition~\ref{def:MSN}. 

We can mitigate errors in a randomized dynamic quantum circuit by applying the inverse of the maximum spacetime noise, as justified by Theorem~\ref{the:MSN} (see Sec.~\ref{app:NRDC_proof} for the proof). By applying the inverse, we replace $\mathcal{N}_{max}$ with $\mathcal{N}_{max}^{-1} \mathcal{N}_{max}$ in the $F$ function, effective realizing an error-free final state. Since the function $F$ is linear, we can apply $\mathcal{N}_{max}^{-1}$ via the Taylor expansion. This approach allows us to mitigate errors by only measuring one parameter, the total error rate of $\mathcal{N}_{max}$. 

\begin{theorem}
Given a randomized dynamic circuit with temporally uncorrelated errors, the unnormalized final state is a linear function of the maximum spacetime noise map, i.e., 
\begin{equation}
\rho_f(\lambda, \mu) = F(\lambda, \mu, \mathcal{N}_{max}),
\label{eq:rhof_F_MSN}
\end{equation}
where $F(\lambda, \mu, \bullet)$ is a linear function. Note that the final state is unnormalized due to measurements in the circuit. 
\label{the:MSN}
\end{theorem}

\subsection{Maximum spacetime noise}

\begin{definition}
{\bf Maximum operation number.} Given a randomized dynamic circuit $(w,C)$, a value of the random variable $
\lambda$ and an outcome $\mu$, the number of the operation $\alpha\in\mathbb{O}$ in the circuit is 
\begin{eqnarray}
N_\alpha(\lambda,\mu) = \sum_{j=1}^N\delta_{\alpha,c_j(\lambda,\mu_{<j})},
\end{eqnarray}
where $\delta$ takes one if $\alpha = c_j(\lambda,\mu_{<j})$ or zero otherwise. The maximum number of the operation $\alpha$ is 
\begin{eqnarray}
N_\alpha^{max} = \max_{\lambda,\mu} \left\{N_\alpha(\lambda,\mu)\right\}.
\end{eqnarray}
\end{definition}

\begin{definition}
{\bf Maximum spacetime noise.} The maximum spacetime noise map of temporally uncorrelated errors is defined as 
\begin{eqnarray}
\mathcal{N}_{max} = \bigotimes_{\alpha\in\mathbb{O}} \mathcal{N}(\alpha)^{\otimes N_\alpha^{max}}.
\end{eqnarray}
\label{def:MSN}
\end{definition}

\subsection{Proof of Theorem~\ref{the:MSN}}
\label{app:NRDC_proof}

We use $\vert\rho\rangle\rangle$ to denote the vectorization of the matrix $\rho$, and we use $\bar{\bar{\mathcal{M}}}$ to denote the matrix acting on vectors $\vert\rho\rangle\rangle$ that represents the map $\mathcal{M}$. Then, the final state reads 
\begin{eqnarray}
\vert\rho_f(\lambda,\mu)\rangle\rangle &=& \bar{\bar{\mathcal{M}}}(C,\lambda,\mu) \vert\rho_i\rangle\rangle,
\end{eqnarray}
where 
\begin{eqnarray}
\bar{\bar{\mathcal{M}}}(C,\lambda,\mu) &=& \bar{\bar{\mathcal{N}}}^L(c_N(\lambda,\mu_{<N}))\bar{\bar{\mathcal{M}}}^I(c_N(\lambda,\mu_{<N}),\mu_N)\bar{\bar{\mathcal{N}}}^R(c_N(\lambda,\mu_{<N}))\times\cdots \notag \\
&&\times \bar{\bar{\mathcal{N}}}^L(c_1(\lambda,\mu_{<1}))\bar{\bar{\mathcal{M}}}(c_1(\lambda,\mu_{<1}),\mu_1)\bar{\bar{\mathcal{N}}}^R(c_1(\lambda,\mu_{<1})),
\end{eqnarray}
and we have taken into account the noise. 

Now, we consider a product of two matrices, $\bar{\bar{\mathcal{M}}}_2\bar{\bar{\mathcal{M}}}_1$. Let $\{\vert j\rangle\rangle\}$ be an orthonormal basis. The product can be rewritten as 
\begin{eqnarray}
\bar{\bar{\mathcal{M}}}_2\bar{\bar{\mathcal{M}}}_1 &=& \sum_{a,b,c} \vert a\rangle\rangle\langle\langle a\vert \bar{\bar{\mathcal{M}}}_2 \vert b\rangle\rangle\langle\langle b\vert \bar{\bar{\mathcal{M}}}_1 \vert c\rangle\rangle\langle\langle c\vert \notag \\
&=& \sum_{a,b,c} \left(\langle\langle a\vert\otimes\langle\langle b\vert \bar{\bar{\mathcal{M}}}_2\otimes \bar{\bar{\mathcal{M}}}_1 \vert b\rangle\rangle\otimes\vert c\rangle\rangle\right) \vert a\rangle\rangle\langle\langle c\vert.
\end{eqnarray}
Similarly, for a product of $K$ matrices, we have 
\begin{eqnarray}
\bar{\bar{\mathcal{M}}}_K\cdots\bar{\bar{\mathcal{M}}}_2\bar{\bar{\mathcal{M}}}_1 &=& \sum_{b_1,b_2,\ldots,b_{K+1}} \vert b_{K+1}\rangle\rangle\langle\langle b_{K+1}\vert \bar{\bar{\mathcal{M}}}_K \vert b_k\rangle\rangle\cdots\langle\langle b_3\vert \bar{\bar{\mathcal{M}}}_2 \vert b_2\rangle\rangle\langle\langle b_2\vert \bar{\bar{\mathcal{M}}}_1 \vert b_1\rangle\rangle\langle\langle b_1\vert \notag \\
&=& \sum_{b_1,b_2,\ldots,b_{K+1}} \left(\langle\langle b_{K+1}\vert\otimes\cdots\otimes\langle\langle b_3\vert\otimes\langle\langle b_2\vert \bar{\bar{\mathcal{M}}}_K\otimes\cdots\otimes\bar{\bar{\mathcal{M}}}_2\otimes\bar{\bar{\mathcal{M}}}_1 \right. \notag \\
&&\times\left.\vert b_K\rangle\rangle\otimes\cdots\otimes\vert b_2\rangle\rangle\otimes\vert b_1\rangle\rangle\right) \vert b_{K+1}\rangle\rangle\langle\langle b_1\vert.
\end{eqnarray}

\begin{definition}
{\bf Linear maps on completely positive maps.} We define two linear maps
\begin{eqnarray}
\mathcal{P}_K\bullet &=& \sum_{b_1,b_2,\ldots,b_{K+1}} \left(\langle\langle b_{K+1}\vert\otimes\cdots\otimes\langle\langle b_3\vert\otimes\langle\langle b_2\vert \bullet \vert b_K\rangle\rangle\otimes\cdots\otimes\vert b_2\rangle\rangle\otimes\vert b_1\rangle\rangle\right) \vert b_{K+1}\rangle\rangle\langle\langle b_1\vert, \\
\mathcal{T}\bullet &=& \langle\langle \openone\vert \bullet \vert \openone/2^n\rangle\rangle.
\end{eqnarray} 
\label{def:maps}
\end{definition}

Using the map $\mathcal{P}_K$, we can rewrite the map $\bar{\bar{\mathcal{M}}}(C,\lambda,\mu)$ as 
\begin{eqnarray}
\bar{\bar{\mathcal{M}}}(C,\lambda,\mu) &=& \mathcal{P}_K \bar{\bar{\mathcal{M}}}_{st}(C,\lambda,\mu),
\end{eqnarray}
where 
\begin{eqnarray}
\bar{\bar{\mathcal{M}}}_{st}(C,\lambda,\mu) &=& \bar{\bar{\mathcal{N}}}^L(c_N(\lambda,\mu_{<N}))\otimes\bar{\bar{\mathcal{M}}}^I(c_N(\lambda,\mu_{<N}),\mu_N)\otimes\bar{\bar{\mathcal{N}}}^R(c_N(\lambda,\mu_{<N}))\otimes\cdots \notag \\
&&\otimes \bar{\bar{\mathcal{N}}}^L(c_1(\lambda,\mu_{<1}))\otimes\bar{\bar{\mathcal{M}}}(c_1(\lambda,\mu_{<1}),\mu_1)\otimes\bar{\bar{\mathcal{N}}}^R(c_1(\lambda,\mu_{<1})).
\end{eqnarray}

\begin{proposition}
Given a randomized dynamic circuit $(w,C)$, a value of the random variable $
\lambda$ and an outcome $\mu$, there exists a permutation operation $\mathcal{S}(\lambda,\mu)$ such that 
\begin{eqnarray}
&&\mathcal{S}(\lambda,\mu)\bar{\bar{\mathcal{N}}}_{max}\otimes\bar{\bar{[\openone]}}^{\otimes\left(\sum_{\alpha\in\mathbb{O}}N_\alpha(\lambda,\mu)\right)}\otimes\bar{\bar{\mathcal{M}}}_{st}^{I}(C,\lambda,\mu) \notag \\
&=& \left[\bigotimes_{\alpha\in\mathbb{O}} \bar{\bar{\mathcal{N}}}(\alpha)^{\otimes\left(N_\alpha^{max}-N_\alpha(\lambda,\mu)\right)}\right]\otimes\bar{\bar{\mathcal{M}}}_{st}(C,\lambda,\mu),
\end{eqnarray}
where 
\begin{eqnarray}
\bar{\bar{\mathcal{M}}}_{st}^I(C,\lambda,\mu) &=& \bar{\bar{\mathcal{M}}}^I(c_N(\lambda,\mu_{<N}),\mu_N)\otimes\cdots\otimes\bar{\bar{\mathcal{M}}}^I(c_1(\lambda,\mu_{<1}),\mu_1).
\end{eqnarray}
Here, $\mathcal{S}(\lambda,\mu)$ is a linear map that permutes matrices in the Kronecker product, and $[\openone]$ denotes the trivial noise maps. 
\end{proposition}

Since $\mathcal{N}(\alpha)$ are trace-preserving completely positive maps, 
\begin{eqnarray}
\mathcal{T}\bar{\bar{\mathcal{N}}}(\alpha) = \Tr[\openone\mathcal{N}(\alpha)\openone/2^n] = 1.
\end{eqnarray}
We can further rewrite the map $\bar{\bar{\mathcal{M}}}(C,\lambda,\mu)$ as 
\begin{eqnarray}
\bar{\bar{\mathcal{M}}}(C,\lambda,\mu) = \mathcal{T}^{\otimes\left(\sum_{\alpha\in\mathbb{O}}N_\alpha^{max}-N_\alpha(\lambda,\mu)\right)}\otimes\mathcal{P}_K \left[\mathcal{S}(\lambda,\mu)\bar{\bar{\mathcal{N}}}_{max}\otimes\bar{\bar{[\openone]}}^{\otimes\left(\sum_{\alpha\in\mathbb{O}}N_\alpha(\lambda,\mu)\right)}\otimes\bar{\bar{\mathcal{M}}}_{st}^{I}(C,\lambda,\mu)\right].
\end{eqnarray}
Therefore, the function reads 
\begin{eqnarray}
F(\lambda,\mu,\bullet) = \mathrm{vec}^{-1}\mathcal{T}^{\otimes\left(\sum_{\alpha\in\mathbb{O}}N_\alpha^{max}-N_\alpha(\lambda,\mu)\right)}\otimes\mathcal{P}_K \left[\mathcal{S}(\lambda,\mu)\bullet\otimes\bar{\bar{[\openone]}}^{\otimes\left(\sum_{\alpha\in\mathbb{O}}N_\alpha(\lambda,\mu)\right)}\otimes\bar{\bar{\mathcal{M}}}_{st}^{I}(C,\lambda,\mu)\right]\vert\rho_i\rangle\rangle,
\label{eq:F_function}
\end{eqnarray}
where $\mathrm{vec}^{-1}$ is the inverse map of vectorization. 

\section{Observable evaluation}
\label{app:observable}

On a quantum computer, we can evaluate observables that are functions of measurement outcomes. For example, suppose we want to evaluate the observable $A = Z_1 + X_2Y_3$. Accordingly, we need to measure the three operators $Z_1$, $X_2$ and $Y_3$ in the quantum circuit, yielding measurement outcomes $\mu_1$, $\mu_2$ and $\mu_3$, respectively. Then, the expected value of the observable reads $\langle A\rangle = \mathrm{E}[\mu_1+\mu_2\mu_3]$. We can find that $A$ corresponds to a function of measurement outcomes. In certain cases, the function depends on a random variable. For example, suppose we want to evaluate the observable $A = Z_1 + X_1Y_2/3$. We can introduce a random variable: $\lambda = 0$ ($\lambda = 1$) with a probability of $1/2$; when $\lambda = 0$, we measure $Z_1$ in the quantum circuit, yielding the measurement outcome $\mu = (\mu_1)$; and when $\lambda = 1$, we measure $X_1$ and $Y_2$ in the quantum circuit, yielding measurement outcomes $\mu = (\mu_1,\mu_2)$. Then, the expected value of the observable reads $\langle A\rangle = \mathrm{E}[a(\lambda,\mu)]$, where 
\begin{eqnarray}
a(\lambda,\mu) = \begin{cases}
2\mu_1, & \text{if $\lambda = 0$;} \\
2\mu_1\mu_2/3, & \text{if $\lambda = 1$.}
\end{cases}
\end{eqnarray}

In general, the expected value of an observable is in the form 
\begin{eqnarray}
\langle A\rangle = \mathrm{E}[a(\lambda,\mu)],
\end{eqnarray}
where $a$ is the function corresponding to the observable $A$, $\lambda$ is a random variable, and $\mu$ denotes measurement outcomes in the circuit. Given a randomized dynamic circuit $(w,C)$, the probability of the outcome $\mu$ is $\Tr\rho_f(\lambda,\mu)$ given an value of $\lambda$, then 
\begin{eqnarray}
\langle A\rangle = \sum_{\lambda,\mu}w(\lambda)a(\lambda,\mu)\Tr\rho_f(\lambda,\mu) = \sum_{\lambda,\mu}w(\lambda)a(\lambda,\mu)\Tr F(\lambda,\mu,\mathcal{N}_{max}).
\end{eqnarray}
The ideal expected value is 
\begin{eqnarray}
\langle A\rangle_I = \sum_{\lambda,\mu}w(\lambda)a(\lambda,\mu)\Tr F(\lambda,\mu,[\openone^{max}]),
\end{eqnarray}
where $[\openone^{max}] \equiv [\openone]^{\otimes N_{total}}$ denotes that the maximum spacetime noise map is identity, i.e. the circuit is error-free. Here, $N_{total} = \sum_{\alpha\in\mathbb{O}} N^{max}_\alpha$. 

\section{Protocol with an ideal error sampler}
\label{app:ideal_protocol}

\begin{figure}[htbp]
\centering
\includegraphics[width=0.5\linewidth]{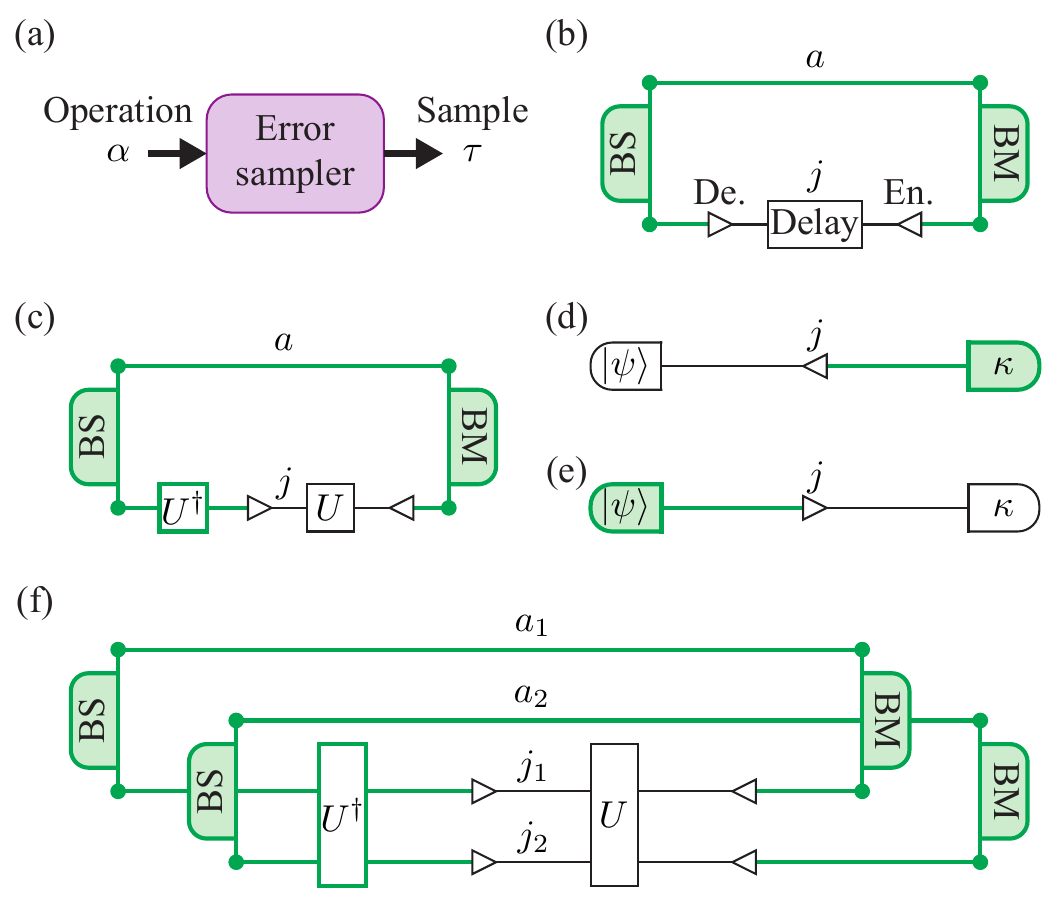}
\caption{
(a) Error Sampler. It takes an operation $\alpha\in \mathbb{O}$ as the input, and the output is a Pauli operator $\tau\in\mathbb{P}_n$. (b-e) Circuits for realizing the error sampler. Thin black lines represent logical qubits, while thick green lines represent super qubits. In (b), (c) and (f), each pair of qubits is prepared in the Bell state  (BS) $(\ket{0}\otimes\ket{0}+\ket{1}\otimes\ket{1})/\sqrt{2}$ and measured in the basis $\{X\otimes X,Z\otimes Z\}$, called Bell measurement (BS). Triangles denote encoding (En.) and decoding (De.) operations that enable transitions between logical qubits and super qubits. $U$ is the unitary operator of a quantum gate, and $\ket{\psi}$ is the eigenstate of the operator $\kappa$ with the eigenvalue $+1$. In the circuits, Pauli twirling is applied to encoding, decoding and non-Pauli stabilizer operations $\alpha\in\mathbb{O}_{S}$ on logical qubits. The output Pauli operator $\tau$ is determined from the measurement outcomes according to Tables \ref{tab:SqErrorsampler}, \ref{tab:TqErrorsampler} and \ref{tab:SpErrorsampler}. 
}
\label{fig:error_sampler}
\end{figure}

In this protocol, we assume that 
\begin{itemize}
\item[i)] Pauli gates are error-free; 
\item[ii)] The noise is Pauli in all non-Pauli operations; 
\item[iii)] There is an ideal error sampler from which we can faithfully sample the Pauli errors of the target noise map. 
\end{itemize}

The error sampler takes an operation $\alpha\in \mathbb{O}$ as the input, and the output is a Pauli operator $\tau\in\mathbb{P}_n$; see Fig.~\ref{fig:error_sampler}(a). Each operation $\alpha$ has an associated Pauli error model $\mathcal{N}(\alpha)$. Ideally, the error sampler produces the Pauli operator $\tau$ according to the distribution described by $\mathcal{N}(\alpha)$: The operator is $\tau$ with a probability of $1-p(\alpha)$ if $\tau = \openone$ or a probability of $\epsilon(\alpha,\tau)$ if $\tau \neq \openone$. The ideal error sampler is described by Algorithm~\ref{alg:ideal_error_sampler} (see Sec.~\ref{app:codes}). To implement the protocol with the ideal error sampler, run the program as described in Algorithm~\ref{alg:spacetime_noise_inversion} as follows: 
\begin{center}
\textsc{InverseNoiseExpansion}($w$,$C$,$a$,$M_P$,$M$,\textsc{Ideal\_ErrorSampler},\textsc{Ideal\_ProcessedErrorSampler}).
\end{center}
Note that we have discriminated the sampling from the noise map $\mathcal{N}(\alpha)$ of an operation $\alpha$, maximum spacetime noise map $\mathcal{N}_{max}$, and the $k$-th order non-trivial noise map $\mathcal{E}^k$ using the terms error sampler, spacetime error sampler, and processed error sampler, respectively. 

\section{Protocol with a practical error sampler}
\label{app:practical_protocol}

In this section, we present a protocol based on a practical error sampler, which further captures various imperfections such as non-Pauli errors and additional errors in error sampler circuits, and thus also presented in Fig.~2 in the main text. The assumptions for the practical error sampler are follows: 
\begin{itemize}
\item[i)] Pauli gates are error-free; 
\item[ii)] The operation set satisfies $\mathbb{O}\subseteq \mathbb{O}_{P}\cup\mathbb{O}_{S}\cup\mathbb{O}_{non-S}$, i.e. only single-qubit and two-qubit operations are used; 
\item[iii)] The circuit is twirled (see Definition~\ref{def:twirled_circuit}); 
\item[iv)] Through encoding and decoding operations, one can transfer states between logical qubits and super qubits, which are qubits encoded in quantum error correction codes with a larger distance than logical qubits such that errors in super qubits are negligible. 
\end{itemize}
The generalization to operation sets with multi-qubit Clifford gates is straightforward. The practical error sampler is realized using quantum circuits in figures (b-f); for a detailed pseudocode, see Algorithm~\ref{alg:practical_error_sampler} in Sec.~\ref{app:codes}. To implement the protocol with the practical error sampler, run the program as described in Algorithm~\ref{alg:spacetime_noise_inversion} as follows: 
\begin{center}
\textsc{InverseNoiseExpansion} ($w$,$C$,$a$,$M_P$,$M$,\textsc{Practical\_ErrorSampler},\textsc{Practical\_ProcessedErrorSampler}).
\end{center}

\begin{table*}[!htbp]
    \setlength{\tabcolsep}{12pt}
    \renewcommand{\arraystretch}{1.5}
    \centering
    \caption{\label{tab:SqErrorsampler} {\bf Single-qubit gate error sampling.}}
    \begin{tabular}{c|ccc}
        \hline \hline
        Operation & Measurement basis & Outcome $\mu$ & Pauli error $\tau$ \\
        \hline
        $[U]$ & $(X_aX_j,Z_aZ_j)$ & $(+1,+1)$ & $\openone$ \\
        & & $(+1,-1)$ & $X_j$ \\
        & & $(-1,+1)$ & $Z_j$ \\
        & & $(-1,-1)$ & $Y_j$ \\
        \hline
    \end{tabular}
\end{table*}

If the operation is a single-qubit gate, i.e. $\alpha\in \mathbb{O}_1$, we can sample the error with the circuit in Fig.~\ref{fig:error_sampler} (c). In the figure, $U$ denotes the unitary operator of the gate. Let $\rho_{Bell} = \ketbra{\phi_{Bell}}{\phi_{Bell}}$ be the Bell state, and $\ket{\phi_{Bell}} = (\ket{0}_a\otimes\ket{0}_j+\ket{1}_a\otimes\ket{1}_j)/\sqrt{2}$. Here, $a$ denotes the ancilla logical qubit, and the single-qubit gate acts on the $j$th logical qubit. The final state of the circuit is $\mathcal{N}(\alpha)\rho_{Bell}$. By measuring the two stabilizer operators of the Bell state $X_aX_j$ and $Z_aZ_j$, we can read out the Pauli error on the Bell state, i.e. the error caused by $\mathcal{N}(\alpha)$. The correspondence between measurement outcomes and the error is illustrated in Table~\ref{tab:SqErrorsampler}. 

\begin{table*}[!htbp]
    \setlength{\tabcolsep}{12pt}
    \renewcommand{\arraystretch}{1.5}
    \centering
    \caption{\label{tab:TqErrorsampler} {\bf Two-qubit gate error sampling.}}
    \begin{tabular}{c|ccc}
        \hline \hline
        Operation & Measurement basis & Outcome $\mu$ & Pauli error $\tau$ \\
        \hline
        $[U]$ & $(X_{a_1}X_{j_1},Z_{a_1}Z_{j_1},X_{a_2}X_{j_2},Z_{a_2}Z_{j_2})$ & $(+1,+1,+1,+1)$ & $\openone$ \\
        & & $(+1,+1,+1,-1)$ & $X_{j_2}$ \\
        & & $(+1,+1,-1,+1)$ & $Z_{j_2}$ \\
        & & $(+1,+1,-1,-1)$ & $Y_{j_2}$ \\
        & & $(+1,-1,+1,+1)$ & $X_{j_1}$ \\
        & & $(+1,-1,+1,-1)$ & $X_{j_1}X_{j_2}$ \\
        & & $(+1,-1,-1,+1)$ & $X_{j_1}Z_{j_2}$ \\
        & & $(+1,-1,-1,-1)$ & $X_{j_1}Y_{j_2}$ \\
        & & $(-1,+1,+1,+1)$ & $Z_{j_1}$ \\
        & & $(-1,+1,+1,-1)$ & $Z_{j_1}X_{j_2}$ \\
        & & $(-1,+1,-1,+1)$ & $Z_{j_1}Z_{j_2}$ \\
        & & $(-1,+1,-1,-1)$ & $Z_{j_1}Y_{j_2}$ \\
        & & $(-1,-1,+1,+1)$ & $Y_{j_1}$ \\
        & & $(-1,-1,+1,-1)$ & $Y_{j_1}X_{j_2}$ \\
        & & $(-1,-1,-1,+1)$ & $Y_{j_1}Z_{j_2}$ \\
        & & $(-1,-1,-1,-1)$ & $Y_{j_1}Y_{j_2}$ \\
        \hline
    \end{tabular}
\end{table*}

Similarly, if the operation is a two-qubit gate, i.e. $\alpha\in \mathbb{O}_2$, we can sample the error with the circuit in Fig.~\ref{fig:error_sampler} (d). This time, we need two copies of the Bell state on four logical qubits: $a_1$ and $a_2$ are ancilla logical qubits, and the two-qubit gate acts on logical qubits $j_1$ and $j_2$. By measuring operators $X_{a_1}X_{j_1}$, $Z_{a_1}Z_{j_1}$, $X_{a_2}X_{j_2}$ and $Z_{a_2}Z_{j_2}$, we can read out the Pauli error. The correspondence between measurement outcomes and the error is illustrated in Table~\ref{tab:TqErrorsampler}. 

\begin{table*}[!htbp]
    \setlength{\tabcolsep}{12pt}
    \renewcommand{\arraystretch}{1.5}
    \centering
    \caption{\label{tab:SpErrorsampler} {\bf State preparation and measurement error sampling.}}
    \begin{tabular}{cc|cc}
        \hline \hline
        State $\ket{\psi}$ &  Measurement basis $\kappa$ & Outcome $\mu$ & Pauli error $\tau$ \\
        \hline
        $\ket{0}$& $Z$ &  $+1$ & $\openone$ \\
        & & $-1$ & $X$ \\
        \hline
        $\ket{+}$& $X$ &  $+1$ & $\openone$ \\
        & & $-1$ & $Z$ \\
        \hline
        $(\ket{0}+e^{i\pi/4}\ket{1})/\sqrt{2}$ & $(X+Y)/\sqrt{2}$ & $+1$ & $\openone$ \\
        & & $-1$ & $Z$ \\
        \hline
        $(\sqrt{2+\sqrt{2}}\ket{0}+\sqrt{2-\sqrt{2}}\ket{1})/2$ & $(X+Z)/\sqrt{2}$ & $+1$ & $\openone$ \\
        & & $-1$ & $Y$ \\
        \hline
        $(\sqrt{2+\sqrt{2}}\ket{0}+i\sqrt{2-\sqrt{2}}\ket{1})/2$ & $(Y+Z)/\sqrt{2}$ & $+1$ & $\openone$ \\
        & & $-1$ & $X$ \\
        \hline
    \end{tabular}
\end{table*}

If the operation is a state preparation, i.e. $\alpha\in \mathbb{O}_{S.P.}$, we can sample the error with the circuit in Fig.~\ref{fig:error_sampler} (e). The operation $\alpha$ prepares the single-qubit state $\ket{\psi}$. Under the assumption $\mathbb{O}\subseteq \mathbb{O}_{P}\cup \mathbb{O}_{S}\cup \mathbb{O}_{non-S}$, there always exists a Hermitian operator $\tau$ such that i) $\tau$ has two non-degenerate eigenstates with eigenvalues $\pm 1$, respectively; ii) $\ket{\psi}$ is the eigenstate of the operator $\kappa$ with the eigenvalue $+1$; and iii) there exists a Pauli operator $E$ satisfying $\{E,\kappa\} = 0$. Then, if the measurement outcome is $+1$ ($-1$), the output error is $\tau = \openone$ ($\tau = E$). The correspondence between measurement outcomes and the error is illustrated in Table~\ref{tab:SpErrorsampler}. We note that there are other eigenstates of Hermitian Clifford operators that are not listed in Table~\ref{tab:SpErrorsampler}, which have similar outcome-error correspondences. 

Measurement errors are sampled in a way similar to state preparation errors. If the operation is a measurement, i.e. $\alpha\in \mathbb{O}_{Mea.}$, we can sample the error with the circuit in Fig.~\ref{fig:error_sampler} (f). The operation measures the operator $\kappa$. Under the assumption $\mathbb{O}\subseteq \mathbb{O}_{P}\cup \mathbb{O}_{S}\cup \mathbb{O}_{non-S}$, we prepare the logical qubit in the eigenstate $\ket{\psi}$. If the measurement outcome is $+1$ ($-1$), the output error is $\tau = \openone$ ($\tau = E$). The correspondence between measurement outcomes and the error is illustrated in Table~\ref{tab:SpErrorsampler}. 

In addition to sampling errors in operations, we also need to sample errors in encoding and decoding operations in order to eliminate their impacts when they are not negligible. We can sample encoding and decoding errors with the circuit in Fig.~\ref{fig:error_sampler} (b). The correspondence between measurement outcomes and the error is the same as single-qubit gates, which is illustrated in Table~\ref{tab:SqErrorsampler}. 

\section{Pseudo codes}
\label{app:codes}

\subsection{Generic algorithms for spacetime noise inversion}

\begin{algorithm}[H]
\caption{Spacetime error sampler.}
\label{alg:spacetime_error_sampler}
\begin{algorithmic}[1]
\Statex
\Function{SpacetimeErrorSampler}{$O$,$N$,\textsc{ErrorSampler}}
\Comment $O = (\alpha_1,\alpha_2,\ldots,\alpha_{K})$ is a list of operations, $\alpha_i\in\mathbb{O}-\mathbb{O}_{P}$ for all $i = 1,2,\ldots,K$, $N = (N_1,N_2,\ldots,N_K)$ is a list of operation numbers, and \textsc{ErrorSampler} is an algorithm that outputs errors associated with computational operations (see Algorithms \ref{alg:ideal_error_sampler} and \ref{alg:practical_error_sampler}). 
\State Create a list of variables $\sigma = \left(\sigma(1,\bullet),\sigma(2,\bullet),\ldots,\sigma(K,\bullet)\right)$, and $\sigma(i,\bullet) = \left(\sigma(i,1),\sigma(i,2),\dots \sigma(i,N_i)\right)$ for all $i = 1,2,\ldots,K$. 
\For{$i=1$ to $K$}
    \For{$j=1$ to $N_i$}
        \State $\sigma(i,j) \gets$ \Call{ErrorSampler}{$\alpha_i$}
    \EndFor
\EndFor
\State \Return $\sigma$
\EndFunction
\end{algorithmic}
\end{algorithm}

\begin{algorithm}[H]
\caption{Estimator of the total error rate.}
\label{alg:estimator_of_total_error_rate}
\begin{algorithmic}[1]
\Statex
\Function{TotalErrorRate}{$O$,$N$,$M_P$,\textsc{ErrorSampler}}
\Comment $M_P$ is the number of spacetime error instances for estimating the total error rate. 
\State $M_{error} \gets 0$
\For{$l=1$ to $M_P$}
    \State $\sigma \gets$ \Call{SpacetimeErrorSampler}{$O$,$N$,\textsc{ErrorSampler}}
    \If{there exists $(i,j)$ such that $\sigma(i,j) \neq \openone$}
    \State $M_{error} \gets M_{error}+1$
    \EndIf
\EndFor
\State $\hat{P} = \frac{M_{error}}{M_{P}}$ 
\State \Return $\hat{P}$
\EndFunction
\end{algorithmic}
\end{algorithm}

\begin{algorithm}[H]
\caption{Quantum operation.}
\label{alg:Quantum_operation}
\begin{algorithmic}[1]
\Statex
\Function{QuantumOperation}{$\alpha$}
\Comment $\alpha\in \mathbb{O}$
\State Apply the operation $\alpha$. 
\If{$\alpha\in \mathbb{O}_{Mea.}$}
    \State Record the measurement outcome $\mu$. 
\Else
    \State $\mu \gets 0$ 
\EndIf
\State \Return $\mu$
\EndFunction
\end{algorithmic}
\end{algorithm}

\begin{algorithm}[H]
\caption{Circuit implementation.}
\label{alg:circuit_implementation}
\begin{algorithmic}[1]
\Statex
\Function{CircuitImplementation}{$\lambda$,$C$,$O$,$\sigma$}
\Comment $O = (\alpha_1,\alpha_2,\ldots,\alpha_{K})$ is a list of operations in $\mathbb{O}-\mathbb{O}_{P}$, $K = \vert\mathbb{O}-\mathbb{O}_{P}\vert$, $\sigma = \left(\sigma(1,\bullet),\sigma(2,\bullet),\ldots,\sigma(K,\bullet)\right)$, and each element $\sigma(i,\bullet) = \left(\sigma(i,1),\sigma(i,2),\dots \sigma(i,N_{\alpha_i}^{max})\right)$ is a list of Pauli operators. 

\State $\mathrm{Flag} \gets (1,1,\ldots,1)$
\Comment $\mathrm{Flag}$ is a $K$-tuple. 
\State $\mu_{<1}\gets ()$

\For{$l=1$ to $N$}
\Comment $N$ is the number of operations in $C$. 
    \If{$c_l(\lambda,\mu_{<l})\in\mathbb{O}_{P}$}
    \Comment The operation is a Pauli gate. 
        \State $\mu_l \gets$ \Call{QuantumOperation}{$c_l(\lambda,\mu_{<l})$}
    \Else
        \State $i \gets$ the label of $c_l(\lambda,\mu_{<l})$ in $O$
        \If{$c_l(\lambda,\mu_{<l})\in\mathbb{O}_{Mea.}$}
        \Comment The operation is a measurement. 
            \State $\tau \gets \sigma\left(i,\mathrm{Flag}(i)\right)$
            \State \Call{QuantumOperation}{$\tau$}
            \State $\mu_l \gets$ \Call{QuantumOperation}{$c_l(\lambda,\mu_{<l})$}
            \Comment Apply the error first then the operation. 
        \Else
        \Comment The operation is a gate or state preparation operation but not a Pauli gate. 
            \State $\tau \gets \sigma\left(i,\mathrm{Flag}(i)\right)$
            \State $\mu_l \gets$ \Call{QuantumOperation}{$c_l(\lambda,\mu_{<l})$}
            \State \Call{QuantumOperation}{$\tau$}
            \Comment Apply the operation first then the error. 
        \EndIf
        \State $\mathrm{Flag}(i) \gets \mathrm{Flag}(i) + 1$
    \EndIf
    \State $\mu_{<l+1} \gets (\mu_{<l},\mu_l)$
\EndFor

\State \Return $\mu \gets \mu_{<N+1}$
\EndFunction
\end{algorithmic}
\end{algorithm}

\begin{algorithm}[H]
\caption{Spacetime noise inversion.}
\label{alg:spacetime_noise_inversion}
\begin{algorithmic}[1]
\Statex
\Function{InverseNoiseExpansion}{$w$,$C$,$a$,$M_P$,$M$,\textsc{ErrorSampler},\textsc{ProcessedErrorSampler}}
\Comment $(w,C)$ represents a randomized dynamic circuit, $a$ is a function corresponding to the observable, $(M_P,M)$ are number of spacetime error instances, and \textsc{ProcessedErrorSampler} is an algorithm that outputs processed spacetime errors used in spacetime noise inversion (see Algorithms \ref{alg:ideal_processed_error_sampler} and \ref{alg:practical_processed_error_sampler}). 

\State Generate a list $O = (\alpha_1,\alpha_2,\ldots,\alpha_K)$ of operations in $\mathbb{O}-\mathbb{O}_{P}$, where $K = \vert\mathbb{O}-\mathbb{O}_{P}\vert$. 
\State $N \gets \left(N_{\alpha_1}^{max},N_{\alpha_2}^{max},\ldots,N_{\alpha_K}^{max}\right)$
\Comment Evaluate maximum operation numbers. 

\State $\hat{P} \gets$ \Call{TotalErrorRate}{$O$,$N$,$M_P$,\textsc{ErrorSampler}}
\Comment Estimate the total error rate. 

\State $\gamma \gets \frac{1}{1-2\hat{P}}$
\Comment Evaluate the normalization factor. 

\For{$l = 1$ to $M$}
    \State Generate a non-negative integer $k$ according to the probability mass function $P(k) = (1-2\hat{P})\frac{\hat{P}^k}{(1-\hat{P})^{k+1}}$. 
    \State $\eta_l \gets (-1)^k$
    \State $\sigma \gets$ \Call{ProcessedSampler}{$O$,$N$,$k$}
    \State Generate $\lambda \sim w$.     
    \State $\mu_l \gets$ \Call{CircuitImplementation}{$\lambda$,$C$,$O$,$\sigma$}
\EndFor
\State $\hat{A}_{QEM} \gets \frac{\gamma}{M}\sum_{l=1}^M \eta_l a(\lambda,\mu_l)$
\State \Return $\hat{A}_{QEM}$
\EndFunction
\end{algorithmic}
\end{algorithm}

\subsection{Algorithms for the ideal error sampler}

\begin{algorithm}[H]
\caption{Ideal error sampler.}
\label{alg:ideal_error_sampler}
\begin{algorithmic}[1]
\Statex
\Function{Ideal\_ErrorSampler}{$\alpha$}
\Comment $\alpha\in \mathbb{O}-\mathbb{O}_{P}$. 
\State Generate $\tau \sim \mathcal{N}(\alpha)$
\State \Return $\tau$
\EndFunction
\end{algorithmic}
\end{algorithm}

\begin{algorithm}[H]
\caption{Ideal processed error sampler.}
\label{alg:ideal_processed_error_sampler}
\begin{algorithmic}[1]
\Statex
\Function{Ideal\_ProcessedErrorSampler}{$O$,$N$,$k$}
\If{k=0}
    \State Create a list of variables $\sigma = \left(\sigma(1,\bullet),\sigma(2,\bullet),\ldots,\sigma(K,\bullet)\right)$, and $\sigma(i,\bullet) = \left(\sigma(i,1),\sigma(i,2),\dots \sigma(i,N_i)\right)$ for all $i = 1,2,\ldots,K$; take $\sigma(i,j) = \openone$ for all entries. 
\Else
    \For{$l=1$ to $k$}
        \State Flag $\gets 0$
        \While{Flag is $0$}
            \State $\sigma_l \gets$ \Call{SpacetimeErrorSampler}{$O$,$N$,\textsc{Ideal\_ErrorSampler}}
            \If{there exists $(i,j)$ such that $\sigma(i,j) \neq \openone$}
                \State Flag $\gets 1$
            \EndIf
        \EndWhile
    \EndFor
    \State $\sigma \gets \prod_{j=1}^{k}\sigma_l$
    \Comment{Entry-wise product.}
\EndIf
\State \Return $\sigma$
\EndFunction
\end{algorithmic}
\end{algorithm}

\subsection{Algorithms for the practical error sampler}

\begin{algorithm}[H]
\caption{Practical error sampler.}
\label{alg:practical_error_sampler}
\begin{algorithmic}[1]
\Statex
\Function{Practical\_ErrorSampler}{$\alpha$}
\Comment $\alpha\in(\mathbb{O}-\mathbb{O}_{P})\cup\{\tilde{Q}(\beta):\beta\in\mathbb{O}-\mathbb{O}_{P}\}$; when $\alpha = \tilde{Q}(\beta)$, an instance of the error associated with encoding and decoding operations on the qubit subset $\tilde{Q}(\beta)$ is generated. 
\If{$\alpha\in \mathbb{O}-\mathbb{O}_{P}$}
    \If{$\alpha\in \mathbb{O}_1$}
        \State Run the circuit in Fig.~\ref{fig:error_sampler} (c) for one shot to generate the measurement outcome $\mu$. 
        \Comment{$Q(\alpha) = \{j\}$}
        \State Map $\mu$ to a Pauli operator $\tau$ according to Table~\ref{tab:SqErrorsampler}. 
    \ElsIf{$\alpha\in \mathbb{O}_2$}
        \State Run the circuit in Fig.~\ref{fig:error_sampler} (f) for one shot to generate the measurement outcome $\mu$. 
        \Comment{$Q(\alpha) = \{j_1,j_2\}$}
        \State Map $\mu$ to a Pauli operator $\tau$ according to Table~\ref{tab:TqErrorsampler}. 
    \ElsIf{$\alpha\in \mathbb{O}_{S.P.}$}
        \State Run the circuit in Fig.~\ref{fig:error_sampler} (d) for one shot to generate the measurement outcome $\mu$. 
        \Comment{$Q(\alpha) = \{j\}$}
        \State Map $\mu$ to a Pauli operator $\tau$ according to Table~\ref{tab:SpErrorsampler}. 
    \ElsIf{$\alpha\in \mathbb{O}_{Mea.}$}
        \State Run the circuit in Fig.~\ref{fig:error_sampler} (e) for one shot to generate the measurement outcome $\mu$. 
        \Comment{$Q(\alpha) = \{j\}$}
        \State Map $\mu$ to a Pauli operator $\tau$ according to Table~\ref{tab:SpErrorsampler}. 
    \EndIf
    \State $\tau' \gets \tau$
    \For{$j\in \tilde{Q}(\alpha)-Q(\alpha)$}
        \State Run the circuit in Fig.~\ref{fig:error_sampler} (b) for one shot to generate the measurement outcome $\mu$, taking the delay time as the same as the operation. 
        \State Map $\mu$ to a Pauli operator $\tau$ according to Table~\ref{tab:SqErrorsampler}. 
        \State $\tau' \gets \tau\tau'$
    \EndFor
    \State $\tau \gets \tau'$
    \Comment{The circuits in lines 3–19 are executed in parallel such that the operation $\alpha$ and delay operations are performed simultaneously.}
\ElsIf{$\alpha\in \{\tilde{Q}(\beta):\beta\in\mathbb{O}-\mathbb{O}_{P}\}$}
    \State $\tau' \gets \openone$
    \For{$j\in \alpha$}
        \State Run the circuit in Fig.~\ref{fig:error_sampler} (b) for one shot to generate the measurement outcome $\mu$, taking the delay time as zero in the circuit. 
        \State Map $\mu$ to a Pauli operator $\tau$ according to Table~\ref{tab:SqErrorsampler}. 
        \State $\tau' \gets \tau\tau'$
    \EndFor
    \State $\tau \gets \tau'$
    \Comment{The circuits in lines 22–26 are executed in parallel.}
\EndIf
\State \Return $\tau$
\EndFunction
\end{algorithmic}
\end{algorithm}

\begin{algorithm}[H]
\caption{Practical processed error sampler.}
\label{alg:practical_processed_error_sampler}
\begin{algorithmic}[1]
\Statex
\Function{Practical\_ProcessedErrorSampler}{$O$,$N$,$k$}
\Comment $O = (\alpha_1,\alpha_2,\ldots,\alpha_{K})$ is a list of operations, $\alpha_i\in\mathbb{O}-\mathbb{O}_{P}$ for all $i = 1,2,\ldots,K$ and $N = (N_1,N_2,\ldots,N_K)$ is a list of operation numbers. 
\State Create a list of variables $\sigma = \left(\sigma(1,\bullet),\sigma(2,\bullet),\ldots,\sigma(K,\bullet)\right)$, and $\sigma(i,\bullet) = \left(\sigma(i,1),\sigma(i,2),\dots \sigma(i,N_i)\right)$ for all $i = 1,2,\ldots,K$. 

\For{$i=1$ to $K$}
    \If{$\alpha_i\notin \mathbb{O}_{Mea.}$}
    \Comment We only sample encoding and decoding errors for gate and state preparation operations. 
        \For{$j=1$ to $N_i$}
            \State $\sigma(i,j) \gets$ \Call{Practical\_ErrorSampler}{$\tilde{Q}(\alpha_i)$}
        \EndFor
    \Else
        \For{$j=1$ to $N_i$}
            \State $\sigma(i,j) \gets \openone$
        \EndFor
    \EndIf
\EndFor

\If{$k\neq 0$}
    \For{$l=1$ to $k$}
        \State Flag $\gets 0$
        \While{Flag is $0$}
            \State $\sigma_l \gets$ \Call{SpacetimeErrorSampler}{$O$,$N$,\textsc{Practical\_ErrorSampler}}
            \If{there exists $(i,j)$ such that $\sigma(i,j) \neq \openone$}
                \State Flag $\gets 1$
            \EndIf
        \EndWhile
    \EndFor
    \State $\sigma \gets \sigma\prod_{j=1}^{k_{S}}\sigma_l$
    \Comment{Entry-wise product.}
\EndIf
\State \Return $\sigma$
\EndFunction
\end{algorithmic}
\end{algorithm}

\section{Error and sampling cost with an ideal error sampler---Proof of Theorem~1}
\label{app:error_cost_ideal}

\begin{theorem}
{\bf Formal version of Theorem~1 in the main text.} Suppose the following conditions hold: 
\begin{itemize}
\item Errors are temporally uncorrelated, i.e.~final states of noisy circuits satisfy Eq.~(\ref{eq:noisy_circuit}), in which the map of each noisy operation is given by Eq.~(\ref{eq:noisy_operation}); 
\item Pauli gates are error-free; 
\item Errors in all other operations are Pauli (Definition~\ref{def:Pauli_noise}). 
\end{itemize}
Assume the existence of an ideal error sampler and apply the protocol in Sec.~\ref{app:ideal_protocol} to an arbitrary randomized dynamic circuit (Definition~\ref{def:RDC}) generated by an arbitrary operation set $\mathbb{O}$. 
Define the following notations: 
\begin{itemize}
\item Let $P$ denote the total error rate of the maximum spacetime noise, and let $\hat{P}$ be its estimate; 
\item Let $M$ be the number of circuit runs used to evaluate $\hat{A}_{QEM}$, and let $M_P$ be the number of spacetime error instances used to estimate the total error rate; 
\item Let $M_{es}$ be the total number of spacetime error instances generated from the error sampler. 
\end{itemize}
For the error-mitigated estimator, the bias has the upper bound 
\begin{eqnarray}
\left\vert\mathrm{E}\left[\hat{A}_{QEM}\vert\hat{P}\right] - \mean{A}_I\right\vert &\leq & \norm{a}_{L^\infty} \left\vert\frac{1}{1-2\hat{P}} - \frac{1}{1-2P}\right\vert.~
\label{eq:bias}
\end{eqnarray}
When $P<1/2$, for any positive numbers $\delta$ and $f$, the error $\vert\hat{A}_{QEM}-\mean{A}_I\vert$ is smaller than $\delta\norm{a}_{L^\infty}$ with a probability at least $1-f$ under conditions 
\begin{eqnarray}
M_P\geq \frac{1}{2t_P^2}\ln\frac{4}{f}
\label{eq:MP_bound}
\end{eqnarray}
and 
\begin{eqnarray}
M\geq \frac{8}{\delta^2(1-2P-2t_P)^2}\ln\frac{4}{f},
\label{eq:M_bound}
\end{eqnarray}
where 
\begin{eqnarray}
t_P = \min\left\{\frac{\delta(1-2P)^2}{4+2\delta(1-2P)},\frac{1}{2}-P\right\}.
\end{eqnarray}
The expected value and variance of the sampling cost are given by: 
\begin{eqnarray}
\mathrm{E}[M_{es}\vert\hat{P}] &=& M_P + \frac{M\hat{P}}{P(1-2\hat{P})}, \label{eq:costE} \\
\mathrm{Var}(M_{es}\vert\hat{P}) &=& \frac{M\hat{P}(2-P-3\hat{P}+2P\hat{P})}{P^2(1-2\hat{P})^2}. \label{eq:costV}
\end{eqnarray}
\label{the:error_cost_ideal}
\end{theorem}

\subsection{Error in $\hat{P}$}

In the estimation of the maximum total error rate $P$, the bias is 
\begin{eqnarray}
\mathrm{E}[\hat{P}] - P = 0.
\end{eqnarray}
According to Hoeffding's inequality, when $M_P\geq \frac{1}{2t_P^2}\ln\frac{4}{f}$, 
\begin{eqnarray}
\mathrm{Pro}\left(\vert\hat{P}-P\vert \geq t_P\right) \leq 2\exp\left(-2t_p^2M_P\right) \leq \frac{f}{2}.
\end{eqnarray}

\subsection{Error in $\hat{A}_{QEM}$}

Since Pauli gates are error-free, the maximum spacetime noise map is in the form 
\begin{eqnarray}
\mathcal{N}_{max} = \mathcal{N}_{P}\otimes\mathcal{N}_{non-P},
\label{eq:NPNnonP}
\end{eqnarray}
where 
\begin{eqnarray}
\mathcal{N}_{P} = \bigotimes_{\alpha\in\mathbb{O}_{P}} \mathcal{N}(\alpha)^{\otimes N_\alpha^{max}} = [\openone_{P}] \equiv [\openone]^{\otimes N_{P}}
\end{eqnarray}
is an identity map, $N_{P} = \sum_{\alpha\in\mathbb{O}_{P}} N^{max}_\alpha$, and 
\begin{eqnarray}
\mathcal{N}_{non-P} = \bigotimes_{\alpha\in\mathbb{O}-\mathbb{O}_{P}} \mathcal{N}(\alpha)^{\otimes N_\alpha^{max}}.
\label{eq:NnonP}
\end{eqnarray}
Because the noise is Pauli, the noise map $\mathcal{N}_{non-P}$ is in the form 
\begin{eqnarray}
\mathcal{N}_{non-P} = \sum_{\tau\in \mathbb{P}_n^{\otimes N_{non-P}}} \epsilon(\tau)[\tau],
\end{eqnarray}
where $N_{non-P} = \sum_{\alpha\in\mathbb{O}-\mathbb{O}_{P}} N^{max}_\alpha$, and $\epsilon(\tau)$ is the rate of the Pauli error $[\tau]$. Similar to the noise map of an operation, we can rewrite the map as 
\begin{eqnarray}
\mathcal{N}_{non-P} = (1-P)[\openone_{non-P}] + P\mathcal{E}_{non-P},
\end{eqnarray}
where $[\openone_{non-P}] \equiv [\openone_{non-P}]^{\otimes N_{non-P}}$, 
\begin{eqnarray}
\mathcal{E}_{non-P} = \frac{1}{P}\sum_{\tau\in \mathbb{P}_n^{\otimes N_{non-P}}-\{\openone_{non-P}\}} \epsilon(\tau)[\tau]
\end{eqnarray}
is a map denoting spacetime errors, and $P = \sum_{\tau\in \mathbb{P}_n^{\otimes N_{non-P}}-\{\openone_{non-P}\}} \epsilon(\tau)$ is the maximum total error rate of the circuit. 

\begin{definition}
{\bf $L^1_{Pauli}$-norm.} For a Pauli noise map in the form 
\begin{eqnarray}
\mathcal{K} = \sum_{\tau\in \mathbb{P}_n} \nu(\tau)[\tau],
\end{eqnarray}
where $\nu(\tau)$ is a real-valued function, its $L^1_{Pauli}$-norm reads 
\begin{eqnarray}
\norm{\mathcal{K}}_{L^1_{Pauli}} = \sum_{\tau\in \mathbb{P}_n} \abs{\nu(\tau)}. 
\end{eqnarray}
\end{definition}

We note that the $L^1_{\text{Pauli}}$-norm is equivalent to the diamond norm~\cite{PhysRevA.85.042311}, which is defined by
\begin{eqnarray}
\|\mathcal{K}\|_\diamond = \max_{X \,:\, \|X\|_1 \leq 1} \left\| \left( \mathcal{K} \otimes \openone \right) X \right\|_1.
\end{eqnarray}

\begin{lemma}
$L^1_{Pauli}$-norm is submultiplicative. Let $\mathcal{K}_1$ and $\mathcal{K}_2$ be two Pauli noise maps, then 
\begin{eqnarray}
\norm{\mathcal{K}_1\mathcal{K}_2}_{L^1_{Pauli}} \leq \norm{\mathcal{K}_1}_{L^1_{Pauli}} \norm{\mathcal{K}_2}_{L^1_{Pauli}}. 
\end{eqnarray}
\end{lemma}

\begin{proof}
While this lemma can be easily shown using the equivalence between the $L^1_{Pauli}$-norm and the diamond norm, we instead provide a proof directly from the definition of the $L^1_{Pauli}$-norm.
Let $\nu_1(\tau)$ and $\nu_2(\tau)$ be coefficient functions of $\mathcal{K}_1$ and $\mathcal{K}_2$, respectively. Then, 
\begin{eqnarray}
\mathcal{K}_1\mathcal{K}_2 = \sum_{\tau,\tau_1,\tau_2\in \mathbb{P}_n} \delta_{[\tau],[\tau_1][\tau_2]}\nu_1(\tau_1)\nu_2(\tau_2)[\tau].
\end{eqnarray}
Its norm is 
\begin{eqnarray}
\norm{\mathcal{K}_1\mathcal{K}_2}_{L^1_{Pauli}} &=& \sum_{\tau\in \mathbb{P}_n}\left\vert\sum_{\tau_1,\tau_2\in \mathbb{P}_n}\delta_{[\tau],[\tau_1][\tau_2]}\nu_1(\tau_1)\nu_2(\tau_2)\right\vert \notag \\
&\leq & \sum_{\tau,\tau_1,\tau_2\in \mathbb{P}_n}\delta_{[\tau],[\tau_1][\tau_2]}\vert\nu_1(\tau_1)\vert\vert\nu_2(\tau_2)\vert = \norm{\mathcal{K}_1}_{L^1_{Pauli}} \norm{\mathcal{K}_2}_{L^1_{Pauli}}.
\end{eqnarray}
\end{proof}

With the ideal sampler, the expected value of $\hat{A}_{QEM}$, the output of Algorithm~\ref{alg:spacetime_noise_inversion}, is 
\begin{eqnarray}
\mathrm{E}\left[\hat{A}_{QEM}\vert\hat{P}\right] &=& \sum_{k = 0}^\infty (-1)^k \frac{\hat{P}^k}{(1-\hat{P})^{k+1}} \sum_{\lambda,\mu} w(\lambda)a(\lambda,\mu) \Tr F\left(\lambda,\mu,([\openone_{P}]\otimes\mathcal{E}_{non-P}^k)\mathcal{N}_{max}\right) \notag \\
&=& \sum_{\lambda,\mu} w(\lambda)a(\lambda,\mu) \Tr F\left(\lambda,\mu,[\openone_{P}]\otimes(\hat{\mathcal{N}}^{-1}\mathcal{N}_{non-P})\right),
\end{eqnarray}
where 
\begin{eqnarray}
\hat{\mathcal{N}} = (1-\hat{P})[\openone_{non-P}] + \hat{P}\mathcal{E}_{non-P}.
\end{eqnarray}
The ideal expected value of the observable is 
\begin{eqnarray}
\mean{A}_I &=& \sum_{\lambda,\mu} w(\lambda)a(\lambda,\mu) \Tr F\left(\lambda,\mu,[\openone_{P}]\otimes(\mathcal{N}_{non-P}^{-1}\mathcal{N}_{non-P})\right).
\end{eqnarray}
Therefore, the bias of $\hat{A}_{QEM}$ has the upper bound 
\begin{eqnarray}
\left\vert\mathrm{E}\left[\hat{A}_{QEM}\vert\hat{P}\right] - \mean{A}_I\right\vert &\leq & \norm{a}_{L^\infty} \norm{\hat{\mathcal{N}}^{-1}-\mathcal{N}_{non-P}^{-1}}_{L^1_{Pauli}},
\label{eq:bias1}
\end{eqnarray}
where $\norm{a}_{L^\infty} = \max_{\lambda,\mu}\abs{a(\lambda,\mu)}$. Here, we have used that for all Pauli operators $\tau\in \mathbb{P}_n^{\otimes N_{non-P}}$, the factor $\Tr F\left(\lambda,\mu,[\openone_{P}]\otimes([\tau]\mathcal{N}_{non-P})\right)$ (the trace of the output state when Pauli gates are inserted into the circuit according to $\tau$) is a normalized weight function of $\mu$; then, 
\begin{eqnarray}
\left\vert\sum_{\lambda,\mu} w(\lambda)a(\lambda,\mu) \Tr F\left(\lambda,\mu,[\openone_{P}]\otimes([\tau]\mathcal{N}_{non-P})\right)\right\vert \leq \norm{a}_{L^\infty},
\end{eqnarray}
which holds even if $\mathcal{N}_{non-P}$ is not Pauli noise. 

Using expansions, the difference between two inverse maps is 
\begin{eqnarray}
\hat{\mathcal{N}}^{-1}-\mathcal{N}_{non-P}^{-1} = \sum_{k = 0}^\infty (-1)^k \left[\frac{\hat{P}^k}{(1-\hat{P})^{k+1}} - \frac{P^k}{(1-P)^{k+1}}\right] \mathcal{E}_{non-P}^k.
\label{eq:inverse_diff}
\end{eqnarray}
Its $L^1_{Pauli}$-norm has the upper bound  
\begin{eqnarray}
\norm{\hat{\mathcal{N}}^{-1}-\mathcal{N}_{non-P}^{-1}}_{L^1_{Pauli}} &\leq & \sum_{k = 0}^\infty \left\vert\frac{\hat{P}^k}{(1-\hat{P})^{k+1}} - \frac{P^k}{(1-P)^{k+1}}\right\vert \notag \\
&=& \left\vert\frac{1}{1-2\hat{P}} - \frac{1}{1-2P}\right\vert.
\end{eqnarray}
Here, we have used that $\norm{\mathcal{E}_{non-P}}_{L^1_{Pauli}} = 1$. Substitute the upper bound of $\norm{\hat{\mathcal{N}}^{-1}-\mathcal{N}_{non-P}^{-1}}_{L^1_{Pauli}}$ into inequality (\ref{eq:bias1}), we obtain the following upper bound of the bias: 
\begin{eqnarray}
\left\vert\mathrm{E}\left[\hat{A}_{QEM}\vert\hat{P}\right] - \mean{A}_I\right\vert &\leq & \norm{a}_{L^\infty} \left\vert\frac{1}{1-2\hat{P}} - \frac{1}{1-2P}\right\vert.
\label{eq:bias2}
\end{eqnarray}

When $\vert\hat{P}-P\vert < t_P$, where 
\begin{eqnarray}
t_P = \min\left\{\frac{\delta(1-2P)^2}{4+2\delta(1-2P)},\frac{1}{2}-P\right\},
\end{eqnarray}
the bias upper bound becomes 
\begin{eqnarray}
\left\vert\mathrm{E}\left[\hat{A}_{QEM}\vert\hat{P}\right] - \mean{A}_I\right\vert &\leq & \norm{a}_{L^\infty} \frac{2\vert\hat{P}-P\vert}{(1-2P)(1-2P-2\vert\hat{P}-P\vert)} < \frac{\delta}{2}\norm{a}_{L^\infty}. 
\end{eqnarray}
According to Hoeffding's inequality, when $\vert\hat{P}-P\vert < t_P$ and $M\geq \frac{8}{\delta^2(1-2P-2t_P)^2}\ln\frac{4}{f}$, 
\begin{eqnarray}
\mathrm{Pro}\left(\vert\hat{A}_{QEM} - \mathrm{E}\left[\hat{A}_{QEM}\vert\hat{P}\right]\vert \geq \frac{\delta}{2}\norm{a}_{L^\infty}\right) \leq 2\exp\left(-\frac{\delta^2(1-2P-2t_P)^2M}{8}\right) \leq \frac{f}{2},
\end{eqnarray}
where the factor $1-2P-2t_P$ is due to the normalization factor $\gamma$ in the algorithm. Therefore, 
\begin{eqnarray}
\mathrm{Pro}\left(\vert\hat{A}_{QEM} - \mean{A}_I\vert \geq \delta\norm{a}_{L^\infty}\right) \leq f.
\end{eqnarray}

\subsection{Expected value and variance of the sampling cost}

To implement SNI, we have to generate a number of spacetime errors from the error sampler. In the following, we refer to $\openone_{non-P}$ as the trivial error and $\mathcal{E}_{non-P}$ as the non-trivial error. The probability of obtaining $k$ non-trivial spacetime errors in $m$ (trivial and non-trivial) spacetime errors (calling Algorithm~\ref{alg:spacetime_error_sampler} for $m$ times in lines from 5 to 10 in Algorithm~\ref{alg:ideal_processed_error_sampler}) is 
\begin{eqnarray}
\mathrm{Pro}(m\vert k) = {{m-1}\choose{k-1}} P^{k}(1-P)^{m-k},
\end{eqnarray}
according to the negative binomial distribution. Taking into account the distribution of $k$, the joint distribution is 
\begin{eqnarray}
\mathrm{Pro}(k,m) = \mathrm{Pro}(k) \mathrm{Pro}(m\vert k),
\end{eqnarray}
where 
\begin{eqnarray}
\mathrm{Pro}(k) = (1-2\hat{P})\frac{\hat{P}^k}{(1-\hat{P})^{k+1}}.
\end{eqnarray}
Then, the expected value and variance of $m$ are 
\begin{eqnarray}
\mathrm{E}[m\vert\hat{P}] = \sum_{k=0}^\infty \mathrm{Pro}(k) \frac{k}{P} = \frac{\hat{P}}{P(1-2\hat{P})}
\end{eqnarray}
and 
\begin{eqnarray}
\mathrm{Var}(m\vert\hat{P}) &=& \sum_{k=0}^\infty \mathrm{Pro}(k) \left[\frac{k(1-P)}{P^2} + \frac{k^2}{P^2}\right] - \mathrm{E}[m]^2 \notag \\
&=& \frac{(1-P)\hat{P}}{P^2(1-2\hat{P})} + \frac{\hat{P}}{P^2(1-2\hat{P})^2} - \left[\frac{\hat{P}}{P(1-2\hat{P})}\right]^2 \notag \\
&=& \frac{\hat{P}}{P^2(1-2\hat{P})}\left(2-P + \frac{\hat{P}}{1-2\hat{P}}\right),
\end{eqnarray}
respectively. Let $M_{es}$ be the total number of spacetime errors generated from the error sampler. It consists of the $M_P$ instances for evaluating the maximum total error rate and instances used for generating the $M$ circuit runs. Then, the expected value and variance of $M_{es}$ are 
\begin{eqnarray}
\mathrm{E}[M_{es}\vert\hat{P}] = M_P + M\mathrm{E}[m\vert\hat{P}] = M_P + M\frac{\hat{P}}{P(1-2\hat{P})}
\end{eqnarray}
and 
\begin{eqnarray}
\mathrm{Var}(M_{es}\vert\hat{P}) &=& M\mathrm{Var}(m\vert\hat{P}) = M\frac{\hat{P}}{P^2(1-2\hat{P})}\left(2-P + \frac{\hat{P}}{1-2\hat{P}}\right),
\end{eqnarray}
respectively. 

When $\vert\hat{P}-P\vert < t_P$, we have upper bounds 
\begin{eqnarray}
\mathrm{E}[M_{es}\vert\hat{P}] &\leq & M_P + \frac{M(P+t_P)}{P(1-2P-2t_P)}, \\
\mathrm{Var}(M_{es}\vert\hat{P}) &\leq & M\frac{P+t_P}{P^2(1-2P-2t_P)}\left(2-P + \frac{P+t_P}{1-2P-2t_P}\right).
\end{eqnarray}
According to the Chebyshev's inequality, the probability that $M_{es}$ exceeds $\mathrm{E}[M_{es}] + \kappa \sqrt{\mathrm{Var}(M_{es})}$ is 
\begin{eqnarray}
\mathrm{Pro}\left(M_{es}\geq \mathrm{E}[M_{es}\vert\hat{P}] + s \sqrt{\mathrm{Var}(M_{es}\vert\hat{P})}\right) \leq \frac{1}{s}.
\end{eqnarray}
Approximating the distribution of $M_{es}$ using the normal distribution, the probability becomes 
\begin{eqnarray}
\mathrm{Pro}\left(M_{es}\geq \mathrm{E}[M_{es}\vert\hat{P}] + s \sqrt{\mathrm{Var}(M_{es}\vert\hat{P})}\right) \lesssim e^{-\frac{s^2}{2}}.
\end{eqnarray}

\section{Error and sampling cost with a practical error sampler}
\label{app:error_cost_practical}

\begin{corollary}
Suppose the following conditions hold: 
\begin{itemize}
\item Errors are temporally uncorrelated, i.e.~final states of noisy circuits satisfy Eq.~(\ref{eq:noisy_circuit}), in which the map of each noisy operation is given by Eq.~(\ref{eq:noisy_operation}); 
\item Pauli gates are error-free; 
\item Super qubits are error-free; 
\end{itemize}
We remark that non-Pauli operations on logical qubits, as well as encoding and decoding operations, may introduce non-Pauli errors. 
Apply the practical-sampler protocol in Sec.~\ref{app:practical_protocol} to an arbitrary twirled circuit (Definition~\ref{def:twirled_circuit}) generated by an arbitrary operation set $\mathbb{O} \subseteq \mathbb{O}_{P} \cup \mathbb{O}_{S} \cup \mathbb{O}_{non-S}$. 
Define the following notations: 
\begin{itemize}
\item Let $P$ be the total error rate of the maximum spacetime noise after accounting for twirling and inserted encoding and decoding errors [defined by Eq.~(\ref{eq:N_eff_nonP})]; 
\item Other notations are the same as in Theorem~\ref{the:error_cost_ideal}. 
\end{itemize}
Then, bounds in Eqs. (\ref{eq:bias}), (\ref{eq:MP_bound}), and (\ref{eq:M_bound}) and cost estimators in Eqs.~(\ref{eq:costE}) and (\ref{eq:costV}) hold. 
\label{cor:error_cost_practical}
\end{corollary}

Except for the bias of $\hat{A}_{QEM}$, all other analyses in Sec.~\ref{app:error_cost_ideal} can be directly applied to the practical-sampler protocol. To analyze the bias of $\hat{A}_{QEM}$ in the practical-sampler protocol, we have to solve two issues. First, the errors in operations could be non-Pauli under the practical-sampler assumption. Since we apply Pauli twirling to non-Pauli stabilizer operations $\mathbb{O}_{S}$, errors in these operations are effectively Pauli (Notice that Pauli gates are error-free). However, although we also apply twirling to non-stabilizer operations $\mathbb{O}_{non-S}$, the post-twirling non-stabilizer operations may have non-Pauli errors because stabilizer operations used in twirling may not be error-free. Second, in the practical error sampler, errors are generated by the computing operations and encoding/decoding operations. Therefore, we need to take into account errors in encoding/decoding operations. In what follows, we prove the corollary. 

\subsection{Effective noise in the error sampler}

Let $\mathcal{N}_{en/de,j}$ be the noise maps associated with encoding/decoding operations on the $j$th qubit, respectively. For a single-qubit gate $\alpha$, the effective operation determining the distribution of errors is $\mathcal{N}_{en,j}\mathcal{N}(\alpha)\mathcal{M}^I(\alpha,\mu)\mathcal{N}_{de,j}$; see Fig.~\ref{fig:error_sampler}(b). Accordingly, the noise map of the effective operation is 
\begin{eqnarray}
\mathcal{N}_{eff}(\alpha) = \mathcal{N}_{en,j}\mathcal{N}(\alpha)\mathcal{M}^I(\alpha,\mu)\mathcal{N}_{de,j}{\mathcal{M}^I(\alpha,\mu)}^{-1}.
\end{eqnarray}
Similarly, effective noise maps of two-qubit gates [Fig.~\ref{fig:error_sampler}(c)], state preparations [Fig.~\ref{fig:error_sampler}(d)] and measurements [Fig.~\ref{fig:error_sampler}(e)] are 
\begin{eqnarray}
\mathcal{N}_{eff}(\alpha) &=& \mathcal{N}_{en,j_1}\mathcal{N}_{en,j_2}\mathcal{N}(\alpha)\mathcal{M}^I(\alpha,\mu)\mathcal{N}_{de,j_1}\mathcal{N}_{de,j_2}{\mathcal{M}^I(\alpha,\mu)}^{-1},
\end{eqnarray}
\begin{eqnarray}
\mathcal{N}_{eff}(\alpha) &=& \mathcal{N}_{en,j}\mathcal{N}(\alpha)
\end{eqnarray}
and 
\begin{eqnarray}
\mathcal{N}_{eff}(\alpha) &=& \mathcal{N}(\alpha)\mathcal{N}_{en,j},
\end{eqnarray}
respectively. 

Since we apply Pauli twirling to encoding, decoding and non-Pauli stabilizer operations, the effective noise of a non-Pauli stabilizer operation is Pauli, i.e. $\mathcal{N}_{eff}(\alpha)$ is Pauli if $\alpha\in \mathbb{O}_{S}$. However, if $\alpha\in \mathbb{O}_{non-S}$, $\mathcal{N}_{eff}(\alpha)$ may not be Pauli. Even for such a non-stabilizer operation, corresponding errors are still generated according to a Pauli noise map. To see this, let us consider a non-Clifford single-qubit gate; see Fig.~\ref{fig:error_sampler}(b). If we randomly choose $V = \openone,X_aX_j,Y_aY_j,Z_aZ_j$ and apply it after the Bell state preparation and before the Bell measurement, we effectively realize the twirling $\mathcal{N}_{eff}(\alpha)$. Because $V$ is in the stabilizer group of the Bell state, applying $V$ does not change the distribution of measurement outcomes. Therefore, the error distribution follows the Pauli noise map $\overline{\mathcal{N}}_{eff}(\alpha)$, which is the post-twirling noise map of $\mathcal{N}_{eff}(\alpha)$. Similarly, error distributions of non-stabilizer state preparations and measurements also follow corresponding post-twirling noise maps. 

Overall, the maximum spacetime noise map corresponding to the practical error sampler is 
\begin{eqnarray}
\overline{\mathcal{N}}_{eff}^{max} = \bigotimes_{\alpha\in\mathbb{O}} \overline{\mathcal{N}}_{eff}(\alpha)^{\otimes N_\alpha^{max}},
\end{eqnarray}
which is Pauli. Notice that $\overline{\mathcal{N}}_{eff}(\alpha) = \mathcal{N}_{eff}(\alpha)$ if $\alpha$ is a non-Pauli stabilizer operation. Similar to the ideal-sampler protocol, we can rewrite $\overline{\mathcal{N}}_{eff}^{max}$ as 
\begin{eqnarray}
\overline{\mathcal{N}}_{eff}^{max} = [\openone_{P}]\otimes\overline{\mathcal{N}}_{eff,non-P},
\end{eqnarray}
where 
\begin{eqnarray}
\overline{\mathcal{N}}_{eff,non-P} &=& \bigotimes_{\alpha\in\mathbb{O}-\mathbb{O}_{P}} \overline{\mathcal{N}}_{eff}(\alpha)^{\otimes N_\alpha^{max}} \notag \\
&=& (1-P)[\openone_{non-P}] + P\mathcal{E}_{eff,non-P}.
\label{eq:N_eff_nonP}
\end{eqnarray}

\subsection{Effective noise in the circuit}

\begin{figure}[htbp]
\centering
\includegraphics[width=0.5\linewidth]{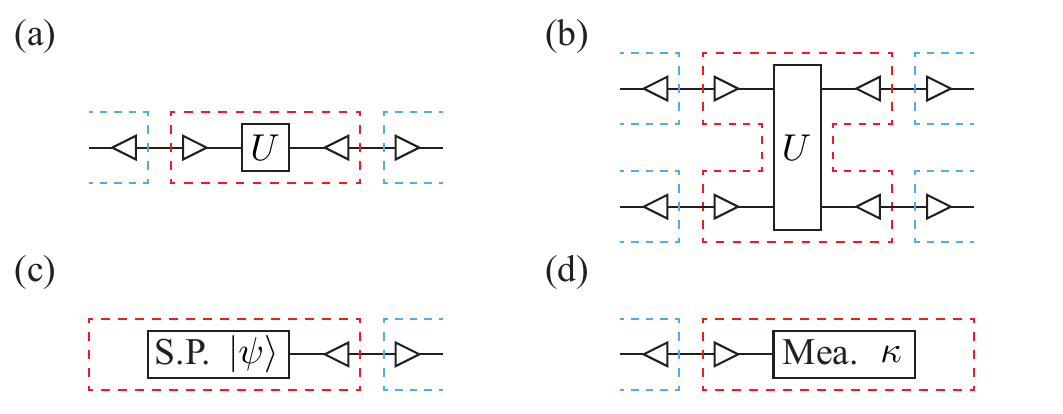}
\caption{
Effective operations after inserting encoding/decoding errors. Triangles represent encoding/decoding noise maps, respectively; see Fig.~\ref{fig:error_sampler}. 
}
\label{fig:en_de_errors}
\end{figure}

In the circuit, we insert Pauli errors after state preparations and non-Pauli gates according to errors generated from encoding/decoding operations (lines from 3 to 9 in Algorithm~\ref{alg:practical_processed_error_sampler}). The encoding/decoding errors are generated according to the noise map $\mathcal{N}_{en,j}\mathcal{N}_{de,j}$. Because $\mathcal{N}_{en,j}$ and $\mathcal{N}_{de,j}$ commute with each other (they are Pauli), the noise inserted into the circuit is effectively $\mathcal{N}_{de,j}\mathcal{N}_{en,j}$ as shown in Fig.~\ref{fig:en_de_errors}. 

For a single-qubit gate, the operation before the gate must be a state preparation or another gate; see Fig.~\ref{fig:en_de_errors}(a). If the previous operation is not a Pauli gate, the encoding/decoding noise $\mathcal{N}_{de,j}\mathcal{N}_{en,j}$ is inserted before the gate; the noise is also inserted after the gate. Then, we can find that the decoding noise before the gate, the gate itself and the encoding noise after the gate constitute the same effective operation in the error sampler, i.e. the effective noise of the gate is $\mathcal{N}_{eff}(\alpha)$. If the previous operation is a Pauli gate, we look for the closest non-Pauli operation before the gate, which contributes the decoding noise before the gate to the effective noise; notice that the decoding noise commutes with Pauli gates. It is similar for two-qubit gates, state preparations and measurements, as shown in Fig.~\ref{fig:en_de_errors}(b), (c) and (d), respectively. 

All encoding/decoding noises inserted into the circuit are taken into account in corresponding effective noises of operations. For a single-qubit gate, the encoding (decoding) noise before (after) the gate is taken into account in the operation before (after) the gate. It is similar for two-qubit gates, state preparations and measurements. 

After inserting encoding/decoding noise into the circuit, the maximum spacetime noise map becomes 
\begin{eqnarray}
\mathcal{N}_{eff}^{max} &=& \bigotimes_{\alpha\in\mathbb{O}} \mathcal{N}_{eff}(\alpha)^{\otimes N_\alpha^{max}} \notag \\
&=& [\openone_{P}]\otimes\mathcal{N}_{eff,non-P},
\end{eqnarray}
where 
\begin{eqnarray}
\mathcal{N}_{eff,non-P} = \bigotimes_{\alpha\in\mathbb{O}-\mathbb{O}_{P}} \mathcal{N}_{eff}(\alpha)^{\otimes N_\alpha^{max}}.
\end{eqnarray}

\subsection{Bias of $\hat{A}_{QEM}$ in the practical-sampler protocol}

With the practical sampler, the expected value of $\hat{A}_{QEM}$ is 
\begin{eqnarray}
\mathrm{E}\left[\hat{A}_{QEM}\vert\hat{P}\right] = \sum_{\lambda,\mu} w(\lambda)a(\lambda,\mu) \Tr F\left(\lambda,\mu,[\openone_{P}]\otimes(\hat{\mathcal{N}}_{eff}^{-1}\mathcal{N}_{eff,non-P})\right),
\end{eqnarray}
where 
\begin{eqnarray}
\hat{\mathcal{N}}_{eff}' = (1-\hat{P})[\openone_{non-P}] + \hat{P}\mathcal{E}_{eff,non-P}.
\end{eqnarray}
In Sec.~\ref{app:twirling}, we prove that 
\begin{eqnarray}
\mean{A}_I &=& \sum_{\lambda,\mu} w(\lambda)a(\lambda,\mu) \Tr F\left(\lambda,\mu,[\openone_{P}]\otimes(\overline{\mathcal{N}}_{eff,non-P}^{-1}\mathcal{N}_{eff,non-P})\right).
\end{eqnarray}
Therefore, the bias upper bound in Eq. (\ref{eq:bias2}) applies to the practical-sampler protocol. 

\subsection{Twirling on non-stabilizer operations}
\label{app:twirling}

We can rewrite the noise map as follows: 
\begin{eqnarray}
\mathcal{N}_{eff,non-P} = \mathcal{N}_{eff,S}\otimes\mathcal{N}_{eff,non-S},
\end{eqnarray}
where 
\begin{eqnarray}
\mathcal{N}_{eff,S/non-S} &=& \bigotimes_{\alpha\in\mathbb{O}\cap\mathbb{O}_{S/non-S}} \mathcal{N}_{eff}(\alpha)^{\otimes N_\alpha^{max}}.
\end{eqnarray}
Then, we have 
\begin{eqnarray}
&& \sum_{\lambda,\mu} w(\lambda)a(\lambda,\mu) \Tr F\left(\lambda,\mu,[\openone_{P}]\otimes(\overline{\mathcal{N}}_{eff,non-P}^{-1}\mathcal{N}_{eff,non-P})\right) \notag \\
&=& \sum_{\lambda,\mu} w(\lambda)a(\lambda,\mu) \Tr F\left(\lambda,\mu,[\openone_{P}]\otimes(\overline{\mathcal{N}}_{eff,S}^{-1}\mathcal{N}_{eff,S})\otimes(\overline{\mathcal{N}}_{eff,non-S}^{-1}\mathcal{N}_{eff,non-S})\right) \notag \\
&=& \sum_{\lambda,\mu} w(\lambda)a(\lambda,\mu) \Tr F\left(\lambda,\mu,[\openone_{P}]\otimes[\openone_{S}]\otimes(\overline{\mathcal{N}}_{eff,non-S}^{-1}\mathcal{N}_{eff,non-S})\right),
\end{eqnarray}
where $[\openone_{S}] \equiv [\openone_{S}]^{\otimes N_{S}}$ denotes that non-Pauli stabilizer operations are error-free, and $N_{S} = \sum_{\alpha\in\mathbb{O}\cap\mathbb{O}_{S}} N^{max}_\alpha$. Here, we have used that $\mathcal{N}_{eff,S} = \overline{\mathcal{N}}_{eff,S}$ (because Pauli twirling is applied on non-Pauli stabilizer operations). Therefore, Pauli gates and non-Pauli stabilizer operations are effective error-free. 

Because Pauli gates and non-Pauli stabilizer operations are effective error-free, they can realize ideal twirling on non-stabilizer operations, i.e. 
\begin{eqnarray}
&& \sum_{\lambda,\mu} w(\lambda)a(\lambda,\mu) \Tr F\left(\lambda,\mu,[\openone_{P}]\otimes[\openone_{S}]\otimes(\overline{\mathcal{N}}_{eff,non-S}^{-1}\mathcal{N}_{eff,non-S})\right) \notag \\
&=& \sum_{\lambda,\mu} w(\lambda)a(\lambda,\mu) \Tr F\left(\lambda,\mu,[\openone_{P}]\otimes[\openone_{S}]\otimes(\overline{\mathcal{N}}_{eff,non-S}^{-1}\overline{\mathcal{N}}_{eff,non-S})\right) \notag \\
&=& \sum_{\lambda,\mu} w(\lambda)a(\lambda,\mu) \Tr F\left(\lambda,\mu,[\openone_{P}]\otimes[\openone_{non-P}]\right) \notag \\
&=& \mean{A}_I.
\end{eqnarray}

\section{Rigorous results regarding temporally correlated errors}
\label{app:correlations}

The bounds and cost estimators in Theorem~\ref{the:error_cost_ideal} remain valid in the presence of temporally correlated errors, provided these are captured by the average noise channel $\mathcal{N}_{ave}$; see Corollary~\ref{cor:correlation}. Moreover, these results extend to more general noise models, assuming the availability of an ideal sampler for Pauli-type spacetime errors. 

However, the error bounds no longer hold rigorously in the presence of temporally correlated non-Pauli errors. This is because exact twirling of non-stabilizer operations requires stabilizer operations to be effectively error-free. When errors are temporally uncorrelated, errors in stabilizer and non-stabilizer operations are independent, allowing us to interpret the process as first mitigating the stabilizer operations, which are then used to twirl the non-stabilizer operations (see the proof in Sec.~\ref{app:error_cost_practical}). While unbiasedness is not theoretically guaranteed under temporally correlated non-Pauli noise, numerical simulations show no observable bias; see Sec.~\ref{app:numerics}. 

\begin{corollary}
Suppose the following conditions hold: 
\begin{itemize}
\item There exists a spacetime noise map $\mathcal{N}_{max}$ such that final states of noisy circuits satisfy Eq.~(\ref{eq:rhof_F_MSN}), where $F$ is given by Eq.~(\ref{eq:F_function}); 
\item Pauli gates are error-free; 
\item Errors in all other operations are Pauli (Definition~\ref{def:Pauli_noise}). 
\end{itemize}
Assume the existence of an ideal spacetime error sampler that generates spacetime errors according to $\mathcal{N}_{max}$. Then, statements in Theorem~\ref{the:error_cost_ideal} still hold. 
\label{cor:correlation}
\end{corollary}

\begin{proof}
Since Pauli gates are assumed to be error-free, the noise map $\mathcal{N}_{max}$ takes the form given in Eq.~(\ref{eq:NPNnonP}). However, the non-stabilizer component $\mathcal{N}_{non-P}$ may not necessarily be a product map as expressed in Eq.~(\ref{eq:NnonP}) (in the presence of temporally correlated errors). Importantly, all conclusions in the proof of Theorem~\ref{the:error_cost_ideal} remain valid even if $\mathcal{N}_{non-P}$ does not have a product form.
\end{proof}

We remark that when error parameters are unstable, the average final state satisfies Eq.~(\ref{eq:rhof_F_MSN}) by taking $\mathcal{N}_{max} = \mathcal{N}_{ave}$. For randomized dynamic circuits, the average noise map is $\mathcal{N}_{ave} = \int d\mathbf{p} g(\mathbf{p}) \mathcal{N}_{max}(\mathbf{p})$, where $\mathcal{N}_{max}(\mathbf{p})$ is the maximum spacetime noise map of the given error rates $\mathbf{p}$. 

\section{Impact of Pauli-gate and super-qubit errors}
\label{app:pauli_SQ_errors}

In Corollary~\ref{cor:error_cost_practical}, errors on Pauli gates and super-qubit operations are neglected. In this section, we analyze their impact. We use the diamond norm as the measure of errors in Pauli gates and super-qubit operations, and we make the following assumptions: 
\begin{itemize}
\item Let $\mathcal{M}_\sigma$ be the completely positive map describing the {\it logical} Pauli gate $\sigma$ with noise. There exists a positive number $\epsilon_P$ such that, the error in $\mathcal{M}_\sigma$ is upper bounded by $\norm{\mathcal{M}_\sigma - [\sigma]}_\diamond \leq \epsilon_P$ for all $\sigma$. 
\item There exists a positive integer $n_P$ such that the number of Pauli gates in each error sampler circuit is not larger than $n_P$. 
\item Given a computational circuit with $N$ operations, we randomly modify the circuit by inserting Pauli gates to mitigate errors. There exists a positive number $\chi_P$ such that the number of Pauli gates in each modified computational circuit is not larger than $\chi_P N$. 
\item There exists a positive number $\epsilon_S$ such that, the error in an operation on super qubits is upper bounded by $\norm{\mathcal{N}^L - [\openone]}_\diamond + \norm{\mathcal{N}^R - [\openone]}_\diamond \leq \epsilon_S$, where $\mathcal{N}^L$ and $\mathcal{N}^R$ are noise maps associated with the operation [see Eq.~(\ref{eq:noisy_operation})]. 
\item There exists a positive integer $n_S$ such that the number of super-qubit operations in each error sampler circuit is not larger than $n_S$. 
\end{itemize}
We focus on applications to surface codes, though the analysis can be generalized to qLDPC codes in which errors on logical qubits within each block are correlated; see Sec.~\ref{app:surface_correlation}. Under the assumption of an operation set $\mathbb{O} \subseteq \mathbb{O}_{P} \cup \mathbb{O}_{S} \cup \mathbb{O}_{non-S}$, each error sampler circuit includes up to fifteen super-qubit operations, which is given by Fig.~\ref{fig:error_sampler}(f): four initialization operations, two controlled-NOT gates for preparing the Bell states, the gate $U^\dag$, idle operations on the two ancilla super qubits, two controlled-NOT gates for measuring the Bell states and four measurement operations. Each error sampler circuit includes up to five logical operations, which is also given by Fig.~\ref{fig:error_sampler}(f): two decoding operations, the gate $U$ and two encoding operations. Each operation requires up to four logical Pauli gates to implement twirling. Therefore, in this case, we can take $n_P = 20$, $\chi_P = 5$ and $n_S = 15$. 

{\bf Error sampler circuits.} Let $\rho$ be the final state of an error sampler circuit when Pauli gates and super-qubit operations are error-free, and let $\rho'$ be the final state when errors in these operations are switched on. Note that by final state, we mean the effective final state after the noise maps associated with final measurements, therefore, measurement errors have been taken into account. According to a property of the diamond norm, we have the trace norm $\norm{\rho' - \rho}_1 \leq n_P\epsilon_P + n_S\epsilon_S$: Specifically, the property is $\norm{\mathcal{M}'\rho'-\mathcal{M}\rho}_1 \leq \norm{\mathcal{M}' - \mathcal{M}}_\diamond + \norm{\rho'-\rho}_1$ when $\mathcal{M}$ and $\mathcal{M}'$ are trace-preserving completely positive maps, and $\rho$ and $\rho'$ are normalized reduced density matrices. These errors in the state deviate the distribution of error samples generated by the circuit, and the distance between two distributions quantified by the $L^1$ norm is upper bounded by the trace norm~\cite{quantum_nielsen_2010}. Let $\mathcal{N}$ and $\mathcal{N}'$ be Pauli noise maps describing the distributions of errors observed from states $\rho$ and $\rho'$, respectively, we have 
\begin{eqnarray}
\norm{\mathcal{N}' - \mathcal{N}}_\diamond = \norm{\mathcal{N}' - \mathcal{N}}_{L^1_{Pauli}} \leq n_P\epsilon_P + n_S\epsilon_S.
\label{eq:N1N0}
\end{eqnarray}

{\bf Spacetime error model.} The distribution of spacetime errors generated by error sampler circuits is described by the spacetime noise map given in Eq.~(\ref{eq:N_eff_nonP}), in which each $\overline{\mathcal{N}}_{eff}(\alpha)$ is the noise map corresponding to an error sampler circuit. We already have an upper bound on the error in each $\overline{\mathcal{N}}_{eff}(\alpha)$ as given in Eq.~(\ref{eq:N1N0}). Let $\overline{\mathcal{N}}_{eff,non-P}$ be the spacetime noise map when Pauli gates and super-qubit operations are error-free, and let $\overline{\mathcal{N}}_{eff,non-P}'$ be the map when errors in these operations are switched on. Their difference is 
\begin{eqnarray}
\norm{\overline{\mathcal{N}}_{eff,non-P}' - \overline{\mathcal{N}}_{eff,non-P}}_{L^1_{Pauli}} \leq N(n_P\epsilon_P + n_S\epsilon_S).
\end{eqnarray}
Then the bias introduced by the difference is upper bounded by [see Eq.~(\ref{eq:bias1})] 
\begin{eqnarray}
&& \norm{a}_{L^\infty} \norm{\overline{\mathcal{N}}_{eff,non-P}^{\prime-1}-\overline{\mathcal{N}}_{eff,non-P}^{-1}}_{L^1_{Pauli}} \notag \\
&\leq & \norm{a}_{L^\infty} \left\vert\frac{1}{1-2P'} - \frac{1}{1-2P}\right\vert
+ \norm{a}_{L^\infty} \frac{1}{(1-2P)^2} \left[2\left\vert\frac{P}{P'} - 1\right\vert + N(n_P\epsilon_P + n_S\epsilon_S)\right].
\end{eqnarray}
The first term is due to the difference in the corresponding total error rates $P'$ and $P$, and the second term is due to the error in $\mathcal{E}_{non-P}$ [see Eq.~(\ref{eq:inverse_diff})]. Here, we have used $\norm{\mathcal{E}_{non-P}^{\prime k}-\mathcal{E}_{non-P}^k}_{L^1_{Pauli}} \leq k \norm{\mathcal{E}_{non-P}'-\mathcal{E}_{non-P}}_{L^1_{Pauli}}$ when $\norm{\mathcal{E}_{non-P}'}_{L^1_{Pauli}} = \norm{\mathcal{E}_{non-P}}_{L^1_{Pauli}} = 1$, and 
\begin{eqnarray}
\norm{\mathcal{E}_{non-P}'-\mathcal{E}_{non-P}}_{L^1_{Pauli}} \leq 2\left\vert\frac{1}{P'} - \frac{1}{P}\right\vert + \frac{1}{P} \norm{\overline{\mathcal{N}}_{eff,non-P}' - \overline{\mathcal{N}}_{eff,non-P}}_{L^1_{Pauli}}.
\end{eqnarray}
Note that $\abs{P'-P} \leq \norm{\overline{\mathcal{N}}_{eff,non-P}' - \overline{\mathcal{N}}_{eff,non-P}}_{L^1_{Pauli}}$. 

When $N(n_P\epsilon_P + n_S\epsilon_S)$ is small, the impact on variance and sampling cost remains limited, except when the total error rate $P$ is close to either $0$ or $1/2$ (see Sec.~\ref{app:modifications} for methods to address cases where $P$ approaches these extremes). 

{\bf Computational circuits.} Pauli-gate and super-qubit errors impact the computational circuits in two ways. First, the inserted encoding and decoding noise is changed by errors in error sampler circuits. Second, the logical Pauli gates in the computational circuits directly cause errors. Let $\tilde{\rho}$ and $\tilde{\rho}'$ be final states of a computational circuit without and with these errors. We have 
\begin{eqnarray}
\norm{\tilde{\rho}' - \tilde{\rho}}_1 \leq 2N(n_P\epsilon_P + n_S\epsilon_S) + \chi_P N \epsilon_P,
\end{eqnarray}
where the first term is due to the inserted encoding and decoding noise (each operation acts on at most two logical qubits, therefore, associates with at most two encoding and decoding noise maps), and the second term is due to logical Pauli gates in the circuit. Accordingly, the bias contributed by these errors is upper bounded by 
\begin{eqnarray}
\frac{1}{1-2P'} \norm{a}_{L^\infty} \norm{\tilde{\rho}' - \tilde{\rho}}_1 \leq \frac{1}{1-2P'} \norm{a}_{L^\infty} N(2n_P\epsilon_P + 2n_S\epsilon_S + \chi_P \epsilon_P),
\end{eqnarray}
where $\frac{1}{1-2P'}$ is due to the PEC normalization factor. 

\section{Modifications to the protocol when $P$ is too large or too small}
\label{app:modifications}

The analysis in Sec.~\ref{app:error_cost_ideal} shows that SNI requires a total error rate satisfying $P < 1/2$. 
To overcome this limitation of SNI, we introduce the following modification to the protocol. We can find a partition of non-Pauli operations in the circuit to satisfy the following condition: For each subset of operations, its spacetime noise map is $\mathcal{N}_{non-P,j}$; the overall spacetime noise map is 
\begin{eqnarray}
\mathcal{N}_{non-P} = \mathcal{N}_{non-P,1}\otimes\mathcal{N}_{non-P,2}\otimes\cdots;
\end{eqnarray}
and the error rate of each noise map $\mathcal{N}_{non-P,j}$ is sufficiently smaller than $1/2$. Then we can individually apply the inverse map to each noise map $\mathcal{N}_{non-P,j}$. However, in this approach, we need to evaluate multiple parameters (error rates of $\mathcal{N}_{non-P,j}$) instead of a single parameter. 

When $P$ is small, generating non-trivial spacetime errors becomes inefficient, leading to a large $\mathrm{Var}(M_{es})$. To overcome this issue, we can take some instances of trivial errors (identity map) as non-trivial errors when generating non-trivial spacetime errors (lines 5 and 6 in Algorithm~\ref{alg:estimator_of_total_error_rate} and lines 9 and 10 in Algorithm~\ref{alg:ideal_processed_error_sampler}): If the error is trivial, we take it as a non-trivial error with a probability of $q$, i.e. we replace the if-statement with \\
${\;\;\;\;\;\;\;\;}$${\;\;\;\;\;\;\;\;}$${\;\;\;\;\;\;\;\;}$${\;\;\;\;\;\;\;\;}$${\;\;\;\;\;\;\;\;}$${\;\;\;\;\;\;\;\;}$${\;\;\;\;\;\;\;\;}$Generate $\nu \sim \textbf{Bernoulli}(q)$ \\
${\;\;\;\;\;\;\;\;}$${\;\;\;\;\;\;\;\;}$${\;\;\;\;\;\;\;\;}$${\;\;\;\;\;\;\;\;}$${\;\;\;\;\;\;\;\;}$${\;\;\;\;\;\;\;\;}$${\;\;\;\;\;\;\;\;}$\textbf{if} there exists $(i,j)$ such that $\sigma(i,j) \neq \openone$ or $\nu$ is $1$ \textbf{then} \\
This modification corresponds to rewriting the spacetime noise as 
\begin{eqnarray}
\mathcal{N}_{non-P} = (1-P')[\openone_{non-P}] + P'\mathcal{E}'_{non-P},
\end{eqnarray}
where $P' = P + q(1-P)$ and 
\begin{eqnarray}
\mathcal{E}'_{non-P} = \frac{P\mathcal{E}_{non-P} + q(1-P)[\openone_{non-P}]}{P'}.
\end{eqnarray}
Therefore, we can increase $P$ to $P'$ in this way, and all results in this section still hold after replacing $P$ with $P'$. 

\subsection{Inaccurate error-rate estimation in the multi-segement protocol}

We can understand the impact of inaccurate error-rate estimation as follows. Consider the single-segment error mitigation scheme. Let $\mathcal{N}$ denote the true spacetime error model. Ideally, if the total error rate is estimated exactly, we can perfectly invert the noise by applying $\mathcal{N}^{-1}$, yielding an unbiased quantum computation. However, if the error rate is estimated inaccurately, the inversion becomes imperfect, resulting in a residual bias. Let $\hat{\mathcal{N}}$ denote the noise model corresponding to the estimated (inaccurate) error rate. The bias in the final observable is upper bounded by
\begin{eqnarray}
\norm{a}_{L^\infty} \left\Vert \hat{\mathcal{N}}^{-1} - \mathcal{N}^{-1} \right\Vert_{L^1_{Pauli}} \leq \norm{a}_{L^\infty} \left| \frac{1}{1 - 2\hat{P}} - \frac{1}{1 - 2P} \right| \simeq \norm{a}_{L^\infty} \frac{2\delta P}{(1 - 2P)^2},
\end{eqnarray}
where $P$ is the true total error rate, $\hat{P}$ is its estimate, and $\delta P = |\hat{P} - P|$ is the estimation error (see Sec.~\ref{app:error_cost_ideal}). This expression shows that the bias grows with $\delta P$, and diverges as $P$ approaches $1/2$. Hence, ensuring that the total error rate remains well below $1/2$ (e.g.~$P \leq 0.4$) is essential for reliable error mitigation.

With two segments, the true spacetime error model can be expressed as $\mathcal{N} = \mathcal{N}_1 \otimes \mathcal{N}_2$: The entire circuit is divided into two sub-circuits (segments) with $\mathcal{N}_1$ acting on the first sub-circuit and $\mathcal{N}_2$ on the second one. Accordingly, the bias is upper bounded by
\begin{eqnarray}
&&\norm{a}_{L^\infty} \left\Vert \hat{\mathcal{N}}_1^{-1} \otimes \hat{\mathcal{N}}_2^{-1} - \mathcal{N}_1^{-1} \otimes \mathcal{N}_2^{-1} \right\Vert_{L^1_{Pauli}} \notag \\
&\leq & \norm{a}_{L^\infty} \left( \left\Vert \hat{\mathcal{N}}_1^{-1} \otimes \hat{\mathcal{N}}_2^{-1} - \mathcal{N}_1^{-1} \otimes \hat{\mathcal{N}}_2^{-1} \right\Vert_{L^1_{Pauli}} + \left\Vert \mathcal{N}_1^{-1} \otimes \hat{\mathcal{N}}_2^{-1} - \mathcal{N}_1^{-1} \otimes \mathcal{N}_2^{-1} \right\Vert_{L^1_{Pauli}} \right) \notag \\
&\leq & \norm{a}_{L^\infty} \left( \frac{1}{1 - 2\hat{P}_2} \left| \frac{1}{1 - 2\hat{P}_1} - \frac{1}{1 - 2P_1} \right| + \frac{1}{1 - 2P_1} \left| \frac{1}{1 - 2\hat{P}_2} - \frac{1}{1 - 2P_2} \right| \right) \notag \\
&\simeq & \frac{2 \norm{a}_{L^\infty}}{(1 - 2P_1)(1 - 2P_2)} \left( \frac{\delta P_1}{1 - 2P_1} + \frac{\delta P_2}{1 - 2P_2} \right).
\end{eqnarray}
Here, $P_j$ is the error rate of the $j$th segment, $\hat{P}_j$ is its estimate, and $\delta P_j = |\hat{P}_j - P_j|$. 

The above analysis can be generalized to $S$ segments. In this case, the bias is upper bounded by
\begin{eqnarray}
\norm{a}_{L^\infty} \left\Vert \bigotimes_{j=1}^S \hat{\mathcal{N}}_j^{-1} - \bigotimes_{j=1}^S \mathcal{N}_j^{-1} \right\Vert_{L^1_{Pauli}}
\lesssim \frac{2\norm{a}_{L^\infty}}{\prod_{j=1}^S (1 - 2P_j)} \sum_{j=1}^S \frac{\delta P_j}{1 - 2P_j}.
\end{eqnarray}
From this estimate, we draw several conclusions:
\begin{enumerate}
\item {\it Bias accumulation with segmentation.} The total bias increases with the number of segments, as the accuracy of each segment’s error-rate estimate contributes to the overall bias. This issue is also present in conventional PEC, where multiple parameter estimates are required. By contrast, our method (in its non-segmented form) requires measuring only a single parameter, offering a clear practical advantage.
\item {\it Bias amplification by quasi-probability overhead.} The bias is amplified by the factor $\left[\prod_{j=1}^S (1 - 2P_j)\right]^{-1}$. For simplicity, assuming each segment has the same error rate $P_0$, the amplification factor becomes approximately $e^{2P_0 S}$ (assume a small $P_0$). This is precisely the quasi-probability overhead factor associated with the segmented protocol, and it matches the scaling of conventional PEC. Thus, it is essential that $P_0 S$ remains moderate to keep the overhead and bias reasonable.
\item {\it Extension of tolerable total error.} Despite the limitations above, the segmentation approach can significantly increase the tolerable total error rate of the full circuit. With segmentation, the condition of applying our method becomes $P_0<50\%$. Regarding the total error rate $1 - (1 - P_0)^S \simeq 1 - e^{-P_0 S}$, it can exceed $50\%$ and approach unity as $S$ increases.
\end{enumerate}

\section{Benchmarking of surface-code operations}
\label{app:benchmarking_SC}
\subsection{On-site benchmarking}

\begin{figure}[!htbp]
\centering
\includegraphics[width=4 in]{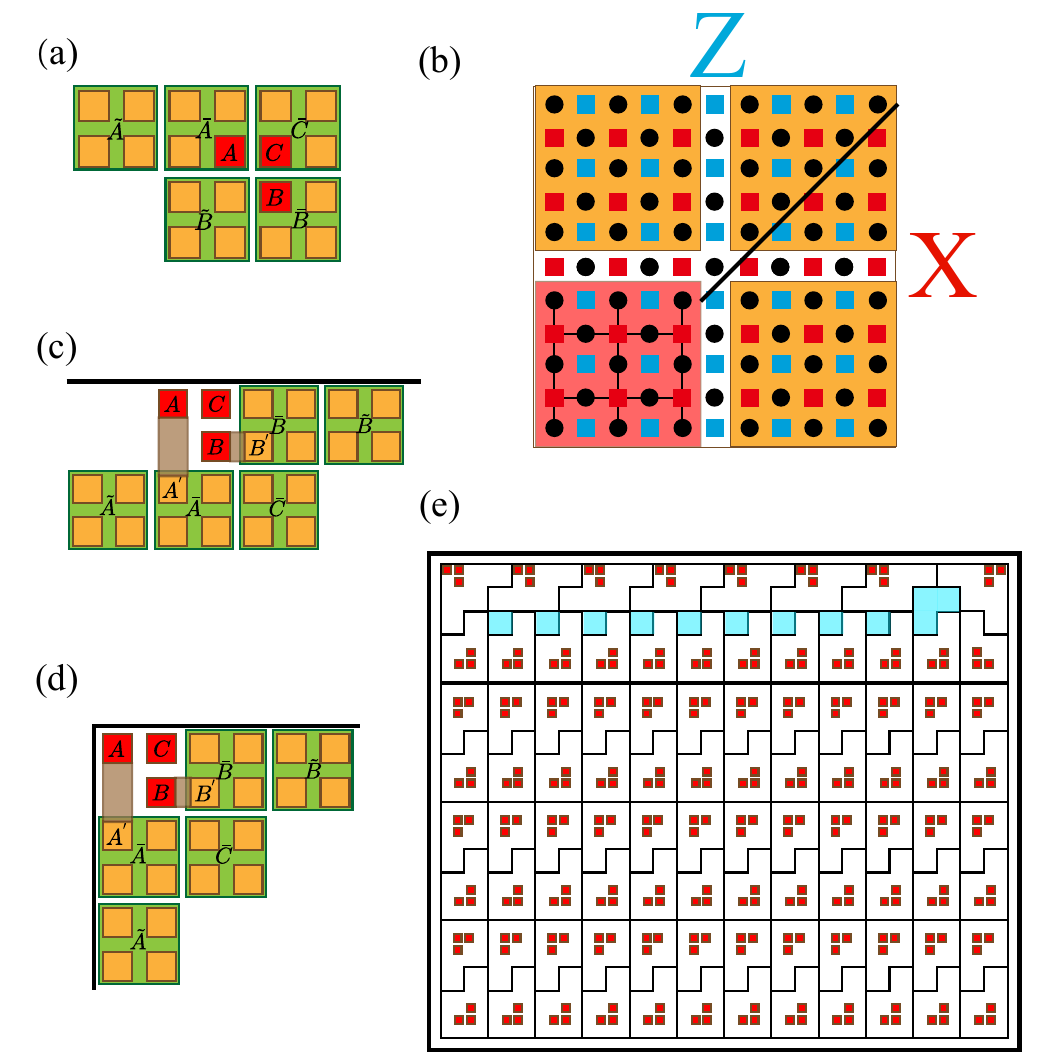}
\caption{
Qubit layout for sampling CNOT gate errors on surface codes. Orange and red squares denote logical qubits, and green squares denote super qubits. Brown regions are reserved for implementing state transfer.
(a) Layout for a CNOT gate located away from the boundary.
(b) Encoding and decoding operations. The data qubits outside the logical qubit (red square) are divided into two regions by the back-diagonal line. Encoding and decoding operations are performed by initializing and measuring, respectively, these qubits in either the $X$ or $Z$ basis as illustrated in the figure~\cite{Li_2015}. To complete the encoding operation, we perform one round of parity-check measurements on the super qubit after initializing the data qubits. 
(c) Layout for a CNOT gate near an edge.
(d) Layout for a CNOT gate at a corner of the qubit array.
(e) Example layout for parallel benchmarking of multiple CNOT gates. Blue squares mark unused qubits, on which error sampling for certain single-qubit operations can be performed to optimize resource utilization.
}
\label{fig:CNOT}
\end{figure}

In platforms such as superconducting qubits, error rates are qubit-dependent, making it essential to benchmark each operation as implemented on specific qubits. Such on-site benchmarking is possible. As an example, we consider logical CNOT gates on surface codes. Although a CNOT gate acts on two logical qubits, its implementation via lattice surgery requires an additional ancilla block~\cite{horsman_surface_2012}. In Fig.~\ref{fig:CNOT}(a), the CNOT gate to be benchmarked is performed specifically between logical qubits $A$ and $B$, with logical qubit $C$ serving as the ancilla.

To benchmark the CNOT gate, we need to use super qubits. For simplicity, we consider the case that super qubits have the distance $d_S = 2d$, meaning each super qubit occupies the physical qubits of four logical qubits [as illustrated in Fig.~\ref{fig:CNOT}(a)]. We emphasize that while this serves as an example, practical implementations may achieve sufficient performance with $d_S \approx 1.3 d$ (see Sec.~\ref{app:applications_SC}).

The CNOT gate benchmarking requires five super qubits and proceeds through six steps. 
\begin{enumerate}
\item Prepare Bell states on super-qubit pairs $(\tilde{A},\bar{A})$ and $(\tilde{B},\bar{B})$ through transversal $X$ initializations and joint $ZZ$ measurements on adjacent super qubits. The joint measurements are implemented via lattice surgery~\cite{horsman_surface_2012}.
\item Execute the CNOT gate between $\bar{A}$ and $\bar{B}$ using ancilla super qubit $\bar{C}$.
\item Transform super qubits $\bar{A}$ and $\bar{B}$ into logical qubits $A$ and $B$, respectively, which is achieved by measuring out extra physical qubits. The pattern of measurement basis is illustrated in Fig.~\ref{fig:CNOT}(b)~\cite{Li_2015}. This step is the decoding operation.
\item Perform the target CNOT gate between $A$ and $B$ using ancilla logical qubit $C$.
\item Encode $A$ and $B$ back into $\bar{A}$ and $\bar{B}$, which involves reinitialising extra physical qubits following the same pattern as decoding. The subsequent round of parity-check measurements on $\bar{A}$ and $\bar{B}$ completes this encoding operation.
\item Conduct Bell measurements on $(\tilde{A},\bar{A})$ and $(\tilde{B},\bar{B})$ via joint $ZZ$ measurements and transversal $X$ measurements. Then, we can extract CNOT gate error information from measurement outcomes.
\end{enumerate}
In this way, we can benchmark the specific CNOT gate implemented on logical qubits $A$ and $B$ (mediated by the ancilla $C$).

The above protocol can be applied to CNOT gates acting on logical qubits located away from the boundary of a two-dimensional qubit array. For CNOT gates acting near the boundary, the arrangement must be modified. For example, we consider logical qubits $A$, $B$, and $C$ near an edge of the qubit array. In this case, the CNOT gate can be benchmarked as illustrated in Fig.~\ref{fig:CNOT}(c). The modified protocol proceeds as follows:
\begin{enumerate} 
\item Prepare Bell states on super-qubit pairs $(\tilde{A},\bar{A})$ and $(\tilde{B},\bar{B})$;
\item Execute the CNOT gate on $\bar{A}$ and $\bar{B}$;
\item Decode $\bar{A}$ and $\bar{B}$ into logical qubits $A'$ and $B'$, respectively;
\item Transport $A'$ and $B'$ to the locations of $A$ and $B$;
\item Execute the CNOT gate on $A$ and $B$;
\item Transport $A$ and $B$ back to the locations of $A'$ and $B'$;
\item Re-encode $A'$ and $B'$ into $\bar{A}$ and $\bar{B}$; 
\item Conduct Bell measurements on $(\tilde{A},\bar{A})$ and $(\tilde{B},\bar{B})$. 
\end{enumerate}
These transport operations can be implemented via lattice surgery~\cite{litinski_game_2019}. Note that we can eliminate the impact of encoding and decoding errors. By incorporating transport operations into encoding and decoding phases, our method can effectively mitigate the associated errors, thus ensuring unbiased computation results. It is similar for a CNOT gate located at a corner of array as shown in Fig.~\ref{fig:CNOT}(d). 

For parallel benchmarking, we arrange the qubits as illustrated in Fig.~\ref{fig:CNOT}(e). Each benchmarking round can only target a subset of operations. It is desirable to develop an algorithm that optimizes the benchmarking schedule to maximize the utilization of the quantum computer and benchmark as many operations as possible in parallel. However, this optimization is beyond the scope of the present work. Here, we illustrate an example pattern that benchmarks multiple CNOT gates simultaneously.

\subsection{Cost of benchmarking}

% Consider the protocol in Fig.~\ref{fig:CNOT}(a) for benchmarking a logical CNOT gate, it uses twenty logical qubits (i.e.~five super qubits) to benchmark an operation acting on three logical qubits. We estimate the total time cost to be $d_S + 2d_S + 2d + d_S = 10d$ rounds of parity-check measurements. Each term corresponds to a specific step in the protocol, and we have neglected the time cost of encoding and decoding. We note that benchmarking a gate near the boundary requires additional time overhead due to the need to transport logical qubits. A small prefactor may also apply, depending on the measurement error rate. We can find that a logical CNOT gate has a spacetime cost of about $3(2d)^2 \times 2d = 24d^3$ (involving three logical qubits over $2d$ rounds), and the benchmarking has a spacetime cost of about $5(4d)^2 \times 10d = 800d^3$. This yields an overhead factor of $800/24 \approx 33$. We remark that this estimate assumes a code distance ratio of $d_S/d = 2$. If instead we use $d_S/d \approx 1.3$, the overhead factor reduces to $800/24 \times (1.3/2)^3 \approx 9$.

Consider the protocol in Fig.~\ref{fig:CNOT}(a) for benchmarking a logical CNOT gate, it uses twenty logical qubits (i.e.~five super qubits) to benchmark an operation acting on three logical qubits. We estimate the total time cost to be $d_S + 2d_S + 2d + d_S = 4d_S+2d$ rounds of parity-check measurements. Each term corresponds to a specific step in the protocol, and we have neglected the time cost of encoding and decoding. We note that benchmarking a gate near the boundary requires additional time overhead due to the need to transport logical qubits. A small prefactor may also apply, depending on the measurement error rate. We can find that a logical CNOT gate has a spacetime cost of about $3(2d)^2 \times 2d = 24d^3$ (involving three logical qubits over $2d$ rounds), and the benchmarking has a spacetime cost of about $5(2d_S)^2 \times (4d_S+2d) = 80d_S^3+40d_S^2d$. This yields an overhead factor of $(80d_S^3+40d_S^2d)/(24d^3)=10(d_S/d)^3/3+5(d_S/d)^2/3$. As shown in Fig.~\ref{fig:CNOT}(a), the overhead take $33$ assuming a code distance ratio of $d_S/d = 2$. If instead we use $d_S/d \approx 1.3$, the overhead factor reduces to $10$.

The above analysis applies generally to logical operations on surface codes implemented via lattice surgery, whose spacetime costs scale as $O(d^3)$~\cite{horsman_surface_2012}. Since our benchmarking protocol uses super qubits with distance $d_S$, the total benchmarking cost scales as $O(d_S^3)$. Consequently, the overhead factor is $O((d_S/d)^3)$, with a prefactor of approximately 3–4 as discussed above.

A subtle issue in benchmarking arises with $T$ gates, which require the preparation of corresponding magic states. The magic states prepared for distance-$d$ logical qubits may not meet the fidelity requirements for distance-$d_S$ super qubits used in benchmarking. As a result, an additional round of magic state distillation may be necessary prior to their use in benchmarking. In the worst-case scenario, this could increase the magic state cost by up to a factor of fifteen~\cite{horsman_surface_2012}, introducing an additional overhead in the benchmarking procedure. That said, this represents a conservative estimate. Various techniques, such as magic state cultivation and advanced magic state distillation, can significantly reduce the cost of preparing high-fidelity magic states~\cite{quantum_ogorman_2017,magic_litinski_2019,magic_gidney_2024,constantoverhead_wills_2024,PhysRevA.109.062438}.

Lastly, we offer an alternative perspective on interpreting the benchmarking cost. In our error mitigation approach, the accuracy of the super qubits fundamentally determines the accuracy of the computation. Thus, given a quantum computer with $K$ logical qubits at code distance $d$, our protocol effectively upgrades its computational accuracy to the level of a system with $K$ logical qubits at a higher code distance $d_S$. From this viewpoint, it is reasonable to compare the time cost of our protocol to that of a quantum computer physically equipped with $K$ logical qubits at distance $d_S$. Such a distance-$d_S$ quantum computer would inherently incur a larger spacetime cost, by a factor of $O((d_S/d)^3)$, and would require higher-fidelity magic states. Viewed this way, our error mitigation protocol introduces only a distance-independent overhead relative to a distance-$d_S$ quantum computer, because each error sampler circuit requires $O(1)$ distance-$d_S$ logical operations.

\section{Numerical simulations}
\label{app:numerics}

We present three numerical simulations. In the first simulation, we study how spatial correlations at the physical level affect logical qubit errors. In the second simulation, we introduce spatially correlated noise mimicking intra-block correlations in qLDPC codes. We show that conventional PEC based on a sparse error model exhibits significant bias. In contrast, SNI remains unbiased under the same conditions. In the third simulation, we test SNI under temporally correlated non-Pauli noise and confirm its robustness: As the benchmarking cost $M_P$ increases, the bias consistently decreases, approaching unbiased computation. 

\subsection{Spatially correlations in surface codes}
\label{app:surface_correlation}

Two surface-code blocks are placed adjacent to each other as illustrated in Fig.~\ref{fig:surface}(a). For a given code distance $d$, stabilizer measurements are repeated for $d$ rounds. We adopt a phenomenological error model to examine the influence of spatial correlations. In each round, every qubit (data or ancilla) undergoes two types of maps, one for $X$ and one for $Z$ (bit-flip and phase-flip errors, respectively). The $\sigma = X,Z$ map associated with qubit $r$ (denoted by the green cross) is $\mathcal{N}_{\sigma, r} = (1 - p)[\openone] + 0.9p[\sigma_r] + 0.1p \vert N_r \vert^{-1} \sum_{r' \in N_r} [\sigma_r \sigma_{r'}]$, where $p$ is the physical error rate, $[\openone]$ is the identity map, and $N_r$ is the set of neighboring qubits around qubit $r$ up to Chebyshev distance two, indicated by the purple box. Note that errors on ancilla qubits manifest as measurement errors, i.e.,~they result in incorrect syndrome data. With $90\%$ probability, the error acts only on qubit $r$, and with $10\%$ probability, it acts on both qubit $r$ and one of its neighbors, modeling spatial correlations. 

%The $\sigma = X,Z$ map associated with qubit $r$ is $\mathcal{N}_{\sigma, r} =  \mathcal{A}_{\sigma, r} \prod_{r' \in N_r} \mathcal{B}_{\sigma, r, r'}$, where $\mathcal{A}_{\sigma, r} = (1 - p)[\openone] + p [\sigma_r]$ describes independent errors, $\mathcal{B}_{\sigma, r, r'} = (1 - \eta p)[\openone] + \eta p [\sigma_r \sigma_{r'}]$ describes correlated errors, $p$ and $\eta p$ are the physical error rates, $\eta = ???$, $[\openone]$ is the identity map, and $N_r$ is the set of neighboring qubits around qubit $r$, indicated by purple boxes. 

Due to the block layout, logical $Z$ errors are more susceptible to cross-block correlations. We estimate each $Z$ logical error rate using $1.28 \times 10^7$ samples in the Monte Carlo simulation, and the results are shown in Fig.~\ref{fig:surface}(b): $P_1$ and $P_2$ denote the probabilities that logical qubits $1$ and $2$, respectively, undergo logical $Z$ errors, and $P_{1\&2}$ denotes the joint probability that both qubits simultaneously experience logical $Z$ errors. A deviation of $P_{1\&2}$ from the product $P_1 P_2$ indicates the presence of error correlations between the two logical qubits. For two neighboring physical qubits [highlighted by the yellow rounded rectangle in Fig.~\ref{fig:surface}(a)], we use $p_1$, $p_2$, and $p_{1\&2}$ to denote the individual and joint physical error probabilities. The large discrepancy between $p_1p_2$ and $p_{1\&2}$ indicates significant correlations at the physical level. In contrast, for logical qubits, the negligible difference between $P_1P_2$ and $P_{1\&2}$ shows that such spatial correlations do not induce observable correlations at the logical level. 

\begin{figure}[!htbp]
\centering
\includegraphics[width=4 in]{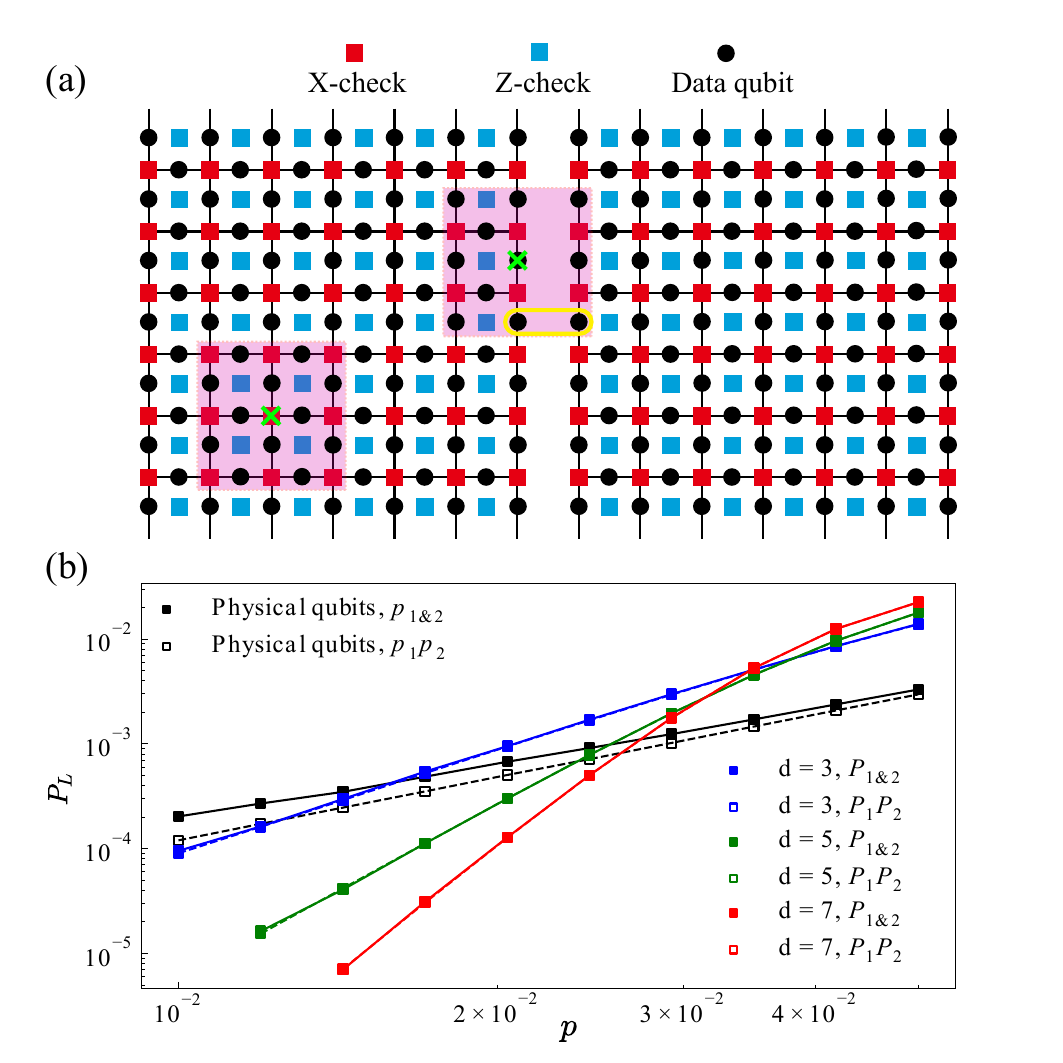}
\caption{
Impact of physical-level spatial correlations on logical qubit errors. 
(a) Error model. 
(b) Logical error rates. 
See the text for details. 
}
\label{fig:surface}
\end{figure}

\subsection{Mitigating spatially correlated errors}
\label{app:space}

\begin{figure}[htbp]
\centering
\includegraphics[width=0.5\linewidth]{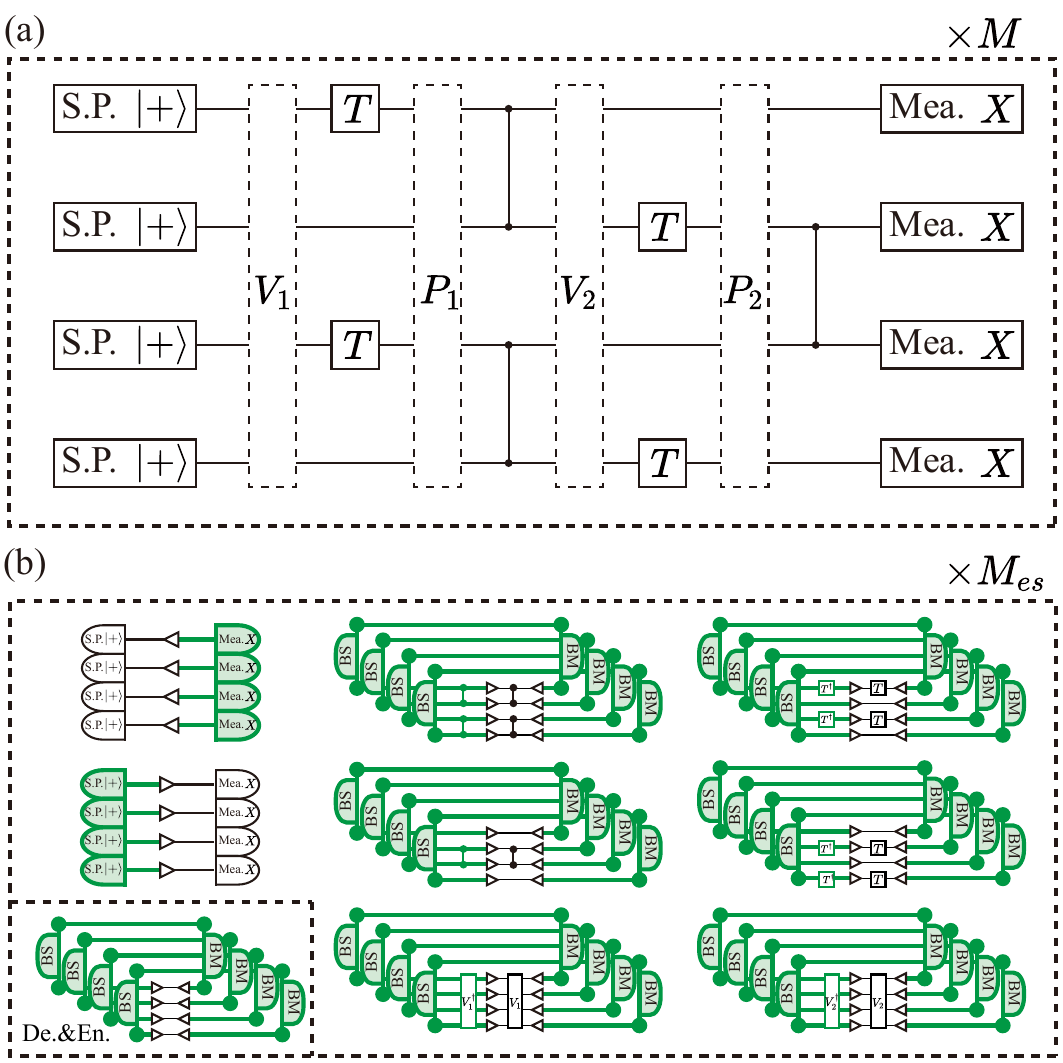}
\caption{
Circuits used in the simulation for spatially correlated errors. 
(a) Computational circuit. Gates denoted by dashed boxes implement twirling operations on $T$ gates; see Table~\ref{tab:twirled_operations}. Here, $P_i$ denotes a Pauli gate, and $V_i$ denotes the corresponding Clifford gate. 
(b) Error sampler circuits. To generate a spacetime error sample, an error sampler circuit is implemented for each possible non-Pauli twirling gate $V_i$. 
}
\label{fig:circuits_space}
\end{figure}

\begin{figure}[htbp]
\centering
\includegraphics[width=0.5\linewidth]{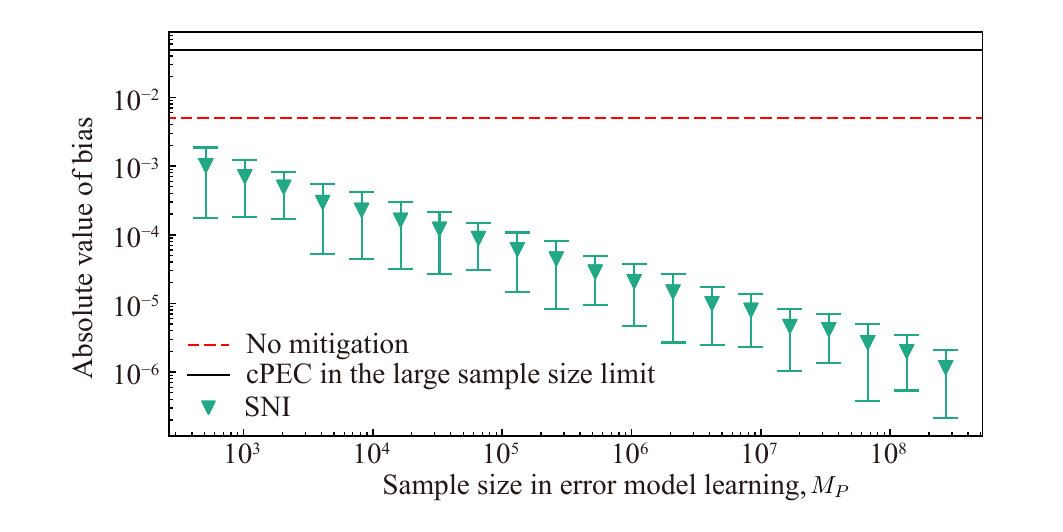}
\caption{
Bias as a function of $M_P$ with spatially correlated errors. The observable $X^{\otimes k}$ is evaluated using the circuit in Fig.~\ref{fig:circuits_space}. The error model parameter $p$ take $0.005$ in the simulations. The large sample size limit means $M_P \rightarrow \infty$. 
}
\label{fig:space_layer}
\end{figure}

In this numerical simulation, we apply SNI and conventional PEC (using an uncorrelated error model) to mitigate spatially correlated errors in a system of $k = 4$ qubits. The system could be a qLDPC code block that encodes $k$ logical qubits, which experiences spatially correlated errors. For a comparison of different methods in the presence of spatially correlated errors, we model the noise using a global depolarizing channel. The global depolarizing channel coincides with the numerical results of logical errors in a qLDPC code block in Ref.~\cite{zhang_demonstrating_2025}, which show that the most probable event is that half of the logical qubits experience errors simultaneously. In addition, non-Clifford gates also undergo non-Pauli errors.

Below is a detailed description of the error model used in our simulations. We simulate a four-qubit computational circuit alongside a set of error sampler circuits (see Fig.~\ref{fig:circuits_space}). The computational circuit includes ten layers of operations. The operations are affected by noise as specified below:
\begin{itemize}
\item Logical Pauli gates and operations on super qubits are assumed to be error-free.
\item A layer of stabilizer operations (one of eight layers other than the two $T$ gate layers), if at least one of the operations is non-Pauli, is subject to $k$-qubit depolarizing noise with depolarizing rate $p$. This noise map is defined as
\begin{eqnarray}
\mathcal{N}_k = (1 - \frac{4^k - 1}{4^k}p)[I^{\otimes k}] + \frac{1}{4^k}p \sum_{\sigma \in \mathbb{P}_k - \{I^{\otimes k}\}} [\sigma],
\end{eqnarray}
and is applied after state preparation and gates, and before measurement. 
\item Layers of encoding and decoding operations are subject to $k$-qubit depolarizing noise with a reduced depolarizing rate of $p/3$. The noise map is given by
\begin{eqnarray}
\mathcal{N}_{en/de} = \left(1 - \frac{4^k-1}{4^k}\frac{p}{3}\right) + \frac{1}{4^k}\frac{p}{3} \sum_{\sigma \in \mathbb{P}_k - \{I^{\otimes k}\}} [\sigma].
\end{eqnarray}
\item A layer of $T$ gate is subject to two types of noise: $k$-qubit depolarizing noise with depolarizing rate $p/2$, modeled as
\begin{eqnarray}
\mathcal{N}_d = \left(1 - \frac{4^k-1}{4^k}\frac{p}{2}\right) + \frac{1}{4^k - 1}\frac{p}{2} \sum_{\sigma \in \mathbb{P}_k - \{I^{\otimes k}\}} [\sigma],
\end{eqnarray}
and coherent over-rotation noise, represented as a rotation by angle $\sqrt{2p}$ around the $Z^{\otimes k}$ axis, 
\begin{eqnarray}
\mathcal{N}_c = [e^{-i\frac{\sqrt{2p}}{2}Z^{\otimes k}}].
\end{eqnarray}
\end{itemize}
Since logical Pauli gates are assumed to be error-free, errors in stabilizer operations as well as encoding and decoding operations can be perfectly twirled into Pauli errors, justifying the Pauli error assumption. Therefore, only the $T$ gate introduces non-Pauli errors. 

In SNI, we estimate the total error rate and modify the computational circuit using instances of spacetime errors generated from error sampler circuits. In conventional PEC, we mitigate errors by using an uncorrelated Pauli error model. Specifically, for each $q$-qubit operation, the associated noise map is assumed to be Pauli noise acting non-trivially only on those $q$ qubits. As a result, each noise map requires estimating $4^q - 1$ parameters. Additionally, the model assumes that all operations of the same type, e.g.,~all $T$ gates, share the same noise map. The parameters are estimated using the same spacetime error instances as in SNI. Because the uncorrelated Pauli error model does not account for all the spatial correlations, conventional PEC exhibits a notable bias.

\subsection{Mitigating temporally correlated errors}
\label{app:time}

\begin{figure}[htbp]
\centering
\includegraphics[width=0.5\linewidth]{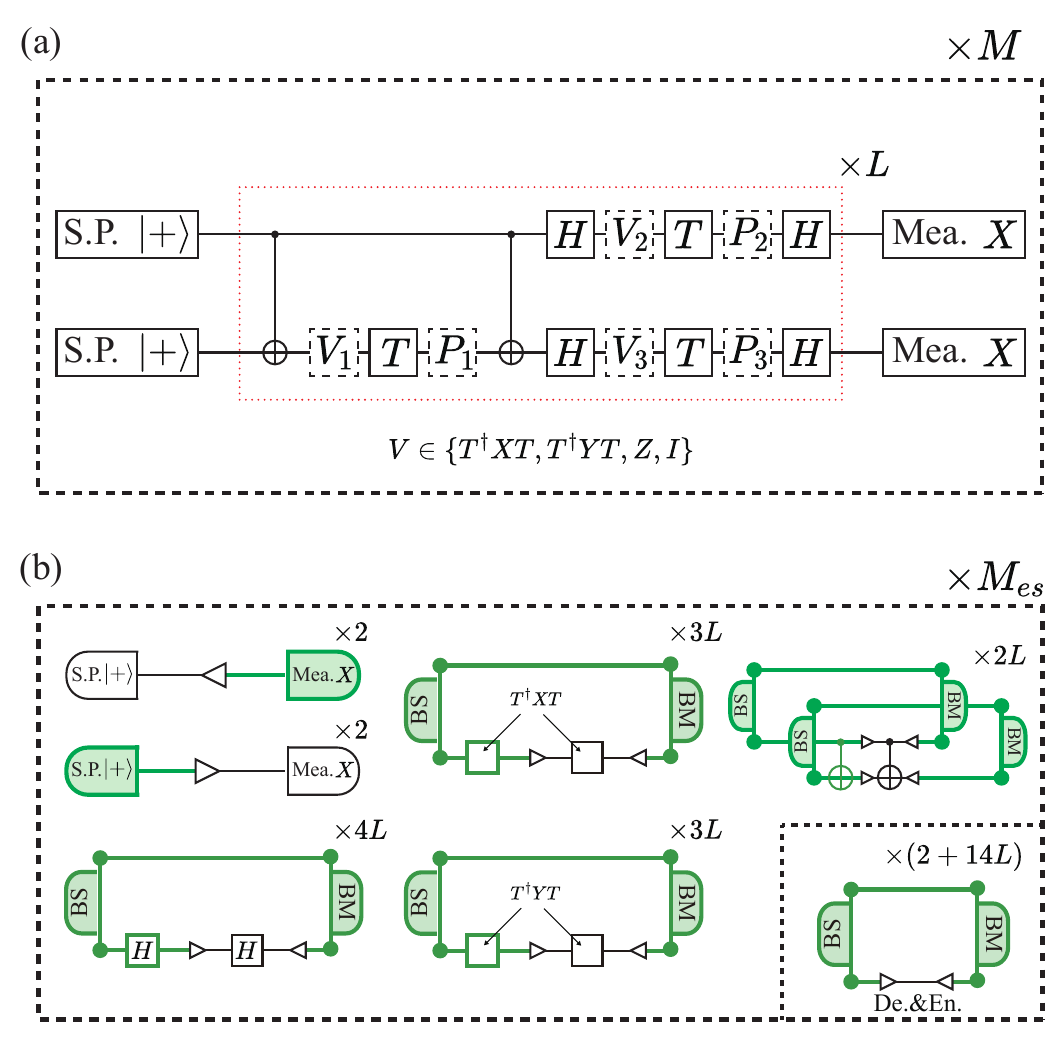}
\caption{
Circuits used in the simulation for temporally correlated errors. 
(a) Computational circuit. Gates denoted by dashed boxes implement twirling operations on $T$ gates; see Table~\ref{tab:twirled_operations}. Here, $P_i$ denotes a Pauli gate, and $V_i$ denotes the corresponding Clifford gate. 
(b) Error sampler circuits. 
}
\label{fig:circuits}
\end{figure}

\begin{figure}[htbp]
\centering
\includegraphics[width=0.5\linewidth]{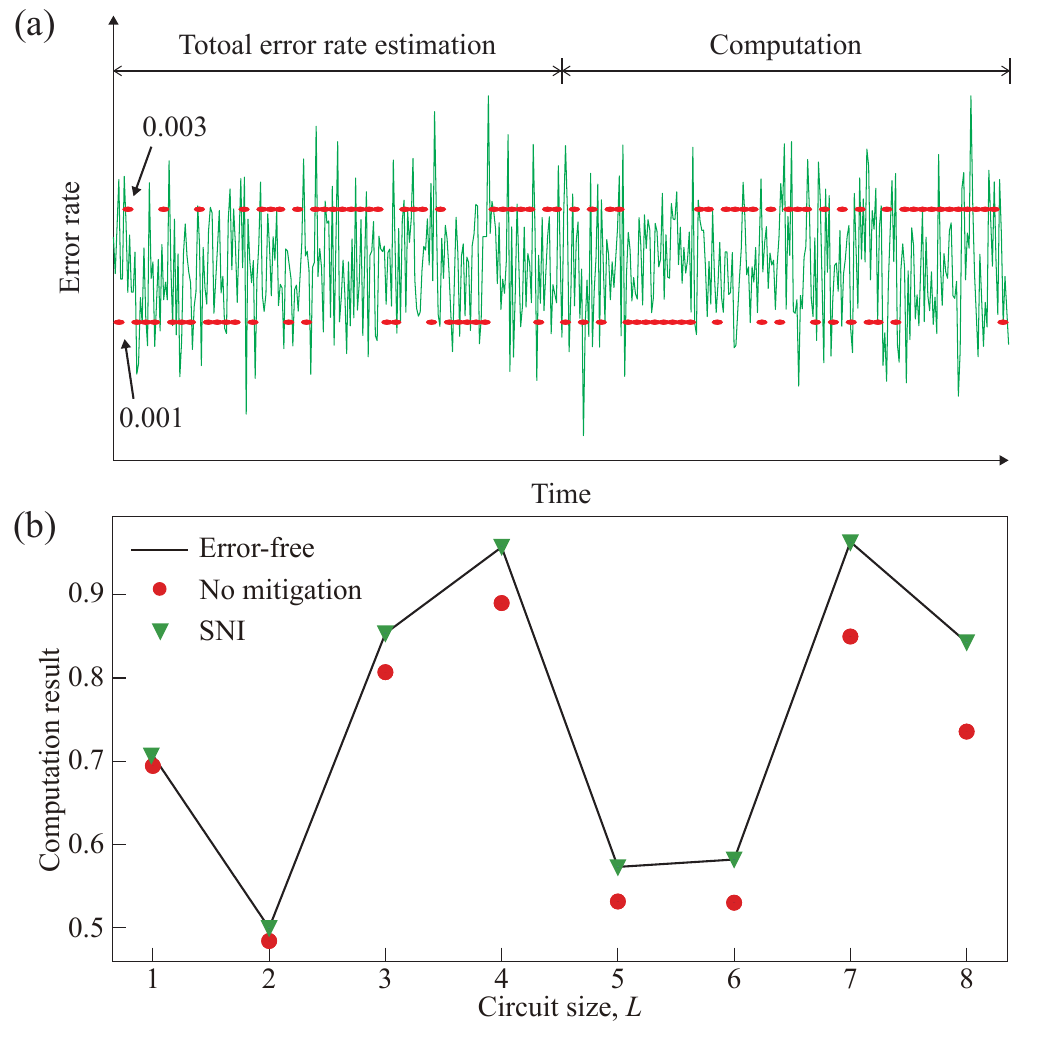}
\caption{
(a) Fluctuation of error rate $p$. 
(b) Computation results under error parameter fluctuations. We use the circuit in Fig.~\ref{fig:circuits}(a) to implement the transformation $[e^{-i\frac{\pi}{8}(X_1+X_2)}e^{-i\frac{\pi}{8}Z_1Z_2}]^L$ on two qubits initialized in the $\ket{+}$ state, with the observable being $X$ on the first qubit. SNI with a practical error sampler is applied for error mitigation, taking $M_P = 2.56 \times 10^8$ and $M = 2.56 \times 10^7$. 
}
\label{fig:time_app}
\end{figure}

Below is a detailed description of the error model used in our simulations. We simulate a two-qubit computational circuit alongside a set of error sampler circuits (see Fig.~\ref{fig:circuits}), where each operation is affected by noise as specified below:
\begin{itemize}
\item Logical Pauli gates and operations on super qubits are assumed to be error-free.
\item Non-Pauli $q$-qubit stabilizer operations, including CNOT gate, Hadamard gate, state preparation, measurement and the twirling gates $V$, are subject to $q$-qubit depolarizing noise with depolarizing rate $p$. This noise map is defined as
\begin{eqnarray}
\mathcal{N}_q = (1 - \frac{4^q-1}{4^q}p)[I^{\otimes q}] + \frac{1}{4^q}p \sum_{\sigma \in \mathbb{P}_q - \{I^{\otimes q}\}} [\sigma],
\end{eqnarray}
and is applied after state preparation and gates, and before measurement. 
\item Encoding and decoding operations are subject to single-qubit depolarizing noise with a reduced depolarizing rate of $p/3$. The noise map is given by
\begin{eqnarray}
\mathcal{N}_{en/de} = \left(1 -  \frac{p}{4}\right)[I] + \frac{p}{12}([X] + [Y] + [Z]).
\end{eqnarray}
\item The $T$ gate is subject to two types of noise: Single-qubit depolarizing noise with depolarizing rate $p/2$, modeled as
\begin{eqnarray}
\mathcal{N}_d = \left(1 - \frac{3p}{8}\right)[I] + \frac{p}{8}([X] + [Y] + [Z]),
\end{eqnarray}
and coherent over-rotation noise, represented as a rotation by angle $\sqrt{2p}$ around the $Z$ axis, 
\begin{eqnarray}
\mathcal{N}_c = [e^{-i\frac{\sqrt{2p}}{2}Z}].
\end{eqnarray}
\end{itemize}

Temporal correlations are caused by fluctuations in the error rate $p$; see Fig.~\ref{fig:time_app}(a). In practice, $p$ may vary randomly but continuously over time. To simplify the simulation, we instead assume $p$ is randomly drawn from the discrete set $\{0.001,0.003\}$ with equal probability. In each run of the computational circuit [see Fig.~\ref{fig:circuits}(a)], we randomly select $p$ from the discrete set, and use the selected $p$ throughout the entire circuit run (i.e.~all operations in the circuit run have the same value of $p$). Each spacetime error sample is generated using a set of error sampler circuits. For each run of such a circuit set [see Fig.~\ref{fig:circuits}(b)], we again randomly generate a value of $p$ from the discrete set, which is applied consistently across all circuits in the set (reflecting that these circuits are executed close in time). The error-mitigated computation requires $M$ runs of the computational circuit and $M_{es}$ runs of spacetime error generation, yielding a total of $M + M_{es}$ instances of $p$. The value of $p$ used for generating error samples is independent from that used in the computational circuit. In this way, we emulate temporally varying noise across different runs. 

Fig.~3 in the main text and Fig.~\ref{fig:time_app}(b) present the results of error-mitigated computation using SNI and conventional PEC under these temporally correlated errors. In conventional PEC, we use the same uncorrelated Pauli error model as in Sec.~\ref{app:space}. Because this model does not account for temporal correlations, conventional PEC exhibits a notable bias under temporally varying noise. 

%In conventional PEC, we mitigate errors by using an uncorrelated Pauli error model. Specifically, for each $k$-qubit operation, the associated noise map is assumed to be Pauli noise acting non-trivially only on those $k$ qubits. As a result, each noise map requires estimating $4^k - 1$ parameters. Additionally, the model assumes that all operations of the same type, e.g., all $T$ gates, share the same noise map. However, because this model does not account for temporal correlations, conventional PEC exhibits a notable bias under temporally varying noise. 

\section{Comparisons to other quantum error mitigation methods}
\label{app:comparisons}
A broad family of quantum error mitigation methods has been developed~\cite{cai_quantum_2023}. While practical applications may ultimately involve a combination of techniques, the core methods can be classified into two categories: exact methods and heuristic methods. 

Exact methods, including zero-noise extrapolation (ZNE) and PEC, produce results with provably bounded bias~\cite{temme_error_2017,li_efficient_2017}. Given an accurate noise model and sufficient overhead~\cite{endo_practical_2018,nielsen_gate_2021,van_den_berg_probabilistic_2023}, these methods can systematically reduce the bias to zero. In contrast, heuristic methods, such as virtual distillation~\cite{PhysRevX.11.041036} and Clifford regression~\cite{error_czarnik_2021}, lack such guarantees, although they avoid the burden of noise benchmarking. For example, Clifford regression assumes a global depolarizing noise model. Under realistic noise, these methods could incur a non-negligible residual bias. 

Our method is a variant of PEC, and thus belongs to the class of exact methods. Below, we compare it to the other exact method, ZNE. 

There are many variants of ZNE, differing primarily in how they amplify noise~\cite{kim_evidence_2023,scalable_kim_2023,scalable_le_2023,digital_giurgica-tiron_2020,error_kandala_2019,wahl2023zne}. A notable variant involves boosting the noise level by simulating additional noise using randomized gate sequences~\cite{kim_evidence_2023}. Given an accurate noise model, the simulated noise can faithfully reproduce the physical noise, enabling extrapolation to an unbiased result as the extrapolation order increases. Consequently, benchmarking the noise model becomes an essential step in this method, much like in PEC. Therefore, the key advantages of our method over conventional PEC, namely, the simplified characterization process (only a single parameter needs to be estimated) and robustness to correlated noise, also apply to conventional ZNE with simulated noise. 

A popular variant of ZNE is gate folding~\cite{digital_giurgica-tiron_2020}, where each gate $U$ is replaced by a gate sequence such as $UU^\dagger U$ to amplify the noise without requiring explicit knowledge of the noise model. While this approach is characterization-free, it does not guarantee faithful noise amplification under general, realistic noise models. Consequently, unbiased error mitigation is not assured. Similar limitations apply to other variants, such as gate stretching~\cite{error_kandala_2019,scalable_kim_2023} and code-distance scaling~\cite{wahl2023zne}, which rely on assumptions about the noise dependence on gate duration and the homogeneity of the qubit array, respectively. 

Finally, we emphasize that the error sampling techniques developed in this work are broadly applicable. For instance, they can be incorporated into ZNE by inserting sampled errors into the computational circuit to amplify noise in a controlled and faithful manner. In this way, our framework provides general-purpose tools for realizing exact error mitigation in the fault-tolerant era.

\section{Applications to surface codes}
\label{app:applications_SC}
In surface codes, each code block encodes a single logical qubit. Considering a two-dimensional qubit array, spatially correlated errors between logical qubits can arise from crosstalk between neighboring physical qubits. However, since cross-block correlations only occur at the boundaries, their influence on logical errors is negligible; see Sec.~\ref{app:surface_correlation} for numerical results on logical error correlations. Consequently, each logical operation can be benchmarked individually: For an operation acting on $q$ logical qubits, it suffices to run the error sampler circuit on those $q$ logical qubits using $2q$ super qubits. The primary benefit of applying SNI to surface codes is the minimized number of error parameters that need to be characterized, along with improved resilience to temporally correlated errors rather than that to specially correlated errors. 

Regarding the required code distance for super qubits, a large distance ratio $d_S/d$ is {\it not} required. It can be estimated using the following empirical formula for surface codes~\cite{fowler_surface_2012}, 
\begin{eqnarray}
p_L = p_0 (p/p_{th})^{\lceil d/2 \rceil},
\end{eqnarray}
where $d$ is the code distance, $p_L$ is the logical error rate per parity-check measurement cycle per logical qubit, $p_0 = 0.03$, $p$ is the physical error rate, and $p_{th} = 0.01$ is the threshold error rate. Suppose that the physical error rate is $p = 0.001$ and the code distance of logical qubits is $d = 19$, achieving a logical error rate of $p_L = 3\times10^{-12}$. If we take a code distance of $d_S = 25$ for super qubits, their error rate is smaller than $d = 19$ logical qubits by a factor of $10^{-3}$, i.e.~super-qubit errors are negligible. Since $(d_S/d)^2 \approx (1.3)^2 \approx 1.7$, we only need qubits fewer than two $d = 19$ blocks to encode a super qubit. 

The empirical formula also suggests that temporal correlations can become particularly severe when physical error rates fluctuate for each circuit run. For instance, consider a surface code with distance $d = 19$: If the physical error rate increases from $p$ to $p + 0.1p$, the logical error rate rises from $p_L$ to $1.1^{10}p_L \approx 2.6p_L$. This example illustrates that even small fluctuations in physical error rates can lead to substantial variations in logical error rates. 

SNI can mitigate errors with temporal correlations, as discussed in the main text. If we can sample spacetime errors with the same distribution as in the computation circuits, errors can be mitigated regardless of their correlations. To accurately sample the errors, it is crucial to run the error sampler circuits in close temporal proximity during the generation of a single instance of spacetime error. If the error rates vary significantly during the generation of an instance, the distribution of spacetime errors will be distorted relative to the computation circuit. 
\section{Applications to qLDPC codes}
\label{app:applications_qLDPC}
In qLDPC codes that encode multiple logical qubits within each code block, logical errors are correlated, even in the absence of crosstalk between physical qubits. A logical error involves errors occurring on at least $d$ physical qubits, where $d$ is the code distance. On each physical qubit, the error can affect a logical qubit if the physical qubit lies within the support of its logical operators. If a physical qubit is in the supports for a number of logical operators, it could simultaneously induce errors on all them, leading to many-logical-qubit correlations. On average, at least $kd/n$ logical operators overlap on each physical qubit, where $k$ is the number of logical qubits encoded in $n$ physical qubits, each logical qubit has an $X$ ($Z$) logical operator, and the support of each logical operator has a size of at least $d$. Therefore, the correlation can become particularly severe when the encoding rate $k/n$ is constant and the code distance is large. A numerical simulation in Ref.~\cite{zhang_demonstrating_2025} illustrates such correlations. As a result of these correlations, each code block must be benchmarked as a whole: Sampling errors for a single-block operation requires $2k$ super qubits, while a two-block operation requires $4k$ super qubits. 

One approach to employing qLDPC codes in fault-tolerant quantum computing is through concatenation codes~\cite{gottesman_fault-tolerant_2014,tamiya_polylog-time-_2024,nguyen_quantum_2024}. This enables universal quantum computation with a constant qubit overhead, which is a key advantage of qLDPC codes with constant encoding rates over surface codes. In this setting, the entire quantum computer consists of multiple qLDPC code blocks. To implement logical gates, resource states are prepared using concatenated codes and consumed via gate teleportation. If these resource states are prepared independently, spatial correlations remain confined within each code block (or pairs of blocks for inter-block gates). As a result, each block (or pairs of blocks) can be benchmarked independently using error sampler circuits. However, for improved time efficiency, an optimized protocol prepares multiple identical resource states collectively~\cite{nguyen_quantum_2024}. This collective preparation may introduce additional inter-block correlations, requiring simultaneous benchmarking of all involved blocks. For example, consider preparing $M$ instances of a resource state used to implement the Hadamard gate on the first logical qubit of a block (noting that each block is acted upon by only one gate at a time, following the protocol in Ref.~\cite{nguyen_quantum_2024}). To sample errors, we initialize each block in a state tailored for this purpose: The first logical qubit is prepared for Hadamard-gate error sampling, while all other logical qubits are set for identity-gate error sampling. This state is prepared using super qubits [comprising all operations up to and including decoding in Fig.~\ref{fig:protocol}(a)], and the same state is prepared across all $M$ blocks. The Hadamard gates are then applied simultaneously to all $M$ blocks, followed by measurement via super qubits. This procedure captures both intra-block and inter-block error correlations. Importantly, if memory errors are negligible, i.e.,~the errors remain unchanged over time, then it is not necessary to benchmark all $M$ blocks simultaneously; in such cases, error sampler circuits can be run independently for each block without affecting the validity of the error sampling. 

In addition to concatenation-based techniques, an alternative approach to quantum computation with qLDPC codes is based on performing Pauli operator measurements via code deformation~\cite{Cohen_2022}, known as lattice surgery. Recent advances demonstrate that multiple Pauli operators can be measured in parallel on a single code block while preserving constant qubit overhead~\cite{zhang_time-efficient_2025,cowtan2025parallellogicalmeasurementsquantum}. To sample errors under such parallelized measurement, we can proceed as follows. Let $\sigma_1, \sigma_2, \ldots, \sigma_q$ be the set of mutually commuting Pauli operators to be measured simultaneously. We can find a Clifford unitary $U$ such that $U \sigma_j U^\dagger = Z_j$ for all $j = 1, \ldots, q$, where $Z_j$ denotes the $Z$ operator acting on the $j$-th logical qubit. In the error sampler circuit, we initialize the code block with the first $q$ logical qubits prepared in states suitable for $Z$-measurement error sampling, while the remaining logical qubits are prepared for idle-gate error sampling. Prior to the decoding operation, we apply the Clifford gate $U$. This is followed by the parallel Pauli measurement, encoding, application of $U^\dagger$, and finally, measurement to read out the errors. When each of the $\sigma_j$ operators acts on one or two logical qubits without overlap, the Clifford unitary $U$ can be implemented using a single layer of one- and two-qubit Clifford gates. Moreover, lattice surgery is a general tool applicable to qLDPC codes, enabling state transfer between code blocks. This capability can be leveraged to perform encoding and decoding operations in our protocol: For instance, super qubits can be encoded into a block with a larger code distance for better protection, and then transferred to or from standard blocks via lattice surgery.

\end{widetext}

\bibliography{references}

\vspace{0.5cm}

\end{document}